\documentclass[11pt,a4paper]{article}

\usepackage[T1]{fontenc}
\usepackage[utf8]{inputenc}
\usepackage{amsmath,amssymb,amsthm,amsfonts}
\usepackage[margin=1in]{geometry} 
\usepackage{thmtools}
\usepackage{comment}

\usepackage{graphicx}
\usepackage{float} 

\usepackage[math]{blindtext}

\usepackage{graphicx}
\usepackage{dsfont}
\usepackage{units}
\usepackage{enumitem}
\usepackage{booktabs}

\usepackage{tikz}			
\usetikzlibrary{arrows,automata,shapes, trees}

\theoremstyle{plain}
\newtheorem{theorem}{Theorem}[section]
\newtheorem{proposition}[theorem]{Proposition}
\newtheorem{lemma}[theorem]{Lemma}
\newtheorem{corollary}[theorem]{Corollary}
\theoremstyle{definition}
\newtheorem{definition}[theorem]{Definition}
\newtheorem{remark}[theorem]{Remark}
\newtheorem{example}[theorem]{Example}

\theoremstyle{remark}
\numberwithin{equation}{section}
\newtheorem*{empty*}{}

\renewcommand\thmcontinues[1]{\textbf{Continued}}

\DeclareMathOperator*{\Var}{Var}

\DeclareMathOperator*{\esssup}{ess\,sup}

\renewcommand{\theta}{\vartheta}
\renewcommand{\epsilon}{\varepsilon}
\renewcommand{\P}{\mathbb{P}}
\newcommand{\Q}{\mathbb{Q}}
\newcommand{\E}{\mathbb{E}}

\newcommand{\NN}{\mathbb{N}}

\newcommand{\RR}{\mathbb{R}}

\newcommand{\cA}{\mathcal{A}}
\newcommand{\cB}{\mathcal{B}}

\newcommand{\cD}{\mathcal{D}}
\newcommand{\cE}{\mathcal{E}}
\newcommand{\cF}{\mathcal{F}}
\newcommand{\cG}{\mathcal{G}}

\newcommand{\cM}{\mathcal{M}}

\newcommand{\cP}{\mathcal{P}}
\newcommand{\cQ}{\mathcal{Q}}

\newcommand{\cX}{\mathcal{X}}

\newcommand{\diff}{\mathrm{d}}
\newcommand{\dd}{\,\mathrm{d}}

\newcommand{\1}{\mathbf{1}}
\newcommand{\0}{\mathbf{0}}

\newcommand*{\as}[1]{#1\text{-a.s.}}

\newcommand*{\ol}[1]{\bar{#1}}

\newcommand{\ES}{\mathrm{ES}}

\newcommand{\VaR}{\mathrm{VaR}}
\newcommand{\SR}{\mathrm{SR}}

\begin{document}
	
	\title{Mean-$\rho$ portfolio selection and $\rho$-arbitrage \\for coherent risk measures\footnote{We are grateful to John Armstrong and Ruodu Wang for fruitful discussions. We also thank two anonymous referees and an Associate Editor for their pertinent remarks, which have significantly improved the paper.}}
	
\author{Martin Herdegen\thanks{University of Warwick, Department of Statistics, Coventry, CV4 7AL, UK, email \texttt{m.herdegen@warwick.ac.uk}.}
	\and 
Nazem Khan\thanks{University of Warwick, Department of Statistics, Coventry, CV4 7AL, UK, email \texttt{nazem.khan@warwick.ac.uk}.}
	}

	\date{\today}
	
	\maketitle

\begin{abstract}
We revisit mean-risk portfolio selection in a one-period financial market where risk is quantified by a positively homogeneous risk measure $\rho$.  We first show that under mild assumptions, the set of optimal portfolios for a fixed return is nonempty and compact.  However, unlike in classical mean-variance portfolio selection, it can happen that no efficient portfolios exist.  We call this situation $\rho$-arbitrage, and prove that it cannot be excluded -- unless $\rho$ is as conservative as the worst-case risk measure.

After providing a primal characterisation of $\rho$-arbitrage, we focus our attention on coherent risk measures that admit a dual representation and give a necessary and sufficient dual characterisation of  $\rho$-arbitrage.  We show that the absence of $\rho$-arbitrage is intimately linked to the interplay between the set of equivalent martingale measures (EMMs) for the discounted risky assets and the set of absolutely continuous measures in the dual representation of $\rho$.  A special case of our result shows that the market does not admit $\rho$-arbitrage for Expected Shortfall at level $\alpha$ if and only if there exists an EMM $\mathbb{Q} \approx \mathbb{P}$ such that $\Vert \frac{\textnormal{d}\mathbb{Q}}{\textnormal{d}\mathbb{P}} \Vert_\infty < \frac{1}{\alpha}$.
\end{abstract}




\section{Introduction}
 It has been widely argued that the financial crisis of 2007-2009 was a result of excessive risk-taking by banks; see e.g.~\cite{erkens2012corporate, williams2010uncontrolled}. Consequently, the financial regulators have tried to impose better risk constraints on financial institutions, which for the banking sector are codified in the Basel accords. One of the key changes from Basel II to Basel III was updating the ‘official’ risk measure from Value at Risk (VaR) to Expected Shortfall (ES) in the hope of better financial regulation; cf.~the discussion in \cite{koch2016unexpected}.
 
While Basel III is probably a step in the right direction, the following questions must be asked:~Does an ES constraint really prevent banks from taking excessive risk?\footnote{For a thorough discussion of other “unexpected shortfalls” of ES we refer to \cite{koch2016unexpected}.}~And if not, are there alternative coherent risk measures that are superior? In this paper we address these questions by revisiting the problem of portfolio optimisation in a one-period financial market -- replacing the variance in the classical theory of Markowitz \cite{Markowitz1952} by a positively homogeneous risk measure $\rho$, with VaR and ES as key examples. We refer to this as mean-$\rho$ portfolio selection in the sequel.

Unlike classical mean-variance portfolio selection, mean-$\rho$ portfolio selection may be \emph{ill-posed} in the sense that there are no efficient portfolios, or even worse, that for each portfolio $\pi$, there is another portfolio $\pi'$ that has \emph{simultaneously} a higher expected return and a lower risk, measured by $\rho$. We refer to these situations as \emph{$\rho$-arbitrage} and \emph{strong $\rho$-arbitrage}, respectively. This terminology is motivated by the fact that $\rho$-arbitrage and strong $\rho$-arbitrage are generalisations of arbitrage of the first kind and arbitrage of the second kind, respectively.\footnote{The term $\rho$-arbitrage was recently coined by Armstrong and Brigo \cite{armstrong2019statistical}. In a more general framework, they say that a payoff $X$ is a $\rho$-arbitrage if $X$ has a non-positive price $\cP(X)$ and non-positive risk $\rho(X)$ but is positive with positive probability. In the special case that the pricing function is given by $\cP(X) = \mathbb{E}[\zeta X]$, where $\zeta$ is a state price density, one can show that a $\rho$-arbitrage in our sense is a $\rho$-arbitrage in their sense. For this reason, we may use the same name without abuse of notation.} 

The occurrence of $\rho$-arbitrage is puzzling at first sight because this is a phenomenon that does not appear in the classical mean-variance framework. The deeper reason for this difference is that the standard deviation (the variance does not have the right dimension) is not a monetary/translation-invariant risk measure but rather a \emph{deviation risk measure}. This latter class of risk measures  has been axiomatically studied by Rockafellar et al.~\cite{rockafellar2006deviation}. They showed in \cite{MasterFundsRockafellar} that if $\cD$ is a deviation risk measure, then mean-$\cD$ portfolio selection is always well posed.\footnote{In the special case that $\mathcal{D}$ is Expected Shortfall Deviation, mean-$\mathcal{D}$ portfolio selection has been studied by Tasche \cite{tasche1999risk} and Bertsimas et al.~\cite{bertsimas2004shortfall}.} 

\medskip

The goal of this paper is to study mean-$\rho$ portfolio selection under minimal assumptions on the (monetary) risk measure $\rho$ and the returns distributions of the risky assets, and to provide necessary and sufficient primal and dual characterisations for the absence of $\rho$-arbitrage.\footnote{The paper focuses on a one-period setting. The extension to a dynamic setting is a challenging question that is left for future research.} 

After describing our model in Section \ref{section:model}, we devote Section \ref{section:portfolio optimisation} to a rigorous study of mean-$\rho$ portfolio selection, which to the best of our knowledge has not been carried out at this level of generality in the literature before: We do not make any assumption on the returns distributions of the risky assets (apart from having first moments), we only require $\rho$ to be positively homogeneous (but not necessarily convex), and we allow $\rho$ to take the value $\infty$. In particular, our results can be applied to VaR, which has been excluded in most of the extant literature.\footnote{\cite{alexander2002economic} is a notable exception but there returns are restricted to be multivariate normal.}
We first study the simplified problem of finding so-called \emph{optimal portfolios}, which have minimal risk given a fixed (excess) return. We show in Theorem \ref{thm:Existence of optimal portfolios} that under mild regularity assumptions on the risk measure, \emph{positive homogeneity alone} (without convexity) is enough to ensure existence of optimal portfolios. We then introduce the notions of $\rho$-arbitrage and strong $\rho$-arbitrage and provide necessary and sufficient primal conditions for the absence of (strong) $\rho$-arbitrage. We explain how $\rho$-arbitrage and strong $\rho$-arbitrage generalise the classical notions of arbitrage of the first and second kind, respectively, and show in Theorem \ref{thm:cannot avoid reg arb} that  (strong) $\rho$-arbitrage cannot be excluded (under standard no-arbitrage) unless $\rho$ is as conservative as the worst-case risk measure. We conclude the primal part of the paper by illustrating our results in the special case of elliptical returns distributions.

In Section \ref{section:dual characterisations}, we restrict $\rho$ to be a \emph{coherent} risk measure that admits a dual representation $\rho(X) = \sup_{Z \in \mathcal{Q}} \mathbb{E}[-ZX]$, where $\cQ$ describes some dual set of probability measures that are absolutely continuous with respect to $\P$.  We then introduce necessary and sufficient conditions for the absence of (strong) $\rho$-arbitrage and show that these conditions are indeed minimal by providing relevant counterexamples that are collected in Appendix \ref{app:counterexamples}. Our main result is Theorem~\ref{thm: no reg arb equivalence} which shows that absence of $\rho$-arbitrage is equivalent to $\cP \cap \tilde \cQ \neq \emptyset$, where $\cP$ describes the set of all equivalent martingale measures (EMMs) for the discounted risky assets and $\tilde \cQ$ is the “interior” of $\cQ$. The precise definition for this “interior” of $\cQ$ is very delicate because both topological and algebraic notions fail. For this reason, we define $\tilde \cQ$ in an abstract way that also gives some additional flexibility. This is worth the effort: As a by-product of our main result, we get a refined version of the fundamental theorem of asset pricing in a one-period market: for returns in $L^1$, we show in Theorem \ref{thm:FTAP} that standard no-arbitrage is equivalent to the existence of an EMM $\Q$ whose Radon-Nikod\'{y}m derivative is uniformly bounded \emph{away from $0$}.

We proceed to apply our dual results to a large variety of examples in Section \ref{sec:Applications}. These examples also highlight an important technical feature of our analysis. In order to achieve a maximum level of generality, we do \emph{not} assume that the set $\cQ$ in the dual representation of $\rho$ (which is not unique) is $L^1$-closed (as for example assumed by \cite{cherny2008pricing}), nor do we assume that it coincides with the maximal dual set. This extra flexibility allows us to get dual characterisations of $\rho$-arbitrage even in cases when $\rho$ might take the value $\infty$ and to explicitly characterise the “interior” set  $\tilde \cQ$ for a large class of examples.

\subsection{Related Works}
Mean-$\rho$ portfolio selection for specific (classes of) risk measures has been well studied in the extant literature.  Alexander and Baptista \cite{alexander2002economic} solved the problem of mean-VaR portfolio selection explicitly for multivariate normal returns distributions. Rockafellar and Uryasev \cite{RockafellarOptimisation} studied mean-ES portfolio selection for continuous returns distributions and showed that the optimisation problem could be reduced to linear programming. Subsequently, the results of \cite{RockafellarOptimisation} were extended to general returns distributions by the same authors  \cite{rockafellar2002conditional} and later generalised to spectral risk measures by Adam et al.~\cite{adam2008spectral}. 

The occurrence of $\rho$-arbitrage was first recognised for VaR by \cite{alexander2002economic} who gave necessary and sufficient conditions for its absence  in the case of multivariate normal returns distributions. For ES, the possibility of $\rho$-arbitrage was first noted in a working paper by De Giorgi \cite{de2002note} in the case of elliptical returns distributions and later observed in a simulation study by Kondor et al.~\cite{kondor2007noise}.\footnote{The working paper \cite{rockafellar2002deviation} also recognised the occurrence of $\rho$-arbitrage for a wide class of coherent risk measures, noting that minimising the risk subject to an inequality constraint on the expected return may fail to have a solution. They called this phenomenon an ``acceptably free lunch''.}  The latter paper led to a more detailed study by Ciliberti et al.~\cite{ciliberti2007feasibility}, who concluded that there is a \emph{phase transition}, i.e., for small values of $\alpha$, mean-ES portfolio selection is well-posed, and from a certain critical value $\alpha^*$ onwards, mean-ES portfolio selection becomes ill-posed. For example, if returns are multivariate normal distributed and the maximal Sharpe ratio that can be attained in the market is $2.5$, then $\alpha^* \approx 0.016$ -- which is below the $2.5\%$ of the Basel III accord; cf.~Example \ref{exa:jointly gaussian returns ES and VaR} for details. More recently, Armstrong and Brigo \cite{armstrong2019risk} showed that VaR and ES constraints may be void for behavioural investors with an $S$-shaped utility. They proceeded to study $\rho$-arbitrage for general coherent risk measures in \cite{armstrong2019statistical}, focusing on multivariate normal returns and looking at the issue from an empirical/statistical perspective, demonstrating that this phenomenon is relevant in practise.

\emph{Strong} $\rho$-arbitrage is closely related to the notion of a \emph{good deal}. For coherent risk measures, good deals were first studied by Jaschke and Küchler \cite{jaschke2001coherent}. At the end, they raise the question: “What parts of Markovitz' theory carry over, what is different in $(\mu, \rho)$-optimisation?” (p. 199). For \emph{expectation bounded} risk measures (which excludes VaR), the two concepts are equivalent if we define a \emph{good deal} as Cherny \cite{cherny2008pricing}, who gives a dual characterisation of \emph{no good deals}. However, Cherny's result only applies to \emph{strong} $\rho$-arbitrage, whereas the main focus of our paper is $\rho$-arbitrage, which is the more important of the two concepts.\footnote{To understand the difference between the two concepts, it is insightful to recall that $\rho$-arbitrage is a generalisation of arbitrage of the first kind whereas strong $\rho$-arbitrage is a generalisation of arbitrage of the second kind. As a concept, arbitrage of the first kind is by far more important than arbitrage of the second kind because (under the existence of a numéraire) absence of the former implies absence of the latter (but not vice versa) and the fundamental theorem of asset pricing provides a necessary and sufficient dual characterisation of arbitrage of the first kind, whereas a necessary and sufficient dual characterisation of arbitrage of the second kind does not exist -- unless the probability space is finite. The same relationship holds between $\rho$-arbitrage and strong $\rho$-arbitrage: The absence of the former implies the absence of latter (but not vice versa), and there are situations where $\rho$-arbitrage, but not strong $\rho$-arbitrage, admits a dual characterisation; an important example is when the returns lie in some Orlicz space $L^\Phi$ but not in the corresponding Orlicz heart $H^\Phi$; see Remark \ref{remark:delta 2} for details.}

Another strand of literature that is not directly related to mean-$\rho$ portfolio selection but is close from a conceptual point of view is mean-variance portfolio selection under \emph{ambiguity aversion}. Here, the idea is that the investor is uncertain about the probabilistic model but otherwise stays in the classical mean-variance framework. Let us just mention two key contributions: Boyle et al.~\cite{boyle:al:12} assume that the investor is uncertain about the mean (but not the variance) of the risky assets and hence first minimises over the expected returns they consider plausible. If the investor has less uncertainty about the returns of some “familiar” assets, they hold -- compared to classical mean-variance portfolio selection -- a higher proportion of “familiar” assets and a lower proportion of “unfamiliar” ones, where they have more uncertainty about the returns. 
Maccheroni et al.~\cite{maccheroni:al:13} consider the Bayesian framework of model uncertainty from \cite{klibanoff:al:05}, where the agent has a prior on plausible models and penalises the mean-variance utility under the so-called ambiguity neutral model by a variance term describing the model uncertainty. In a setting with a riskless and two risky assets (one with and one without model uncertainty), they show that the \emph{alpha} of the ambiguous asset is the key additional statistic in this problem. Both mean-$\rho$ portfolio selection and mean-variance portfolio selection under ambiguity aversion can be seen as a way of making the classical Markowitz problem more robust. In the former, the focus is on making the risk measure more robust (and correcting for the theoretical shortcomings of the variance as a measure of risk). In the latter, the focus is on making the probabilistic model more robust by taking uncertainty on the mean/the probabilistic model into account. Both extensions are important but address different issues. It is an interesting direction for future research to combine both extensions.

Last but not least, $\rho$-arbitrage is conceptually related to the notion of \emph{regulatory arbitrage}. Here the idea is that the risk measure constraint can be interpreted as a regulatory capital requirement imposed by the regulator. If the agent can act in some way to avoid (or weaken) the regulatory constraint, they perform a regulatory arbitrage.\footnote{The term ‘regulatory-arbitrage’ has been emphasised in the literature more intensively since 2004 as explained in \cite{willesson2017and}.  However, there is no universal definition for this concept. The general consensus is that it is a notion that refers to actions performed by financial institutions to avoid unfavourable regulation.} The closest paper to ours in that direction is Wang \cite{wang:16}, who defines regulatory arbitrage quantitatively as the level of \emph{superadditivity} that a risk measure possesses. The larger the latter, the more the agent can weaken the regulatory constraint by splitting up their position. While this definition is somewhat different from our notion of $\rho$-arbitrage, it captures the same idea that risk measure constraints may be (partially) avoided by financial agents. In our case, for certain market environments (“too high Sharpe ratio”), the regulatory constraint becomes void in portfolio optimisation, whereas in \cite{wang:16}, the action of splitting up the position can weaken the regulatory constraint. It is an interesting direction for future research to develop an axiomatic framework of properties that a risk measure should possess to eliminate such situations/actions as far as possible.

\section{Model}
\label{section:model}

We consider a one-period $(1+d)$-dimensional market $(S^0_t, \ldots, S^d_t)_{t \in \{0, 1\}}$ on some probability space $(\Omega, \cF, \P)$.  We assume that $S^0$ is riskless and satisfies
$S^0_0 = 1$ and $S^0_1 = 1 + r$, where $r > -1$ denotes the riskless rate. We further assume that $S^1, \ldots, S^d$ are risky assets, where $S^1_0, \ldots, S^d_0 > 0$ and $S^1_1, \ldots, S^d_1$ are real-valued $\cF$-measurable random variables. We denote the (relative) \emph{return} of asset $i \in \{0, \ldots, d\}$ by 
\begin{equation*}
R^i := \frac{S^i_1 - S^i_0}{S^i_0},
\end{equation*}
and set $S:= (S^1, \ldots, S^d)$ and $R:= (R^1, \ldots, R^d)$ for notational convenience.

We may assume without loss of generality that the market is \emph{nonredundant} in the sense that $\sum_{i =0}^d \theta^i S^i =0$ $\as{\P}$ implies that $\theta^i = 0$ for all $i \in \{0, \ldots, d\}$. We also impose that the risky returns $R^1, \ldots, R^d$ are $\mathbb{P}$-integrable, which is a minimum requirement for mean-$\rho$ portfolio selection.  Thus, each asset $i \in \{0, \ldots,d\}$ has a finite expected return $\mu^i := \mathbb{E}[R^i]$, and we set $\mu:= (\mu^1, \ldots, \mu^d) \in \RR^d$.  Finally, we assume that the risky returns are \emph{nondegenerate} in the sense that for at least one $i \in \{1, \ldots, d\}$, $\mu^i \neq r$.\footnote{If $\mu^i = r$ for all $i \in \{1, \ldots, d\}$, then every portfolio $\pi \in \RR^d$ has zero expected excess return.  There would be no incentive to invest and mean-risk portfolio optimisation becomes meaningless.}  Note that this implies that $\P$ itself is not an equivalent martingale measure for the discounted risky assets $S/S^0$.

\subsection{Portfolios}

As $S^0_0, \ldots, S^d_0 > 0$, we can parametrise trading in \emph{fractions of wealth}, and we assume that trading is frictionless. More precisely, we fix an initial wealth $x_0 > 0$ and describe any portfolio (for this initial wealth) by a vector $\pi =(\pi^1, \ldots, \pi^d) \in \RR^d$, where $\pi^i$ denotes the fraction of wealth invested in asset $i \in \{1, \ldots, d\}$. The fraction of wealth invested in the riskless asset is in turn given by $\pi^0 := 1 - \sum_{i =1}^d \pi^i = 1 - \pi \cdot \1$, where $\1 :=(1, \ldots, 1) \in \RR^d$.  The \emph{return} of a portfolio $\pi \in \RR^d$ can be computed by\footnote{Note that the (relative) return of a portfolio does not depend on the initial wealth $x_0$.}
\begin{equation*}
R_\pi := (1 - \pi \cdot \1)r + \pi \cdot R,
\end{equation*}
and the \emph{excess return} of a portfolio $\pi \in \RR^d$ over the riskless rate $r$ is in turn given by
\begin{equation}
\label{eq:excess return}
X_{\pi} := R_\pi - r = (1 - \pi \cdot \1)r + \pi \cdot R - r =\pi \cdot(R -r\1).
\end{equation}
It follows that $\cX=\{X_\pi:\pi \in \RR^d\}$ is a subspace of $L^1$.  The \emph{expected excess return} of a portfolio $\pi \in \RR^d$ over the riskless rate $r$ can be calculated as
\begin{equation*}
\mathbb{E}[X_{\pi}] = \pi \cdot (\mu - r \1).
\end{equation*}
For fixed $\nu \in \mathbb{R}$, we set 
\begin{equation}
\label{eq:def:Pi nu}
\Pi_{\nu} := \{\pi \in \RR^d: \mathbb{E}[X_{\pi}] = \nu\},
\end{equation}
i.e., $\Pi_{\nu}$ denotes the set of all portfolios with expected excess return $\nu$. By nondegeneracy, $\Pi_{\nu} \neq \emptyset$ for all $\nu \in \RR$.\footnote{Indeed, $\frac{\nu (\mu - r \1)}{(\mu - r \1)\cdot (\mu - r \1)} \in \Pi_\nu$ for each $\nu \in \RR$.}  Moreover, it is easy to check that $\Pi_{\nu}$ is closed and convex for each $\nu \in \RR$. Finally, the definition of $\Pi_{\nu}$ in \eqref{eq:def:Pi nu} implies that
\begin{equation}
\label{eq:Pi k}
\Pi_{k} = 
\begin{cases}
k\Pi_{1} := \{k\pi : \pi \in \Pi_{1} \}, &\text{if } k > 0, \\
(-k)\Pi_{-1} := \{-k\pi : \pi \in \Pi_{-1} \}, &\text{if } k < 0.
\end{cases}
\end{equation}
In the following, we will only focus on nonnegative excess returns.

\section{Mean-$\rho$ portfolio selection and $\rho$-arbitrage}
\label{section:portfolio optimisation}

As $\Pi_{\nu} \neq \emptyset$ for all $\nu \geq 0$, it is clear that in order to study portfolio selection, some kind of risk constraint has to be imposed. In the classical mean-variance approach pioneered by Markowitz~\cite{Markowitz1952}, risk is measured by variance.  Here we adapt the axiomatic approach of Artzner et al.~\cite{artzner1999coherent}.  Assume that $L^\infty \subset L \subset L^1$ is a Riesz space that contains $\cX=\{X_\pi : \pi \in \RR^d\}$.  We focus on a positively homogeneous monetary measure of risk $\rho: L \xrightarrow{} (-\infty, \infty]$, which satisfies the axioms:
\begin{itemize}
	\item \emph{Monotonicity}: For any $X_1, X_2 \in L$ such that $X_1 \leq X_2 \ \mathbb{P}$-a.s., $\rho(X_{1}) \geq \rho(X_{2})$.
	\item \emph{Cash invariance:}  If $X \in L$ and $c \in \mathbb{R}$, then $\rho(X+c) = \rho(X)-c$.
	\item \emph{Positive homogeneity:} For all $X \in L$ and $\lambda \geq 0$, $\rho(\lambda X) = \lambda \rho(X)$.
	\end{itemize}

A couple of remarks are in order.
\begin{remark}
(a) The Riesz space $L$ can be seen as an ambient space of $\cX$.  Key examples for $L$ include $L^p$-spaces, for $p \in [1,\infty]$, or more generally Orlicz spaces (cf.~Appendix \ref{app:subsection:primer on orlicz spaces}).

(b) In some situations, it is useful to allow $\rho$ to take the value $\infty$. For example, if all the returns $R^i$ are bounded from below but unbounded from above and only in $L^1$ (so that $L = L^1$), it makes perfect sense to consider for $\rho$ the worst-case risk measure $\mathrm{WC}$ (cf.~Definition \ref{defn:VaR and ES} below). Then $\mathrm{WC}(X_\pi)$ is finite if $\pi^i \geq 0$ for all $i \in \{1, \ldots,d \}$  but it may take the value $\infty$ if $\pi^i < 0$ for some $i \in \{1, \ldots, d\}$.

(c) Note that \emph{none} of the results in Section \ref{section:portfolio optimisation} makes use of monotonicity of $\rho$.\footnote{Note, however, that the results of Section \ref{section:dual characterisations} and  \ref{sec:Applications} do use monotonicity.} Notwithstanding, we have included it since monotonicity is perhaps the most natural property that a risk measure should possess.
\end{remark}

We recall the definition of the three most prominent examples of risk measures satisfying the above axioms, Value at Risk, Expected Shortfall and the worst-case risk measure.  More examples of risk measures are given in Section \ref{sec:Applications}.  

\begin{definition}
\label{defn:VaR and ES}
Let $X \in L^1$.
\begin{itemize}
\item The \emph{Value at Risk} (VaR) of $X$ at confidence level $\alpha \in (0, 1)$ is given by
\begin{equation*}
\textnormal{VaR}^{\alpha}(X) := \inf \{m \in \mathbb{R} :  \mathbb{P}[m+X < 0] \leq \alpha \}.
\end{equation*}
\item The \emph{Expected Shortfall} (ES) of $X$ at confidence level $\alpha \in (0, 1)$ is given by
\begin{equation*}
\textnormal{ES}^{\alpha}(X) := \frac{1}{\alpha} \int_{0}^{\alpha} \textnormal{VaR}^{u}(X) \dd u.
\end{equation*}
\item The \emph{worst-case risk} of $X$ is given by $\textnormal{WC}(X):=\esssup(-X)$.\footnote{This can be interpreted as the Expected Shortfall at confidence level $\alpha  = 0$.}
\end{itemize}
\end{definition}

We start our discussion on mean-$\rho$ portfolio selection by introducing a partial preference order on the set of portfolios. This preference order formalises the idea that return is “desirable” and risk is “undesirable”.

\begin{definition}
	\label{def:rho preferred}
A portfolio $\pi \in \RR^d$ is \emph{strictly $\rho$-preferred} over another portfolio $\pi^\prime \in \RR^d$ if $\mathbb{E}[X_{\pi}] \geq \mathbb{E}[X_{\pi^\prime}]$ and $\rho(X_{\pi}) \leq \rho(X_{\pi^\prime})$, with at least one inequality being strict.
\end{definition}

\begin{remark}
Linearity of the expectation and cash-invariance of $\rho$ imply that a portfolio $\pi \in \RR^d$ is strictly $\rho$-preferred over another portfolio $\pi^\prime \in \RR^d$ if and only if $\mathbb{E}[R_\pi] \geq \mathbb{E}[R_{\pi^\prime}]$ and $\rho(R_\pi) \leq \rho(R_{\pi^\prime})$, with at least one inequality being strict. This equivalent formulation might seem more natural from an economic perspective. However, it turns out that working with excess returns is mathematically more convenient.\footnote{One might also wonder why we apply $\rho$ to the (relative) return of a portfolio rather than the absolute return $x_0 R_\pi$ or the final cash value $x_0 (R_\pi + 1)$. As $\rho$ is cash-invariant and positive homogeneous, this does not matter since in each case we get exactly the same preference order as in Definition \ref{def:rho preferred}. We note in passing that if we drop the assumption of positive homogeneity, this is no longer true and one has to be much more careful in thinking about which quantity the risk measure should be applied to.}
\end{remark}
There are two versions of mean-$\rho$ portfolio selection:
\begin{enumerate}[label=({\arabic*})] 
\item Given a minimal desired expected excess return $\nu_{\text{min}} \geq 0$, minimise $\rho(X_\pi)$ among all portfolios $\pi \in \RR^d$ that satisfy $\mathbb{E}[X_{\pi}] \geq \nu_{\textnormal{min}}$.
    \item Given a maximal risk threshold $\rho_{\textnormal{max}} \geq 0$, maximise $\mathbb{E}[X_\pi]$ among all portfolios $\pi \in \RR^d$ that satisfy $\rho(X_\pi) \leq \rho_{\textnormal{max}}$.
   
\end{enumerate}
The way to tackle these two problems is to first study problem (1) with an equality constraint, i.e., for fixed $\nu \geq 0$, find the minimal risk among the portfolios in $\Pi_\nu$:
\begin{enumerate} 
    \item [(1')] For $\nu \geq 0$, minimise $\rho(X_\pi)$ among all portfolios $\pi \in \Pi_\nu$.
\end{enumerate}
In classical mean-variance portfolio selection, the solution to (1') exists for all $\nu \geq 0$. It also gives the solution to (1) and provides the so-called efficient frontier, which in turn can be used to derive the solution to (2). In particular, (1) and (2) are always well-posed and equivalent problems. By contrast, we shall see that in the mean-$\rho$ setting, existence in (1') is not guaranteed. Moreover, even if (1') has a solution for all $\nu \geq 0$, (1) and (2) may both be ill-posed, or (1) may be well-posed and (2) ill-posed. This implies in particular that (1) and (2) are no longer equivalent. We shall see that these issues arise exactly when the market admits so-called \emph{$\rho$-arbitrage}.

\subsection{Optimal portfolios}
We approach mean-$\rho$ portfolio selection by first looking at the slightly simplified problem (1') of finding the minimum risk portfolio(s) given a fixed excess return.  Since a negative excess return corresponds to an expected loss, we only focus on portfolios with nonnegative expected excess returns.

\begin{definition}
\label{def:optimal portfolio}
Let $\nu \geq 0$. A portfolio $\pi \in \Pi_\nu$ is called \emph{$\rho$-optimal} for $\nu$ if $\rho(X_\pi) < \infty$ and $\rho(X_\pi) \leq \rho(X_{\pi'})$ for all $\pi' \in \Pi_\nu$. We denote the set of all $\rho$-optimal portfolios for $\nu$ by $\Pi^\rho_\nu$. Moreover, we set
\begin{equation}
\label{eq:def:optimal portfolio:rho nu}
    \rho_\nu := \inf \{ \rho(X_\pi) : \pi \in \Pi_{\nu} \} \in [-\infty,\infty],\footnote{If $\rho_\nu = -\infty$, since $\rho$ can only take values in $(-\infty,\infty]$, for every portfolio in $\Pi_\nu$ there is another portfolio in $\Pi_\nu$ with strictly lower risk.  Thus, $\Pi^\rho_\nu = \emptyset$.  If $\rho_\nu=\infty$,  every portfolio in $\Pi_\nu$ has infinite risk, and so $\Pi^\rho_\nu = \emptyset$.}
\end{equation}
and define the \emph{$\rho$-optimal boundary} by
\begin{equation*}
\mathcal{O}_{\rho} := \{ (\rho_\nu, \nu): \nu \geq 0\} \subset [-\infty,\infty] \times [0,\infty).
\end{equation*}
\end{definition}

As the riskless portfolio has zero risk, $\rho_0 \leq 0$. Positive homogeneity implies that either $\rho_0 =-\infty$ (in which case $\Pi^{\rho}_0 =\emptyset$) or $\rho_0 =0$  (in which case $\0 \in \Pi^{\rho}_0$). For $\nu > 0$, positive homogeneity gives $\Pi^\rho_\nu = \nu \Pi^\rho_1$ and  $\rho_\nu = \nu \rho_1$. Thus, the $\rho$-optimal boundary is given by
\begin{equation}
\mathcal{O}_{\rho} = \{(\rho_{0},0)\} \cup \{ (k\rho_1, k): k > 0\}, \label{eq:opt bound}
\end{equation}
where $\rho_{0} \in \{ -\infty, 0\}$ and $\rho_1 \in [-\infty,\infty]$.  Note that the $\rho$-optimal boundary is nonempty even if $\rho$-optimal portfolios do not exist.  Depending on the sign of $\rho_1$, Figure \ref{optimal boundary picture} gives a graphical illustration of the three different shapes $\mathcal{O}_{\rho}$ can take when $\rho_{0}=0$ and $\rho_1 \in \mathbb{R}$.

\begin{figure}[H]
	\centering
	\includegraphics[width=1.0\textwidth]{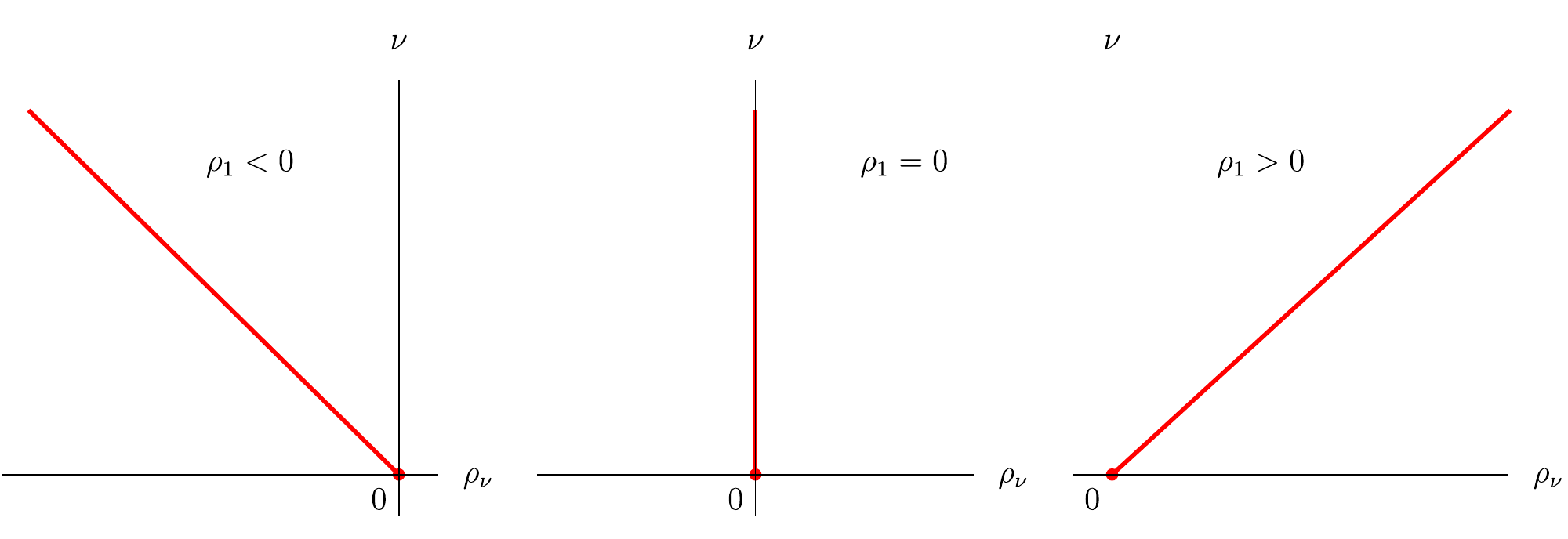}
	\caption{General shapes of the $\rho$-optimal boundary when $\rho_0=0$ and $\rho_1 \in \mathbb{R}$}
	\label{optimal boundary picture}
\end{figure}

We now seek to understand under which conditions $\rho$-optimal portfolios exist and which properties $\rho$-optimal sets have. First, we consider the case $\nu = 0$, which is also of key importance for the case $\nu > 0$.

\begin{proposition}
\label{prop:two cases depending on rho zero}
$\Pi^{\rho}_0  \neq \emptyset$ if and only if $\rho_0 = 0$. Moreover in this case, either $\Pi^{\rho}_0 = \{\0\}$ or $\Pi^{\rho}_0 $ fails to be compact.
\end{proposition}

\begin{proof}
If $\rho_0 = 0$, then $\0 \in \Pi^\rho_0$. If $\rho_0 \neq 0$, then $\rho_0 = -\infty$ and $\Pi^\rho_0 = \emptyset$. Moreover, if $\rho_0 = 0$ and there is $\pi \neq \0$ with $\rho(X_\pi) = 0$, it follows from positive homogeneity that $\lambda \pi \in \Pi^\rho_0$ for all $\lambda \geq 0$ and hence $\Pi^\rho_0$ fails to be compact.
\end{proof}

We proceed to provide sufficient conditions that guarantee $\rho_0 = 0$ or even  $\Pi^{\rho}_0 = \{\0\}$.

\begin{definition}
The risk measure $\rho$ is called \emph{expectation bounded} if $\rho(X) \geq \mathbb{E}[-X]$ for all $X \in L$. It is called \emph{strictly expectation bounded} if $\rho(X) > \mathbb{E}[-X]$ for all non-constant $X \in L$.
\end{definition}

\begin{remark}
\label{rem:exp bound}
(a) Expectation boundedness is implied by, but strictly weaker than,  \emph{dilatation monotonicity}, i.e., $\rho(X) \geq \rho(\mathbb{E}[X|\cG])$  for all $X \in L$ and all sub-$\sigma$-algebras $\cG \subset \cF$. The latter concept was introduced in \cite{leitner2004balayage} and has far reaching implications. For example, every dilatation monotone convex risk measure on an atomless probability space is law-invariant \cite{cherny2007dilatation}, and every dilatation monotone risk measure that satisfies the Fatou property can be extended to $L^1$ \cite{rahsepar2020extension}.

(b) Strict expectation boundedness -- first introduced in \cite{rockafellar2006deviation} -- is a natural requirement on a risk measure that is satisfied by Expected Shortfall and a large class of coherent risk measures; see Remark~\ref{rem:dual char}(d) and Proposition \ref{prop:P in tilde Q rho strictly exp bdd}.  In fact, when the underlying probability space is atomless and  $\rho$ is law-invariant and coherent, then it is automatically strictly expectation bounded unless $\rho(X) = \mathbb{E}[-X]$; see \cite[Corollary 5.1]{follmer2013convex}.

(c) Value at Risk is \emph{not} expectation bounded (apart from degenerate probability spaces). For example, if $Z$ is a standard normal random variable, then $\VaR^\alpha(Z) < 0 = \mathbb{E}[-Z]$ for $\alpha > 1/2$. This failure of expectation boundedness for Value at Risk has some undesirable consequences like the non-existence of optimal portfolios; cf.~Remark \ref{rem:elliptical reg arbitrage}.

(d) By cash-invariance of $\rho$, it suffices to consider $X \in L$ with $\mathbb{E}[X] = 0$ in the definition of (strict) expectation boundedness.
\end{remark}

We proceed to show that under (strict) expectation boundedness of $\rho$, optimal portfolios for $\nu = 0$ exist (and are unique).

\begin{corollary}
\label{cor:exp bounded}
If $\rho$ is expectation bounded, then $\rho_0 = 0$.  If $\rho$ is even strictly expectation bounded, then $\Pi^\rho_0=\{\mathbf{0}\}$.
\end{corollary}

\begin{proof}
If $\rho$ is expectation bounded, then for any $\pi \in \Pi_0$, $\rho(X_\pi) \geq \mathbb{E}[-X_\pi] = 0$ and we may conclude that $\rho_0 = 0$.  
 If $\rho$ is strictly expectation bounded, fix $\pi \in \Pi_0 \setminus \{\mathbf{0}\}$. Then $X_\pi$ is non-constant by nonredundancy of the financial market. Strict expectation boundedness of $\rho$ gives $\rho(X_\pi) >  \mathbb{E}[-X_\pi] = 0$. We may conclude that $\Pi^\rho_0 = \{ \mathbf{0}\}$.
\end{proof}

We next consider $\rho$-optimal sets for $\nu>0$.  To this end, we recall the Fatou property for $\rho$.

\begin{definition}
\label{def:Fatou simple}
The risk measure $\rho$ is said to satisfy the Fatou property on $\cX$, if $X_n \to X$ $\as{\P}$ for $X_n, X \in \cX$ and $|X_n| \leq Y$ $\as{\P}$ for some $Y \in L$ implies that $\rho(X) \leq \liminf_{n \to \infty} \rho(X_n)$.
\end{definition}

We now come to our main result of this section, which establishes existence of $\rho$-optimal portfolios under very weak assumptions on $\rho$, only requiring that $\rho$ satisfies the Fatou property on $\mathcal{X}=\{X_\pi:\pi \in \RR^d\}$ and $\Pi^\rho_0 = \{\mathbf{0}\}$. In particular, we do \emph{not} require $\rho$ to be convex, which is a key assumption in the extant literature; see e.g.~\cite[Proposition 4]{MasterFundsRockafellar}.

\begin{theorem}
\label{thm:Existence of optimal portfolios}
Assume $\Pi^\rho_0 = \{ \mathbf{0} \}$, $\rho_1 \in \mathbb{R}$ and $\rho$ satisfies the Fatou property on $\mathcal{X}=\{X_\pi:\pi \in \RR^d\}$.  Then for any $\nu \geq 0$, the set $\Pi^\rho_\nu$ of $\rho$-optimal portfolios for $\nu$ is nonempty and compact.
\end{theorem}

\begin{proof}
The key idea of the proof is to consider the function $f_{\rho}: \mathbb{R}^d \to [0,\infty]$, defined by
\begin{equation*}
f_{\rho}(\pi) = 
    \begin{cases}
    \rho(X_{\pi}) + (|\rho_1|+1)\mathbb{E}[X_{\pi}], &\text{if }   \pi \in \cup_{k \geq 0} \Pi_{k},\\
	\infty, &\text{if } \pi \in \cup_{k < 0} \Pi_{k}.
	
    \end{cases}
\end{equation*}
Then $f_\rho$ is a nonnegative, positively homogeneous function and satisfies $f_\rho^{-1}(\{0\}) = \{ \mathbf{0} \}$. Moreover, if $\pi_n \to \pi$ in $\RR^d$, we have $\mathbb{E}[X_\pi] = \lim_{n \to \infty}\mathbb{E}[X_{\pi_{n}}]$ as well as $\rho(X_{\pi}) \leq \liminf_{n \to \infty} \rho(X_{\pi_{n}})$ because
$\rho$ satisfies the Fatou property on $\mathcal{X}$ (and $L \supset \cX$ is a Riesz space). This implies that $f_\rho$ is lower semi-continuous.

We proceed to show that $f_{\rho}$ has compact sublevel sets. As $\rho_1 < \infty$, there is at least one portfolio $\pi^* \in \Pi_1$ with $\rho(X_{\pi^*})<\infty$.  Let $S = \{ x \in \mathbb{R}^d : \lVert x \rVert_{2} = \lVert \pi^* \rVert_2 \}$. As $S$ is compact and $f_\rho$ lower semi-continuous,
$m:=\min\{ f_\rho(x) : x \in S \}$ is well defined. Note that $m > 0$ since $\lVert \pi^* \rVert_2 > 0$ and $f_\pi^{-1}(\{0\}) = \{ \mathbf{0} \}$. As $f_\rho$ is positively homogeneous, $f_\rho(\pi) \geq \frac{m}{\lVert \pi^* \rVert_2}\lVert \pi \rVert_2$ for any portfolio $\pi \in \RR^d$. Thus, $f_\rho$ has bounded sublevel sets, which are also closed since $f_\rho$ is lower semi-continuous.

We finish by a standard argument. For $\delta \geq 0$, set $A_{\delta} := \{\pi \in \mathbb{R}^d : f_{\rho}(\pi) \leq \delta \} \cap \Pi_{1}$ and $\delta_{1} := \inf \{f_{\rho}(\pi) : \pi \in \Pi_{1} \}$. Note that $\delta_1 < \infty$ since $\rho_1 \in \mathbb{R}$. Moreover, each $A_\delta$ is compact and nonempty for $\delta > \delta_1$. As the $A_\delta$ are nested (i.e., $A_{\delta} \subset A_{\delta^\prime}$ for $\delta \leq \delta^\prime$), it follows that
\begin{equation*}
\Pi^{\rho}_{1} = A_{\delta_{1}} = \bigcap_{\delta > \delta_{1}} A_{\delta}
\end{equation*}
is nonempty and compact.  Whence, so is $\Pi^\rho_\nu = \nu \Pi^\rho_1$ for any $\nu >0$. (For $\nu = 0$, the claim is trivial.)
\end{proof}

\begin{remark}
\label{rmk:existence of rho optimal portfolios}
(a) The requirement that $\rho$ satisfies the Fatou property on $\mathcal{X}=\{X_\pi:\pi \in \RR^d\}$ is a mild assumption, which is satisfied by VaR, ES and WC. Anticipating ourselves a bit, we note that it is satisfied 
by any risk measure $\rho: L \xrightarrow{} (-\infty, \infty]$ admitting a dual representation $\rho(X) = \sup_{Z \in \mathcal{Q}} \mathbb{E}[-Z X]$ for some nonempty set $\cQ$ of Radon-Nikod{\'y}m derivatives satisfying $Z R^i \in L^1$ for all $Z \in \cQ$ and $i \in \{1, \ldots, d\}$; cf.~Proposition \ref{prop:cond I}. 

(b) By Corollary \ref{cor:exp bounded}, the requirement that $\Pi^\rho_0 = \{ \mathbf{0} \}$ is automatically satisfied if $\rho$ is strictly expectation bounded.  By Remark \ref{rem:exp bound}(b), this is not very restrictive.

(c) If $\rho$ is in addition \emph{convex}, i.e., $\rho(\lambda X_1+(1-\lambda)X_2) \leq \lambda \rho(X_1)+ (1-\lambda) \rho(X_2)$ for $X_{1}, X_{2} \in L$ and $\lambda \in [0,1]$, 
then we also have convexity of $\rho$-optimal sets. Indeed, let $\nu \geq 0$, $\pi,\pi^\prime \in \Pi^\rho_\nu$, and $\lambda\in [0,1]$.  Then $\rho(X_{\lambda\pi+(1-\lambda)\pi^{\prime}}) = \rho(\lambda X_{\pi} + (1-\lambda) X_{\pi^\prime}) \leq \lambda\rho(X_{\pi}) + (1-\lambda)\rho(X_{\pi^{\prime}}) = \rho_\nu$.  Therefore, $\lambda\pi+(1-\lambda)\pi^{\prime} \in \Pi^\rho_\nu$. 
    
(d) If $|\rho_1| = \infty$, then $\Pi^\rho_\nu = \emptyset$ for all $\nu > 0$.  If $\rho_1 \in \mathbb{R}$ and $\{ \mathbf{0} \} \subsetneq\Pi^\rho_0 $, then boundedness of the sublevel sets is lost (since $f_\rho^{-1}(\{0\}) $ is unbounded) and $\Pi^\rho_\nu$ can be empty for all $\nu > 0$; see Example~\ref{example:rho one not attained but finite} for a concrete counterexample.
\end{remark}

\subsection{Efficient portfolios}
We proceed to study the notion of $\rho$-efficient portfolios, which are defined in analogy to efficient portfolios in the classical mean-variance sense.

\begin{definition}
	\label{def:rho efficient portfolio}
A portfolio $\pi \in \RR^d$ is called \emph{$\rho$-efficient} if $\mathbb{E}[X_\pi] \geq 0$ and there is no other portfolio $\pi' \in \RR^d$ that is strictly $\rho$-preferred over $\pi$. We denote the \emph{$\rho$-efficient frontier} by
	\begin{equation*}
	\mathcal{E}_{\rho} := \{ (\rho(X_{\pi}),\mathbb{E}[X_{\pi}]): \pi \text{ is $\rho$-efficient} \} \subset \mathbb{R}^2.
\end{equation*}
\end{definition}

\begin{remark}
\label{rmk:efficient portfolios are optimal}
(a) If $\pi \in \RR^d$ is $\rho$-efficient, it follows that $\rho(X_\pi) < \infty$. Indeed, if $\mathbb{E}[X_\pi] = 0$ and $\rho(X_\pi)  = \infty$, then $\0$ is strictly $\rho$-preferred over $\pi$, and if 
$\mathbb{E}[X_\pi] > 0$ and $\rho(X_\pi) = \infty$, then $\lambda \pi$ is strictly $\rho$-preferred over $\pi$ for $\lambda > 1$.

(b) It follows from (a) that every $\rho$-efficient portfolio is $\rho$-optimal.

(c) If $\rho$ is expectation bounded, we may drop the assumption that $\mathbb{E}[X_\pi] \geq 0$ for $\pi$ to be efficient since under expectation boundedness, for any portfolio $\pi$ with $\mathbb{E}[X_\pi] < 0$, we have $\rho(X_\pi) \geq \mathbb{E}[-X_\pi] > 0$, and so the riskless portfolio $\0$ is strictly $\rho$-preferred over $\pi$.
\end{remark}

The mean-$\rho$ portfolio selection problems (1) and (2) from the beginning of Section \ref{section:portfolio optimisation} are both well-posed and admit solutions when $\rho$-efficient portfolios exist, i.e., when $\cE_\rho \neq \emptyset$.  Remark \ref{rmk:efficient portfolios are optimal}(b) implies that $\mathcal{E}_{\rho} \subset \mathcal{O}_{\rho}$.  However, unlike in the case of mean-variance portfolio optimisation, it can happen that there are no $\rho$-efficient portfolios -- even if $\rho$-optimal portfolios exist for all $\nu \geq 0$.  The following result shows that when $\Pi^{\rho}_{\nu} \neq \emptyset$ for all $\nu \geq 0$ (which is satisfied under the conditions of Theorem \ref{thm:Existence of optimal portfolios}), then the existence of the $\rho$-efficient frontier depends only on the sign of $\rho_{1}$.  

\begin{proposition}
\label{prop:efficient portfolios}
Assume $\Pi^{\rho}_{\nu} \neq \emptyset$ for all $\nu \geq 0$.  
Then the following are equivalent:
\begin{enumerate}
    \item $\rho_1>0$.
    \item $\mathcal{E}_{\rho} \neq \emptyset$.
\end{enumerate}
Moreover, if $\rho_1>0$, the $\rho$-efficient frontier is given by 
\begin{equation*}
	\mathcal{E}_{\rho} = 
	\{(k\rho_1,k) : k\geq 0\}.
\end{equation*}
\end{proposition}

\begin{proof}
First assume that $\rho_1 >0$. We proceed to show that  any $\rho$-optimal portfolio is $\rho$-efficient. It then follows 
from Remark \ref{rmk:efficient portfolios are optimal}(b) and Proposition \ref{prop:two cases depending on rho zero} that
\begin{equation*}
	\mathcal{E}_{\rho} = \mathcal{O}_{\rho} = 
	\{(k\rho_1,k) : k\geq 0\}.
\end{equation*}
Seeking a contradiction, let $\pi \in \Pi^\rho_\nu$ for some $\nu \geq 0$ and assume that there is $\pi' \in \RR^d$ such that $\mathbb{E}[X_{\pi^{\prime}}] \geq \mathbb{E}[X_\pi] =\nu$ and $\rho(X_{\pi^{\prime}}) \leq \rho(X_\pi) =\nu \rho_1$, with one inequality being strict. Set $\nu' := \mathbb{E}[X_{\pi^{\prime}}]$. If $\nu' =\nu$, then $\rho(X_{\pi^{\prime}}) < \rho(X_\pi)$ and we arrive at a contradiction as $\pi \in \Pi^\rho_\nu$. Otherwise, if $\nu'>\nu$, let $\pi^* \in \Pi^\rho_{\nu'}$. Then  $ \nu' \rho_1 = \rho(X_{\pi^{*}}) \leq \rho(X_{\pi^{\prime}}) \leq \rho(X_\pi) =\nu \rho_1$. Since $\rho_1 > 0$, we arrive at the contradiction that $\nu' > \nu$ and $\nu' \leq \nu$.

Now assume that $\rho_1 \leq 0$.  We proceed to show that there does not exist any $\rho$-efficient portfolio, even though $\Pi^\rho_\nu \neq \emptyset$ for all $\nu \geq 0$. Seeking a contradiction, suppose that $\pi \in \RR^d$ is $\rho$-efficient. Then by Remark \ref{rmk:efficient portfolios are optimal}(b), $\pi \in \Pi^\rho_\nu$ for some $\nu \geq 0$. Pick $\nu' > \nu$ and let $\pi' \in \Pi^\rho_{\nu'}$. Then $\mathbb{E}[X_{\pi^{\prime}}] = \nu' > \nu = \mathbb{E}[X_\pi]$ and $\rho(X_{\pi^{\prime}}) = \nu' \rho_1 \leq  \nu \rho_1 = \rho(X_\pi)$  by positive homogeneity of $\rho$ and $\rho_1 \leq 0$. Hence, $\pi'$ is strictly $\rho$-preferred over $\pi$ and we arrive at a contradiction.
\end{proof}

\begin{remark}
	\label{rem:prop:efficient portfolios}
A close inspection of the proof of Proposition \ref{prop:efficient portfolios} reveals that the equivalence between (a) and (b) remains true if we only require that $\Pi^{\rho}_{\nu} \neq \emptyset$ for all $\nu > 0$. However, if $\Pi^{\rho}_{0} = \emptyset$, the $\rho$-efficient frontier is given by 	$\mathcal{E}_{\rho} = 
\{(k\rho_1,k) : k> 0\}.$\footnote{It is an open question if there exists a risk measure satisfying $\Pi^{\rho}_{\nu} \neq \emptyset$ for all $\nu > 0$ but $\Pi^{\rho}_{0} =\emptyset$. It is clear that if it exists, $\rho$ fails to be convex.}
\end{remark}

The following figure gives a graphical illustration of Proposition \ref{prop:efficient portfolios}.
\begin{figure}[H]
	\centering
	\includegraphics[width=1.0\textwidth]{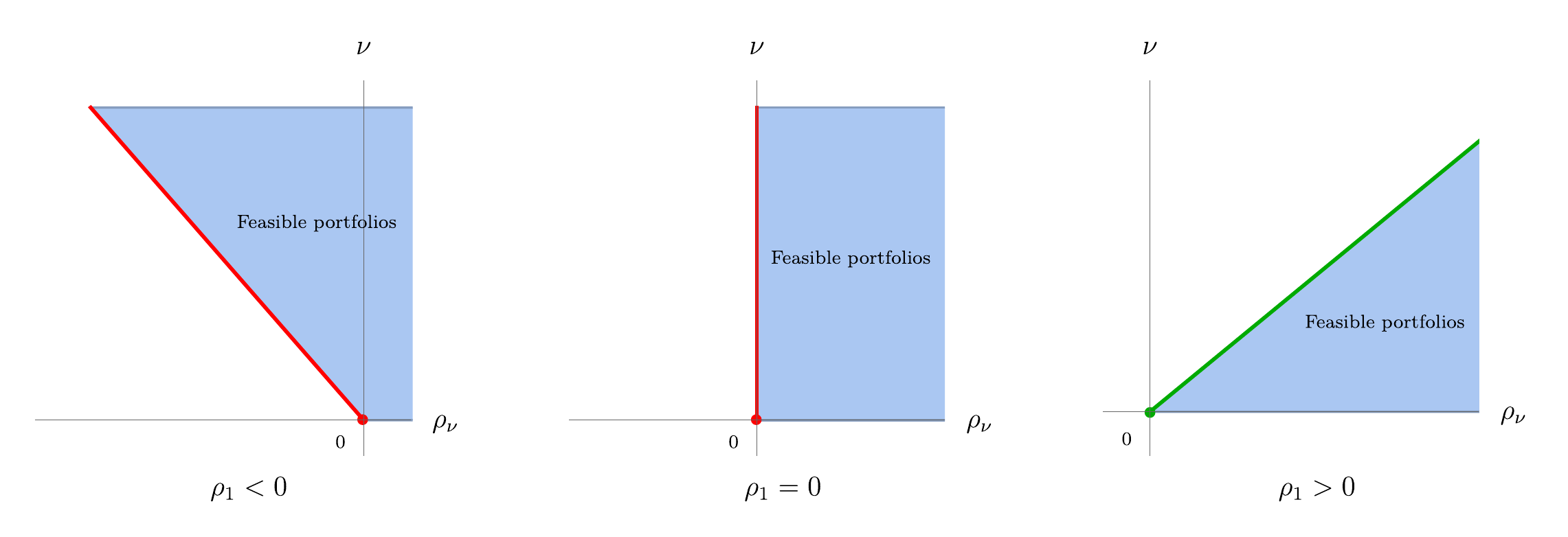}
	\caption{$\rho$-optimal boundary (red) and $\rho$-efficient frontier (green) when $\Pi^{\rho}_{\nu} \neq \emptyset$ for all $\nu \geq 0$}
	\label{efficient frontier picture}
\end{figure}

\subsection{$\rho$-arbitrage}

We have seen above that mean-$\rho$ portfolio selection is not always well defined as it can happen that there are no $\rho$-efficient portfolios. 
We call this situation \emph{$\rho$-arbitrage}.

\begin{definition} The market $(S^0, S)$ is said to satisfy \emph{$\rho$-arbitrage} if there are no $\rho$-efficient portfolios. It is said to satisfy \emph{strong $\rho$-arbitrage} if for any portfolio $\pi \in \RR^d$, there exists another portfolio $\pi^\prime$  such that 
\begin{equation*}
\mathbb{E}[X_{\pi'}] > \mathbb{E}[X_{\pi}]  \quad \text{and} \quad \rho(X_{\pi'}) < \rho(X_\pi).
\end{equation*}
\end{definition}

It is clear that strong $\rho$-arbitrage implies $\rho$-arbitrage but not vice versa. The following two theorem give primal characterisations.  Whereas strong $\rho$-arbitrage is fully characterised by the sign of~$\rho_{1}$, defined in \eqref{eq:def:optimal portfolio:rho nu}, the case of  $\rho$-arbitrage is more subtle.

\begin{theorem}
\label{thm:strong reg arb first characterisation}
The market $(S^0, S)$ admits strong $\rho$-arbitrage if and only if $\rho_{1}<0$.
\end{theorem}

\begin{proof}
First, assume that the market satisfies strong $\rho$-arbitrage.  As the riskless portfolio has zero risk and zero return, by definition of strong $\rho$-arbitrage, there is a portfolio $\pi \in \RR^d$ with $\mathbb{E}[X_\pi]=:\nu>0$ and $\rho(X_{\pi}) < 0$.  Let $\pi^\prime := \tfrac{1}{\nu}\pi$.  Then $\pi^\prime \in \Pi_1$, and
\begin{equation*}
    \rho_{1} \leq \rho(X_{\pi^{\prime}}) = \tfrac{1}{\nu}\rho(X_\pi) < 0.
\end{equation*}

Conversely, assume that $\rho_1 < 0$. Then there exists a portfolio $\pi \in \Pi_1$ with $\rho(X_\pi) < 0$.  Thus,  $\mathbb{E}[X_{k\pi}] \to \infty$ and $\rho(X_{k\pi}) \to -\infty$ as $k \to \infty$.  Therefore, for any portfolio $\pi^{\prime}$ (recalling that $\mathbb{E}[X_{\pi^{\prime}}] \in \mathbb{R}$ and $\rho(X_{\pi^{\prime}}) \in (-\infty,\infty]$), there exists $k \in \NN$ such that $\mathbb{E}[X_{k\pi}] > \mathbb{E}[X_{\pi^{\prime}}] $ and $\quad \rho(X_{k\pi})  < \rho(X_{\pi^{\prime}})$. Hence, the market satisfies strong $\rho$-arbitrage.
\end{proof}

\begin{remark}
It follows directly from Theorem \ref{thm:strong reg arb first characterisation} and its proof that 
the market $(S^0, S)$ admits strong $\rho$-arbitrage if and only if there exists a sequence of portfolios $(\pi_{n})_{n \in \NN} \subset \RR^d$ with
\begin{equation*}
	\mathbb{E}[X_{\pi_{n}}] \uparrow \infty \quad \textnormal{and} \quad \rho(X_{\pi_{n}}) \downarrow -\infty.
\end{equation*}
This alternative characterisation of strong $\rho$-arbitrage shows in a particular striking way how undesirable this property is from a regulatory perspective.
\end{remark}

\begin{theorem}
	\label{thm:reg arb:first characterisation}
We have the following three cases:
\begin{enumerate}
    \item If $\Pi^{\rho}_{1} \neq \emptyset$, then the market $(S^0, S)$ admits $\rho$-arbitrage if and only if $\rho_{1} \leq 0$. 
    \item If $\Pi^{\rho}_{1} = \emptyset$ and $\Pi^{\rho}_{0} \neq \emptyset$, then the market $(S^0, S)$ admits $\rho$-arbitrage if and only if $\rho_{1} < 0$. 
    \item If $\Pi^{\rho}_{1} = \emptyset$ and $\Pi^{\rho}_{0} = \emptyset$, then the market $(S^0, S)$ admits $\rho$-arbitrage.
\end{enumerate}
\end{theorem}

\begin{proof}
(a) This follows from Proposition \ref{prop:efficient portfolios} and Remark \ref{rem:prop:efficient portfolios}.

(b) If $\rho_1<0$, by Theorem \ref{thm:strong reg arb first characterisation} the market admits strong $\rho$-arbitrage and a fortiori $\rho$-arbitrage.  Conversely, if $\rho_1 \geq 0$, any portfolio $\pi \in \RR^d$ with $\mathbb{E}[X_\pi]=:\nu > 0$ has $\rho(X_\pi) > \nu \rho_1 = 0$ because $\Pi^\rho_\nu = \nu \Pi^\rho_1 = \emptyset$. Thus, any portfolio in $\Pi^\rho_0$ is 
$\rho$-efficient because $\Pi^\rho_0 \neq \emptyset$ (and therefore $\rho_0 = 0$). Thus, the market does not admit $\rho$-arbitrage.  

(c) This follows from Remark \ref{rmk:efficient portfolios are optimal}(b).
\end{proof}

The following corollary relates the absence of (strong) $\rho$-arbitrage to the existence of the mean-$\rho$ portfolio selection problems (1) and (2) from the beginning of Section \ref{section:portfolio optimisation}. The proof is straightforward and hence omitted.

\begin{corollary}
	Assume that $\Pi^{\rho}_{\nu} \neq \emptyset$ for all $\nu \geq 0$, so that the problem 	{\normalfont(1')} is well posed.
	\begin{enumerate}
\item The mean-$\rho$ portfolio selection problem 	{\normalfont(1)} is well posed if any only if the market $(S^0,S)$ does not satisfy strong $\rho$-arbitrage. In this case, the portfolios that solve {\normalfont(1)} are in
\begin{equation*}
    \begin{cases}
\Pi^\rho_{\nu_{\min}},&\text{if } \rho_1>0,\\ \cup_{\nu \geq \nu_{\min}} \Pi^\rho_{\nu},&\text{if } \rho_1=0.
\end{cases}
\end{equation*}
\item The mean-$\rho$ portfolio selection problem 	{\normalfont(2)} is well posed if and only if the market  $(S^0,S)$  does not satisfy $\rho$-arbitrage. In this case, the portfolios that solve {\normalfont(2)} are in $\Pi^\rho_{\rho_{\max}/\rho_1}$.
	\end{enumerate}
\end{corollary}

A natural question that arises is how (strong) $\rho$-arbitrage is related to the ordinary notion of arbitrage.  To this end, recall that the market $(S^0,S)$ is said to satisfy
\begin{itemize}
	\item \emph{arbitrage of the first kind} if there exists a trading strategy $(\theta^0,\theta) \in \mathbb{R}^{1+d}$ such that\footnote{Note that $(\theta^0,\theta)$ parametrises trading in \emph{numbers of shares} rather than in fractions of wealth.}
	\begin{equation*}
	\theta^0 S^0_0 + \theta \cdot S_0 \leq 0, \quad \theta^0 S^0_1 + \theta \cdot S_1 \geq 0 \; \mathbb{P}\textnormal{-a.s.} \quad \textnormal{and} \quad \mathbb{P}[\theta^0 S^0_1 + \theta \cdot S_1 > 0] > 0. 
	\end{equation*}
	\item \emph{arbitrage of the  second kind} if there exists a trading strategy $(\theta^0,\theta) \in \mathbb{R}^{1+d}$ such that 
	\begin{equation*}
	\theta^0 S^0_0 + \theta \cdot S_0 < 0, \quad \textnormal{and} \quad \theta^0 S^0_1 + \theta \cdot S_1 \geq 0 \; \mathbb{P}\textnormal{-a.s.}
	\end{equation*}
	\end{itemize}
The following result shows that if $\rho$ is given by the worst-case risk measure $\textnormal{WC}$, (strong) $\textnormal{WC}$-arbitrage is equivalent to arbitrage of the first (second) kind.  Thus, $\rho$-arbitrage can be seen as an extension of the ordinary notion of arbitrage.

\begin{proposition}
\label{prop:WC arbitrage equivalent to ordinary arbitrage}
    The market $(S^0,S)$ satisfies (strong) \textnormal{WC}-arbitrage if and only if the market satisfies arbitrage of the first (second) kind.
\end{proposition}

\begin{proof}
First note that by Theorem \ref{thm:Existence of optimal portfolios} either $\textnormal{WC}_1 = \infty$ or $\Pi^{\textnormal{WC}}_1 \neq \emptyset$.  

Now if $(\theta^0,\theta) \in \mathbb{R}^{1+d}$ is an arbitrage of the first (second) kind, then $\pi := (\theta^1 S^1_0, \ldots, \theta^d S^d_0) \neq \0$ satisfies $X_\pi =  \pi \cdot (R - r\1) = (\theta^0 S^0_1 + \theta \cdot S_1) - (1 + r) (\theta^0 S^0_0 + \theta \cdot S_0) \geq 0 \;(>0)\; \mathbb{P}$-a.s., which implies that $\mathrm{WC}(X_\pi)$ is nonpositive (negative).  Hence, $\textnormal{WC}_1 \leq 0$ $(\textnormal{WC}_1 < 0$) and $\Pi^{\textnormal{WC}}_1 \neq \emptyset$.  It follows that the market satisfies (strong) WC-arbitrage by Theorem \ref{thm:reg arb:first characterisation}(a) (Theorem \ref{thm:strong reg arb first characterisation}). 

Conversely, if the market does not satisfy arbitrage of the first (second) kind, then for all $\pi \in \Pi_1$, $\textnormal{WC}(X_\pi) > 0$ ($\textnormal{WC}(X_\pi) \geq 0$).  Since $\textnormal{WC}_1 = \infty$ or $\Pi^{\textnormal{WC}}_1 \neq \emptyset$, it follows that $\textnormal{WC}_1 > 0$ ($\textnormal{WC}_1 \geq 0$).  Hence, the market does not satisfy (strong) WC-arbitrage by Theorem \ref{thm:reg arb:first characterisation} (Theorem \ref{thm:strong reg arb first characterisation}). 
\end{proof}

We say that the market is arbitrage-free if it does not admit arbitrage of the first kind.  One might wonder if (strong) $\rho$-arbitrage is not just a pathology that disappears for reasonable (i.e., arbitrage-free) markets and risk measures. This is not the case. The following result shows that unless $\rho$ is as conservative as the worst-case risk measure, one can \emph{always} construct a financial market that is arbitrage-free but admits strong $\rho$-arbitrage.

\begin{theorem}
\label{thm:cannot avoid reg arb}  
Assume $\rho:L \to (-\infty,\infty]$ is not as conservative as the worst-case risk measure, \textnormal{WC}. Then there exists a market $(S^0, S)$ that is arbitrage-free but admits strong $\rho$-arbitrage.
\end{theorem}

\begin{proof}
It is enough to construct a random variable $R \in L$ with $\mathbb{E}[R] > 0$, $\P[R < 0] > 0$ and $\rho(R) < 0$. Indeed, we can then define the market $(S^0, S)$ by $S^0 \equiv 1$ and $S := S^1$, where $S^1_0 = 1$ and $S^1_1 = 1 + R$. This is nonredundant, nondegenerate, and arbitrage-free but admits strong $\rho$-arbitrage by Theorem \ref{thm:strong reg arb first characterisation} since $\rho_1 < 0$.
	
	First, if $\rho$ is not expectation bounded, there exists $X \in L$ such that $ \mathbb{E}[-X] - \rho(X) := \epsilon > 0$.  By cash-invariance of $\rho$, this implies that $X$ cannot be constant so $\esssup(-X+E[X]) > 0$. Set $\delta \in (0, \esssup(-X+E[X]))$ and let $R:=X-\mathbb{E}[X] + \delta$. Then $\mathbb{E}[R]=\delta > 0$, $\P[R < 0] > 0$ and $\rho(R) =- \epsilon - \delta < 0$.
	
	Next, if $\rho$ is expectation bounded but not as conservative as WC, there exists $X \in L$ such that $\rho(X) < \esssup(-X) \leq \infty$.  Let $m \in (\rho(X),\esssup(-X))$ and $R:=X+m$.  Then $\P[R < 0] > 0$, $\rho(R)<0$ and $\mathbb{E}[R] \geq -\rho(R) > 0$ by expectation boundedness of $\rho$. 
\end{proof}

\subsection{$\rho$-arbitrage for elliptical returns}
\label{sec:elliptical}
The primal characterisations of (strong) $\rho$-arbitrage in Theorems \ref{thm:strong reg arb first characterisation} and \ref{thm:reg arb:first characterisation} are particularly useful when returns  are elliptically distributed with finite second moments and the risk measure is law-invariant. We briefly recall both concepts.

\begin{definition}\label{Elliptical Distribution Definition}
An $\mathbb{R}^d$-valued random vector $X = (X_{1},\dots,X_{d})$ has an \emph{elliptical distribution} if there exists a \emph{location vector} $\widetilde{\mu} \in \RR^d$, a $d \times d$ nonnegative definite \emph{dispersion matrix} $\widetilde{\Sigma} \in \RR^{d \times d}$, and a \emph{characteristic generator} $\psi:  \left[0,\infty \right) \rightarrow \mathbb{R}$ such that the characteristic function of $X$, $\phi_{X}$ can be expressed as
\begin{equation*}
\phi_{X}(t) = e^{i t^\top \widetilde{\mu}}\psi(t^{T}\widetilde{\Sigma}t) \quad \textnormal{for all} \ t \in \mathbb{R}^{d}.
\end{equation*}
In this case we write $X \sim \tilde E_{d}(\widetilde{\mu}, \widetilde{\Sigma}, \psi)$. 
\end{definition}

Elliptical distributions are generalisations of the multivariate normal distribution, which allow for heavy tail models while possessing many useful properties.  Indeed, the fat tails of most of their members make them natural candidates in modelling the distribution of speculative returns. Examples of elliptical distributions include the multivariate normal distribution, the multivariate t-distribution and the multivariate symmetric Laplace distribution. For a thorough description of elliptical distributions refer to~\cite{EllipticalDistBook,EllipticalDistPaper}.

\begin{remark}\label{RemarkAboutElliptical} 
If $X$ has an elliptical distribution with finite second moments, $X$ is also characterised by its mean vector $\mu \in \RR^d$, covariance matrix $\Sigma \in \RR^{d \times d}$  and characteristic generator $\psi$. Therefore, we may write $X \sim E_{d}\left(\mu, \Sigma, \psi\right)$; see \cite[Remark 3.27]{QRMbook} for details.
\end{remark}

\begin{definition}
\label{defn:law invariance}
A risk measure $\rho : L \to (-\infty, \infty]$ is called \emph{law-invariant} if $\rho(X_1)=\rho(X_2)$ whenever $X_1,X_2 \in L$ have the same law.
\end{definition}

The following result shows why elliptical distributions and law-invariant risk measures work particularly nicely together.

\begin{lemma}
	\label{lemma:risk of elliptical dist}
	Suppose $\rho$ is law-invariant and the return vector $R$ has an elliptical distribution with mean vector $\mu \in \RR^d$, covariance matrix $\Sigma \in \RR^{d \times d}$ and characteristic generator $\psi$.  Assume $\{X \sim E_{1}(\mu_{X}, \sigma^2_X, \psi):\mu_X \in \RR, \ \sigma^{2}_X \geq 0 \} \subset L$ and let $Z \sim E_{1}\left(0, 1, \psi\right)$.  Then for any $\pi \in \mathbb{R}^d$,\footnote{Note that $\rho(Z) \in (-\infty,\infty]$. We employ the convention that $\infty \times 0 = 0$, so that $\rho(X_\pi) = -\mathbb{E}[X_\pi]$ if $\Var(X_{\pi})=0$.} 	 
	\begin{equation}
	\label{eqn:rho for elliptical dist}
	\rho(X_\pi) = -\mathbb{E}[X_\pi] + \rho(Z)\sqrt{\Var(X_{\pi})} = - \pi^\top (\mu - r\1) + \rho(Z)\sqrt{\pi^\top \Sigma \pi}.
	\end{equation}
Moreover, $\rho(Z)$  is nonnegative (positive) if $\rho$ is (strictly) expectation bounded.
\end{lemma}

\begin{proof}
Standard properties of elliptical distributions imply that 
$$\pi \cdot (R-r\mathbf{1}) \sim E_{1}\left(\pi \cdot (\mu - r\mathbf{1}), \pi^{\textnormal{T}}\Sigma \pi, \psi\right)$$ for any portfolio $\pi \in \mathbb{R}^d$. This means that $X_\pi \,{\buildrel d \over =}\, \pi^\top (\mu - r\1) + Z\sqrt{\pi^\top \Sigma \pi}$, where $Z \sim E_{1}\left(0, 1, \psi\right)$.  As $\rho$ is a law-invariant,  $\rho(X_{\pi}) = - \pi^\top (\mu - r\1) + \rho(Z)\sqrt{\pi^\top \Sigma \pi}$. The final claim follows from the fact that $\mathbb{E}[Z] = 0$ because $Z \sim E_{1}\left(0, 1, \psi\right)$ has a symmetric distribution.
\end{proof}

With the help of Lemma \ref{lemma:risk of elliptical dist}, we can give a very simple characterisation for the absence of (strong) $\rho$-arbitrage in terms of the maximal \emph{Sharpe ratio}.

\begin{corollary}
	\label{cor:SR}
	Suppose $\rho$ is law-invariant and the return vector $R$ has an elliptical distribution with mean vector $\mu \in \RR^d$ satisfying $\mu \neq r\1$, positive definite covariance matrix $\Sigma \in \RR^{d \times d}$ and characteristic generator $\psi$. Assume $\{X \sim E_{1}(\mu_{X}, \sigma^2_X, \psi):\mu_X \in \RR, \ \sigma^{2}_X \geq 0 \} \subset L$ and let $Z \sim E_{1}\left(0, 1, \psi\right)$.
	Define the maximal Sharpe ratio as
	\begin{equation}
	\label{eq:cor:SR}
	\SR_{\max} := \max_{\pi \in \RR^d \setminus \{\0\}} \frac{\mathbb{E}[X_\pi]}{\sqrt{\Var(X_{\pi})}} = \sqrt{(\mu - r\1)^\top \Sigma^{-1}(\mu - r \1)}.
	\end{equation}
Then we have the following trichotomy:
	\begin{enumerate}
		\item If $\SR_{\max} < \rho(Z)$, the market $(S^0, S)$ does not admit $\rho$-arbitrage.
		\item If $\SR_{\max} = \rho(Z)$, the market $(S^0, S)$ admits $\rho$-arbitrage but not strong $\rho$-arbitrage.
		\item If $\SR_{\max} > \rho(Z)$, the market $(S^0, S)$ admits strong $\rho$-arbitrage.
	\end{enumerate}
In particular, if $\rho(Z) \leq 0$, the market $(S^0, S)$ admits strong $\rho$-arbitrage, independent of $\mu$ or $\Sigma$. Moreover, if $\rho(Z) < 0$ and $d \geq 2$, $\rho$-optimal portfolios fail to exist for any $\nu \geq 0$, independent of $\mu$ or $\Sigma$.
\end{corollary}

\begin{proof}
For $\pi \in \RR^d \setminus \{\0\}$, set $\textnormal{SR}_{\pi} := \mathbb{E}[X_\pi]/\sqrt{\Var(X_{\pi})}$ and note that this is well defined because $\mu \neq r\mathbf{1}$ and $\Sigma$ is positive definite. It follows from linearity of the expectation and positive homogeneity of the standard deviation that $\SR_{\max} := \max_{\pi \in \Pi_1} \SR_\pi$. It is not difficult to check that the portfolio $\pi^* := \frac{1}{(\mu-r\mathbf{1})^{\textnormal{T}}\Sigma^{-1}(\mu-r\mathbf{1})} \Sigma^{-1}(\mu-r\mathbf{1}) \in \Pi_1$ has maximal Sharpe ratio given by the right-hand side of \eqref{eq:cor:SR}.

If $\rho(Z) \in (0,\infty)$, then by Lemma \ref{lemma:risk of elliptical dist} for any $\pi \in \Pi_1$,
\begin{equation*}
\rho(X_\pi) = -1+\rho(Z) \sqrt{\Var(X_{\pi})} = -1 +  \frac{\rho(Z)}{\textnormal{SR}_{\pi}}.
\end{equation*}
Thus, minimising $\rho(X_\pi)$ over $\pi \in \Pi_1$ is equivalent to maximising $\textnormal{SR}_\pi$ over $\Pi_1$.  Whence 
\begin{equation*}
\rho_{1} := -1 + \frac{\rho(Z)}{\textnormal{SR}_{\textnormal{max}}} = -1 + \frac{\rho(Z)}{\textnormal{SR}_{\pi^*}} = \rho(X_{\pi^*}).
\end{equation*}
Parts (a), (b) and (c) now follow from  Theorems \ref{thm:strong reg arb first characterisation} and Theorem \ref{thm:reg arb:first characterisation}(a).

If $\rho(Z) = \infty$, every portfolio has infinite risk except the riskless portfolio which has zero risk.  Whence $\Pi^{\rho}_{0} = \{ \mathbf{0} \}$, $\Pi^{\rho}_{1}=\emptyset$ and $\rho_1 = \infty$. Now part (a) follow from Theorem \ref{thm:reg arb:first characterisation}(b).

If $\rho(Z) = 0$, $\rho(X_\pi) = -\mathbb{E}[X_\pi]$ for every portfolio $\pi \in \mathbb{R}^d$.  Thus, $\rho_\nu=-\nu$ for any $\nu \geq 0$ and the market admits strong $\rho$-arbitrage by Theorem \ref{thm:strong reg arb first characterisation}.

Finally, if $\rho(Z)< 0$, Lemma \ref{lemma:risk of elliptical dist} gives for $\nu \geq 0$,
\begin{equation*}  \rho_{\nu} = \inf_{\pi \in \Pi_{\nu}} \rho(X_{\pi}) = \inf_{\pi \in \Pi_{\nu}} \{-\nu + \rho(Z) \sqrt{\Var(X_{\pi})} \} = -\nu + \rho(Z) \sup_{\pi \in \Pi_{\nu}} \sqrt{\Var(X_{\pi})} < 0,
\end{equation*}
whence, the market admits strong $\rho$-arbitrage by Theorem \ref{thm:strong reg arb first characterisation}. If $d \geq 2$, it is not difficult to check that $\sup_{\pi \in \Pi_{\nu}} \sqrt{\Var(X_{\pi})} = \infty$, and hence $\rho_{\nu} = -\infty$, which implies that $\Pi^\rho_{\nu} = \emptyset$.
\end{proof}

\begin{remark}
\label{rem:elliptical reg arbitrage}
Corollary \ref{cor:SR} shows that in general it is \emph{not} true that for elliptically distributed returns and a law-invariant  risk measure $\rho$, the $\rho$-optimal portfolios coincide with the Markowitz optimal portfolios.\footnote{ This is for instance claimed in \cite[Theorem 1]{embrechts2002correlation}.} Indeed, Corollary \ref{cor:SR} shows that in \emph{every} elliptical market, $\VaR^\alpha$-optimal portfolios fail to exist if 
$\alpha >   \P[Z \leq 0] = 1/2 + 1/2 \P[Z = 0]$, where $Z \sim E_{1}\left(0, 1, \psi\right)$.\footnote{Note that $Z \sim E_{1}\left(0, 1, \psi\right)$ has a symmetric distribution.} In particular, $\VaR^\alpha$-optimal portfolios fail to exist for $\alpha > 1/2$ in every multivariate Gaussian market. The underlying reason is that Value at Risk fails to be expectation bounded.
\end{remark}

We illustrate the above result by considering the case that $R$ has multivariate Gaussian returns and the risk measure is either Value at Risk or Expected Shortfall.

\begin{example}
\label{exa:jointly gaussian returns ES and VaR}
Assume the return vector $R$ has a multivariate normal distribution with mean vector $\mu \in \RR^d$ satisfying $\mu \neq r\1$ and a positive definite covariance matrix $\Sigma \in \RR^{d \times d}$.  Let $Z \sim N(0,1)$.  Then for $\alpha \in (0,1)$, we have
\begin{equation*}
    \textnormal{VaR}^{\alpha}(Z) = \Phi^{-1}(1-\alpha) \quad \textnormal{and} \quad \textnormal{ES}^{\alpha}(Z) = \frac{\phi(\Phi^{-1}(\alpha))}{\alpha},
\end{equation*}
where $\phi$ and $\Phi$ denote the pdf and cdf of a standard normal distribution, respectively.  By Corollary \ref{cor:SR}, we can fully characterise (strong) $\rho$-arbitrage in this market for both risk measures by looking at the maximal Sharpe ratio.  Figure \ref{phase transition picture} gives a graphical illustration. 

\begin{figure}[H]
	\centering
	\includegraphics[width=0.65\textwidth]{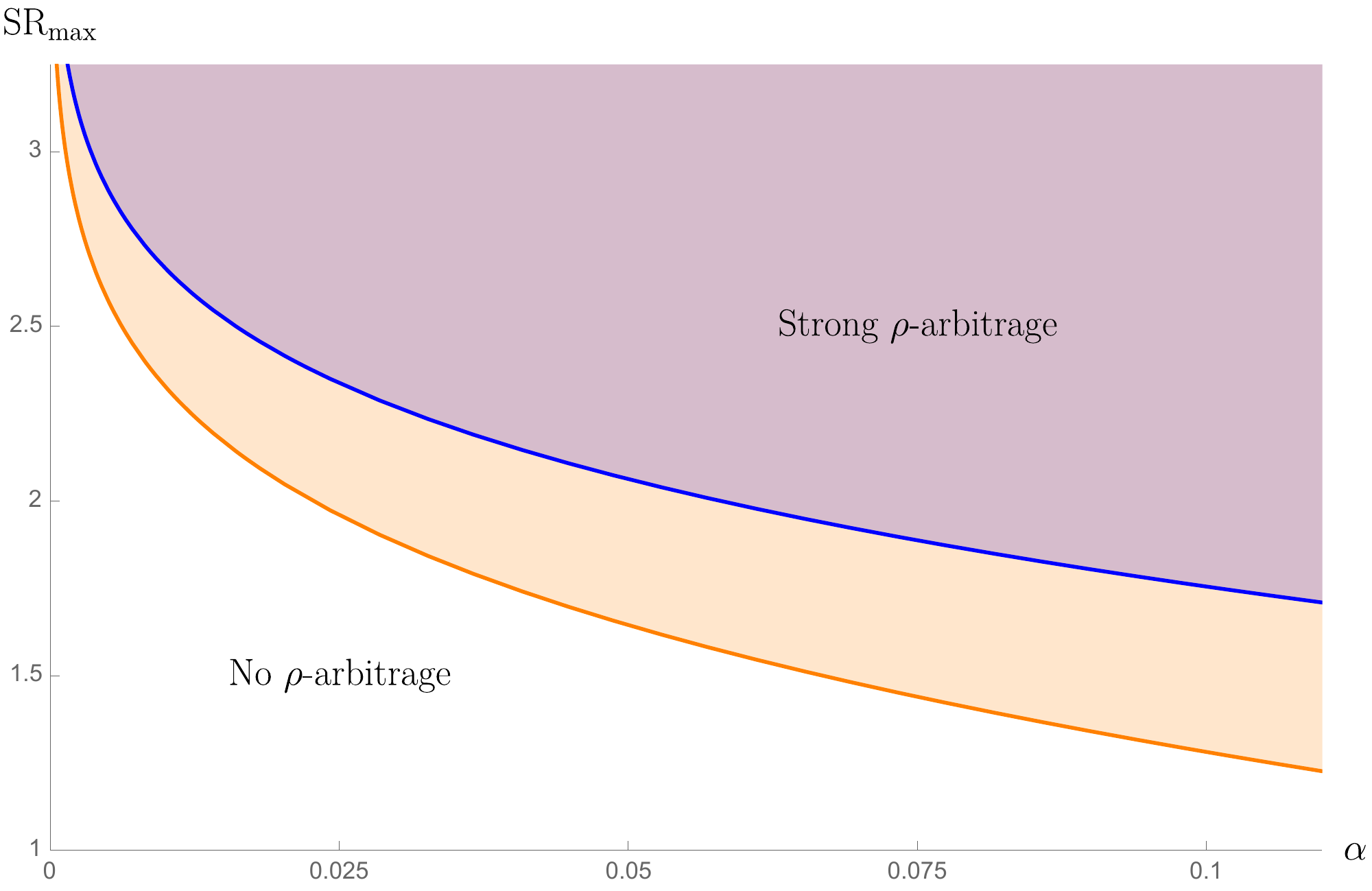}
	\caption{$\rho$-arbitrage for ES (blue) and VaR (orange), for multivariate normal returns}
	\label{phase transition picture}
\end{figure}

If $\textnormal{SR}_{\textnormal{max}}$ lies above the blue (orange) curve, then this Gaussian market admits strong $\ES^\alpha$($\textnormal{VaR}^\alpha$)-arbitrage.  If it lies below the blue (orange) curve then the market does not admit $\ES^\alpha$($\textnormal{VaR}^\alpha$)-arbitrage.  And in the intermediate case, the market admits $\ES^\alpha$($\VaR^\alpha$)-arbitrage, but not strong $\ES^\alpha$($\VaR^\alpha$)-arbitrage.

Also note that for  Value at Risk, if $\alpha > 1/2$, then $\Phi^{-1}(1-\alpha) < 0$. Hence, in this case we always have strong $\VaR^\alpha$-arbitrage and  $\textnormal{VaR}^\alpha$-optimal portfolios fail to exist for $d \geq 2$, independent of $\mu$ or $\Sigma$. 

\end{example}

\section{Dual characterisation of (strong) $\rho$-arbitrage}
\label{section:dual characterisations}

Theorems  \ref{thm:strong reg arb first characterisation} and \ref{thm:reg arb:first characterisation} provide a full characterisation of strong $\rho$-arbitrage and $\rho$-arbitrage, respectively. However, the criterion is rather indirect as it requires to calculate $\rho_1$, which relies on a nontrivial optimisation problem. In this section, we consider the case that $\rho$ is in addition convex (and hence coherent), expectation bounded and has a dual representation. We then derive a \emph{dual characterisation} of (strong) $\rho$-arbitrage.

Let $\mathcal{D} := \{Z \in L^1 : Z \geq 0 \ \mathbb{P}\textnormal{-a.s.} \textnormal{ and }  \mathbb{E}[Z]=1  \}$ be the set of all Radon-Nikod{\'y}m derivatives of probability measures that are absolutely continuous with respect to $\P$. Throughout this section, we assume that $\rho: L \xrightarrow{} (-\infty, \infty]$ is an expectation bounded, coherent risk measure and admits a dual representation
\begin{equation} \label{dual_char_of_rho}
    \rho(X) = \sup_{Z \in \mathcal{Q}} (\mathbb{E}[-ZX]),
\end{equation}
for some $\mathcal{Q} \subset \mathcal{D}$. Since $\rho$ is expectation bounded, we may assume without loss of generality that $1 \in \cQ$. Moreover, 
taking the supremum over $\mathcal{Q}$ is equivalent to taking the supremum over its convex hull, and therefore, we may assume without loss of generality that $\cQ$ is convex.

\begin{remark}
	\label{rem:dual char}
(a) Since $-ZX$ may not be integrable, we define $\mathbb{E}[-ZX] := \mathbb{E}[ZX^-] - \mathbb{E}[ZX^+]$, with the conservative convention that if $\mathbb{E}[ZX^{-}] = \infty$, then $\mathbb{E}[-ZX]=\infty$. 

(b) Apart from the (natural) assumption that $\rho$ is expectation bounded, this is the most general class of coherent risk measures on $L$ that admit a dual representation. For instance, we do not impose $L^1$-closedness or uniformly integrability of $\mathcal{Q}$ (which is for instance assumed in \cite{cherny2008pricing}). A wide range of examples of risk measures satisfying \eqref{dual_char_of_rho} are given in Section \ref{sec:Applications}.

(c) The representation in \eqref{dual_char_of_rho} is not unique. However, it is not difficult to check that the \emph{maximal} dual set for which \eqref{dual_char_of_rho} is satisfied is given by
\begin{equation}
	\label{eq:Q rho}
\cQ_\rho := \{Z \in \cD: \mathbb{E}[Z X] \geq 0 \text{ and } \mathbb{E}[Z X^-] < \infty \text{ for all } X \in \cA_\rho\},
\end{equation}
where $\cA_\rho : = \{X\in L : \rho(X) \leq 0\}$ is the \emph{acceptance set} of $\rho$.\footnote{Note that in general $\cQ_{\rho}$ is not $L^1$-closed.} However, it turns out that for the dual characterisation of $\rho$-arbitrage it is sometimes useful \emph{not} to consider the maximal dual set; cf.~some of the examples in Section \ref{sec:Applications}.

(d) If we define $\rho$ by \eqref{dual_char_of_rho} for some convex set $\cQ$ containing $1$, it follows that $\rho$ is $(-\infty, \infty]$-valued, expectation bounded and a coherent risk measure (i.e., it is monotone, cash-invariant, positively homogeneous and convex).
\end{remark}

\subsection{Preliminary considerations and conditions}
In this section, we introduce and discuss some additional conditions that are needed (and necessary) for our main results, Theorems \ref{thm: no strong reg arb} and \ref{thm: no reg arb equivalence}.

We start by introducing two conditions concerning the (uniform) integrability of the returns under the probability measures “contained” in the dual set $\cQ$.

\medskip
\noindent\textbf{Condition I.}~For all $i \in \{1,\dots,d \}$ and any $Z \in \mathcal{Q}$, $Z R^i \in L^1$.
\medskip

\noindent\textbf{Condition UI.}~$\mathcal{Q}$ is uniformly integrable, and for all $i \in \{1,\dots,d \}$, $R^i\mathcal{Q}$ is uniformly integrable, where $R^i\mathcal{Q} := \{ R^iZ : Z \in \mathcal{Q} \}.$

\begin{remark}
(a) Condition I may depend on the choice of the dual set $\cQ$ in the dual representation \eqref{dual_char_of_rho} of $\rho$. In particular, it may not be satisfied for the maximal dual set $\cQ_{\rho}$; cf.~Section \ref{subsec:worst case risk} for a concrete example. For this reason, one might want to choose a “small” dual set $\cQ$ for~$\rho$.

(b) By contrast, Condition UI \emph{essentially} does not depend on the choice of the dual set $\cQ$ in the dual representation of $\rho$. More precisely, this statement is true if $\rho$ is such that all representing dual sets have the same $L^1$-closure. One important example is when $L$ is an Orlicz space and $\rho$ is real valued; cf.~Proposition \ref{prop:orlicz dual}.
\end{remark}

While Condition I is quite weak, it has some important consequences.

 \begin{proposition}
	\label{prop:cond I}
	Suppose that Condition I is satisfied. Then the set 
	\begin{equation}
	\label{eq: C Q}
	C_\mathcal{Q} := \{\mathbb{E}[-Z(R-r\mathbf{1})] : Z \in \mathcal{Q} \}
	\end{equation}
	is a convex subset of $\mathbb{R}^d$ and for any portfolio $\pi \in \RR^d$,
	\begin{equation}
	\label{relating Q to C}
	\rho(X_{\pi}) = \sup_{c \in C_{\mathcal{Q}}} (\pi \cdot c).
	\end{equation}
	Moreover, $\rho$ satisfies the Fatou property on  $\cX=\{X_\pi:\pi \in \RR^d\}$.
\end{proposition}

\begin{proof}
The set $C_\mathcal{Q}$ is real valued by Condition I and convex by convexity of $\cQ$. This together with linearity of the expectation implies that
\begin{equation*}
\rho(X_{\pi}) = \sup_{Z \in \mathcal{Q}} (\mathbb{E}[-ZX_\pi]) =\sup_{Z \in \mathcal{Q}} (\mathbb{E}[-Z(\pi \cdot (R - r \1))]) = \sup_{c \in C_{\mathcal{Q}}} (\pi \cdot c).
\end{equation*}
Finally, to establish the Fatou property on  $\cX$, assume that $X_{\pi_n} \to X_\pi$ $\as{\P}$ Nondegeneracy of the market implies that $\pi_n \to \pi$. Then for any $Z \in \cQ$, Condition I, linearity of the expectation and the definition of $\rho$ in \eqref{dual_char_of_rho} gives
\begin{equation*}
\mathbb{E}[-Z X_\pi] = \pi \cdot \mathbb{E}[-Z(R-r\mathbf{1})] = \lim_{n \to \infty} \pi_n \cdot \mathbb{E}[-Z(R-r\mathbf{1})] = \lim_{n \to \infty} \mathbb{E}[-Z X_{\pi_{n}}] \leq \liminf_{n \to \infty} \rho(X_{\pi_{n}}).
\end{equation*}
Taking the supremum over $Z \in \mathcal{Q}$ gives $\rho(X_{\pi}) \leq \liminf_{n \to \infty} \rho(X_{\pi_{n}})$.
\end{proof}

\begin{remark}
Example \ref{example: weird example} shows that without Condition I, the set $C_\cQ$ may fail to be convex or $\RR^d$-valued and \eqref{relating Q to C} may break down.
\end{remark}

  Condition UI  is a uniform version of Condition I.\footnote{Note that $Z \in L^1$ for all $Z \in \cQ$ even though this does not appear explicitly in Condition I.} The following result shows that under Condition UI, the supremum in \eqref{relating Q to C} can be replaced by a maximum, if we replace $\cQ$ in \eqref{eq: C Q} by its $L^1$-closure.
  
 \begin{proposition}
 	\label{prop:cond 1b}
  Suppose that Condition UI is satisfied. Denote by $\ol\cQ$ the $L^1$-closure of $\cQ$. Then the set 
  \begin{equation}
  \label{eq: C bar Q}
  C_{\ol\cQ} := \{\mathbb{E}[-Z(R-r\mathbf{1})] : Z \in \ol \cQ \}
  \end{equation}
  is a convex and compact subset of  $\mathbb{R}^d$. Moreover, 
  for any portfolio $\pi \in \RR^d$,
  \begin{equation}
  \label{relating Q bar to C}
  \rho(X_{\pi}) = \max_{c \in C_{\ol\cQ}} (\pi \cdot c).
  \end{equation}
 \end{proposition}

\begin{proof}
	Since Condition UI implies Condition I, \eqref{relating Q to C} gives
	\begin{equation}
	\label{eq: sup over Q leq sup over Q bar}
	\rho(X_\pi) = \sup_{c \in C_{\cQ}} (\pi \cdot c) \leq \sup_{c \in C_{\ol\cQ}} (\pi \cdot c).
	\end{equation}
Since $\cQ$ is UI and convex, $\ol \cQ$ is convex and weakly compact by the Dunford-Pettis theorem. To show that the supremum on the right side of \eqref{eq: sup over Q leq sup over Q bar} is attained, let $(Z_{n})_{n \in \NN}$ be a maximising sequence in $\ol \cQ$. Since $\ol \cQ$ is weak sequentially compact by the Eberlein-\v{S}mulian theorem, after passing to a subsequence, we may assume that $Z_n$ converges weakly to some $Z \in \ol\cQ$. Since the map $\tilde Z \mapsto \E[-\tilde Z(R - r\1)]$ is weakly continuous on $\ol \cQ$ by Proposition \ref{prop: cond 1 implies weak cts map}, $Z$ is a maximiser. The same argument, but now for a maximising sequence in $\cQ \subset \ol \cQ$, shows that we have  have equality in \eqref{eq: sup over Q leq sup over Q bar}. Finally, using again that the map $\tilde Z \mapsto \E[-\tilde Z(R - r\1)]$ is weakly continuous on $\ol \cQ$ and $\ol \cQ$ is weakly compact, it follows that $C_{\ol\cQ}$ is compact.
\end{proof}

\begin{remark}
	Example \ref{example:UI} shows that without Condition UI (even when Condition I is satisfied), the set $C_{\ol \cQ}$ may fail to be convex, compact or a subset of $\RR^d$ and \eqref{relating Q bar to C} may break down.
\end{remark}

\begin{remark}
\label{rmk:Chernys assumptions equivalent to 1b}
  In \cite{cherny2008pricing}, it is assumed that  $\mathcal{Q}$ is uniformly integrable and that $R^i \in L^1(\mathcal{Q})$, where
\begin{equation}
	\label{eq:rem:Cherny}
    L^1(\mathcal{Q}) := \{X \in L^0 : \lim_{a \to \infty} \sup_{Z \in \mathcal{Q}} \mathbb{E}[Z|X|\mathds{1}_{\{{|X|>a}\}}] = 0 \}.
\end{equation}
By Proposition \ref{prop:cond 1b equiv to cherny}, this is equivalent to Condition UI. However, we believe that Condition UI better highlights why this is a uniform version of Condition I.
\end{remark}

We next aim to introduce a notion of “interior” for $\cQ$, which is crucial for the dual characterisation of $\rho$-arbitrage. This turns out to be rather subtle since neither algebraic nor topological notions of interior work in general; cf.~Remark \ref{rem:QL}. Instead, we define our notion of “interior” in an abstract way. More precisely, we look for (nonempty) subsets $\tilde \cQ \subset \cQ$ satisfying

\medskip
\noindent\textbf{Condition POS.}~$\tilde Z > 0$ $\as{\P}$ for all $\tilde Z \in \tilde \cQ$.

\medskip
\noindent\textbf{Condition MIX.}~$\lambda Z + (1-\lambda) \tilde Z \in \tilde \cQ$ for all $Z \in \cQ$, $\tilde Z \in \tilde \cQ$ and $\lambda \in (0, 1)$.
\medskip

\noindent\textbf{Condition INT.}~For all $\tilde Z \in \tilde \cQ$, there is an $L^\infty$-dense subset $\cE$ of $\cD \cap L^\infty$ such that for all $Z \in \cE$,  there is $\lambda \in (0, 1)$ such that $\lambda Z + (1-\lambda) \tilde Z \in \cQ$.
\medskip

A few comments are in order.

\begin{remark}
	\label{rem:QL}
	(a) Condition MIX implies in particular that $\tilde \cQ$ is convex.
	
	(b) Condition INT of $\tilde \cQ$ is inspired by the definition of the core/algebraic interior. Indeed, recall that for a vector space $V$, the algebraic interior of a set $M \subset V$ with respect to a vector subspace $X \subset V$ is defined by\footnote{We refer the reader to \cite{zalinescu:02} for details. The case that $X = V$ is more standard and called the core/algebraic interior of $M$.}
	\begin{equation*}
	\label{eq:core}
	\mathrm{aint}_X M:= \{ m \in M : \text{for all } x \in X,\text{ there is } \lambda > 0 \textnormal{ such that } m +\delta x \in M  \text{ for all } \delta \in [0,\lambda]  \}.
	\end{equation*}
	When $M$ is convex, one can show that 
		\begin{equation*}
	\mathrm{aint}_X M = \{ m \in M : \text{for all } x \in X,\text{ there is } \lambda > 0 \textnormal{ such that  } m +\lambda x \in M \},
	\end{equation*}
	and	any strict convex combination of a point in $M$ and $\mathrm{aint}_X M$  belongs to $\mathrm{aint}_X M$. To see the link to our setup, assume that $\cQ \subset L^\infty$. Set $M := \cQ$,  $V := L^\infty$ and $X := \{Z \in L^\infty : \mathbb{E}[Z] = 0\}$. Then $\mathrm{aint}_{X} M$ satisfies conditions POS, MIX and INT. Moreover, for certain examples (e.g.~Expected Shortfall), $\mathrm{aint}_{X} M \neq \emptyset$. Note, however, that if $\cQ \not\subset L^\infty$, it is not possible to define a nonempty set $\tilde \cQ$ satisfying Conditions POS, MIX and INT via the algebraic interior.
	
	(c) One might wonder if one could define $\tilde \cQ$ as the topological interior of $\mathcal{Q}$ in a suitable subspace topology of $\cD \cap V$, where $L^\infty  \subset V \subset L^1$ is a vector subspace. Again if $\cQ \subset  L^\infty$, for certain examples (e.g.~Expected Shortfall), the topological interior of $\mathcal{Q}$ in the subspace topology of $\cD \cap L^\infty$ is nonempty and satisfies Conditions POS, MIX and INT. However, if $\cQ \not\subset L^\infty$, this approach does not work since the topological interior may fail to satisfy Condition MIX (because $\mathcal{D} \cap V$ is not a vector space).
	
	(d) In light of Propositions \ref{prop:P in tilde Q rho strictly exp bdd} and \ref{prop:app:WC dual}, one could slightly relax Condition INT, by requiring that the sets $\cE$ are only $\sigma(L^\infty, L^1)$-dense in $\cD \cap L^\infty$. However, this additional level of generality does not seem to be useful in concrete examples. On the other hand, considering $L^\infty$-dense subsets of $\cD \cap L^\infty$ is useful; cf.~Section \ref{sec:spectral}.
\end{remark}

We proceed to characterise the maximal subset of $\cQ$ satisfying Conditions POS, MIX and INT. This is surprisingly simple and shows that we can expect $\tilde{\mathcal{Q}}_{\max}$ to be nonempty for most risk measures $\rho$.
\begin{proposition}
\label{prop:tQmax}
Define the set
$\tilde{\mathcal{Q}}_{\max}$ by
\begin{align*}
\tilde{\mathcal{Q}}_{\max} :=\{ \tilde Z > 0 \in \mathcal{Q} :  \,& \text{there is an $L^\infty$-dense subset $\cE$ of $\cD \cap L^\infty$ such that for all } Z \in \cE,\\
& \text{there is } \lambda \in (0,1)  \text{ such that } \lambda Z + (1-\lambda)\tilde Z \in \mathcal{Q} \}. \label{eq:prop:tQmax}
\end{align*}
Then $\tilde{\mathcal{Q}}_{\max}$ satisfies Conditions POS, MIX and INT. Moreover, 	if $\tilde \cQ \subset \cQ$ satisfies Conditions POS, MIX and INT, then $\tilde \cQ \subset \tilde \cQ_{\max}$.
\end{proposition}

\begin{proof}
$\tilde \cQ_{\max}$ satisfies Conditions POS and INT by definition. To establish Condition MIX, let $Z \in \cQ$, $\tilde Z \in \tilde \cQ_{\max}$ and $\mu \in (0, 1)$. Clearly $\mu Z  + (1 -\mu) \tilde Z > 0$ $\as{\P}$ It remains to show that there exists an $L^\infty$-dense subset $\cE'$ of $\cD \cap L^\infty$ such that for all $Z' \in \cE'$, there is $\lambda' > 0$ such that $\lambda' Z' + (1 - \lambda') (\mu Z + (1 -\mu) \tilde Z) \in \cQ$. Let $\cE$ be the $L^\infty$-dense subset of $\cD \cap L^\infty$ for $\tilde Z$ in the definition of $\tilde \cQ_{\max}$. Set $\cE' := \cE$. Let $Z' \in \cE'$. Then there is $\lambda > 0$ such that $\lambda Z' + (1 -\lambda) \tilde Z\in \tilde \cQ_{\max} \subset \cQ$. Set $\mu' := \tfrac{1 -\mu}{1 - \mu \lambda} \in (0,1)$ and $\lambda' := \lambda \mu' \in (0, 1)$. Then by convexity of $\cQ$,
\begin{align*}
\mu' \big(\lambda Z' + (1 -\lambda) \tilde Z\big) + (1 - \mu') Z  = \lambda' Z' + (1 - \lambda') (\mu Z + (1 -\mu) \tilde Z) \in \cQ.
\end{align*}
The additional claim follows immediately from the definition of $\tilde \cQ_{\max}$.
\end{proof}

\begin{remark}
	\label{rem:Qmax}
(a) If $\cQ' \subset \cQ$ are dual sets representing $\rho$, then $\tilde{\mathcal{Q}'}_{\max} \subset \tilde \cQ_{\max}$.

(b) While Proposition \ref{prop:tQmax} is insightful from a theoretical perspective, it is very difficult in practise to compute $\tilde{\mathcal{Q}}_{\max}$. For this reason, it is often easier to find a nonempty subset $\tilde \cQ \in \cQ$ satisfying Conditions POS, MIX and INT directly. This is the approach that we take in virtually all of the examples in Section \ref{sec:Applications}. Since Condition MIX is easier to satisfy if $\cQ$ is smaller, one sometimes might even first have to find a smaller representing dual set $\cQ' \subset \cQ$ for $\rho$ and then a nonempty subset $\tilde \cQ' \in \cQ'$ satisfying Conditions POS, MIX and INT;  see Section \ref{sec:spectral} for a concrete example.
	\end{remark}

We finish this section by explaining the role of Conditions POS and INT for establishing existence of $\rho$-optimal portfolios.

\begin{proposition}
	\label{prop:P in tilde Q rho strictly exp bdd}
Suppose Condition I is satisfied.	Let $\tilde \cQ \subset \cQ$ satisfy Conditions POS and INT.\footnote{Note that $\tilde \cQ$ does not need to satisfy Condition MIX.} If $1\in \tilde{\mathcal{Q}}$, then $\rho$ is strictly expectation bounded and $\Pi^\rho_0 = \{\0\}$. If in addition $\rho_1 < \infty$,  then for all $\nu \geq 0$, $\Pi^\rho_\nu$ is nonempty, compact and convex.
\end{proposition}

\begin{proof}
Strict expectation boundedness of $\rho$ follows from Lemma \ref{lemma: Z tilde and Z} (with $\tilde Z = 1$) and Remark~\ref{rem:exp bound}(c). Corollary \ref{cor:exp bounded} then gives $\Pi^\rho_0 = \{\mathbf{0}\}$. Finally, if $\rho_1 < \infty$, it follows that $\rho_1 \in \RR$ since $\rho_1 \geq -1$ by expectation boundedness of $\rho$.  Now the remaining claim follows from Proposition \ref{prop:cond I}, Theorem \ref{thm:Existence of optimal portfolios}  and Remark \ref{rmk:existence of rho optimal portfolios}(c).
\end{proof}

\subsection{Dual characterisation of strong $\rho$-arbitrage}

In this section, we provide a dual characterisation of strong $\rho$-arbitrage in terms of absolutely continuous martingale measures (ACMMs) for the discounted risky assets $S/S^0$. To this end, set
\begin{equation}
\label{eq:cM}
\mathcal{M} = \{Z \in \mathcal{D} : \mathbb{E}[Z(R^{i}-r)] =0 \textnormal{ for all } i=1,\dots ,d \},
\end{equation}
and note that each $Z \in \cM$ is the Radon-Nikod{\'y}m derivative of an ACMM for $S/S^0$.

A first step towards a dual characterisation is the following equivalent characterisation of strong $\rho$-arbitrage.
\begin{proposition}
\label{prop:reformulation of strong reg arb}
The market $(S^0,S)$ satisfies strong $\rho$-arbitrage if and only if $\rho(X_\pi) < 0$ for some portfolio $\pi \in \mathbb{R}^d$.
\end{proposition}

\begin{proof}
If the market admits strong $\rho$-arbitrage, then $\rho_1<0$ by Theorem \ref{thm:strong reg arb first characterisation}. Hence, $\rho(X_\pi)<0$ for some portfolio $\pi \in \mathbb{R}^d$.

Conversely, if $\rho(X_{\pi}) < 0$ for some portfolio $\pi$, $\mathbb{E}[X_\pi] \geq -\rho(X_\pi) >0$ because $1\in \mathcal{Q}$. Thus, $\rho_1 < 0$, and the market satisfies strong $\rho$-arbitrage.
\end{proof}

\begin{remark}
	\label{rem:good deal}
The condition $\rho(X_\pi) < 0$ for some $\pi \in \RR^d$ is referred to as a \emph{Good Deal} in the literature, see e.g.~\cite{cherny2008pricing}.  Note, however, that the equivalence of Proposition \ref{prop:reformulation of strong reg arb} crucially relies on $\rho$ being expectation bounded (via $1 \in \cQ$) since otherwise a portfolio with negative risk may have a negative expected excess return.\footnote{Note that Proposition \ref{prop:reformulation of strong reg arb} holds more generally for positively homogeneous (not necessarily convex) risk measures that are expectation bounded.
Also note that assuming that $\rho$ is expectation bounded is a real restriction as it is not satisfied by Value at Risk.}
\end{remark}
	
Our next result shows that if $\mathcal{Q}$ contains an ACMM, the market does not admit strong $\rho$-arbitrage.  

\begin{proposition}
\label{Q intersect M nonempty}
If $\mathcal{Q} \cap \mathcal{M} \neq \emptyset$, then the market $(S^0,S)$ does not admit strong $\rho$-arbitrage.
\end{proposition}

\begin{proof}
Let $Z \in \mathcal{Q} \cap \mathcal{M}$.  Then for any portfolio $\pi \in \mathbb{R}^d$,
\begin{equation*}
    \rho(X_\pi) \geq \mathbb{E}[-ZX_\pi] = 0.
\end{equation*}
Therefore, by Proposition \ref{prop:reformulation of strong reg arb} the market does not admit strong $\rho$-arbitrage.
\end{proof}

The converse of Proposition \ref{Q intersect M nonempty} is false.  Example \ref{counterexample cond 1a not enough} shows that even under Condition I, $\mathcal{Q} \cap \mathcal{M} = \emptyset$ is not enough to imply strong $\rho$-arbitrage.  However, under condition UI, the converse of Proposition \ref{Q intersect M nonempty} is essentially true.

\begin{theorem}
\label{thm: no strong reg arb}
Assume $\mathcal{Q}$ satisfies UI.  Denote by $\ol\cQ$ the $L^1$-closure of $\cQ$.  The following are equivalent:
\begin{enumerate}
    \item The market $(S^0,S)$ does not admit strong $\rho$-arbitrage.
    \item $\ol\cQ \cap \mathcal{M} \neq \emptyset$.
\end{enumerate}
\end{theorem}

\begin{proof}
First, assume $\ol\cQ \cap \mathcal{M} \neq \emptyset$. Let $Z \in \ol\cQ \cap \mathcal{M}$.  Then Proposition \ref{prop:cond 1b} gives $\rho(X_\pi) \geq \mathbb{E}[-ZX_\pi] = 0$ for any $\pi \in \mathbb{R}^d$. Therefore, the market does not admit strong $\rho$-arbitrage by Proposition \ref{prop:reformulation of strong reg arb}.

Conversely, assume $\ol\cQ \cap \mathcal{M} = \emptyset$.  By Proposition \ref{prop:cond 1b}, $\{ \mathbf{0} \}$ and $C_{\ol\cQ}$ are two nonempty disjoint convex and compact subsets of $\mathbb{R}^d$.  By the strict separation theorem (cf.~\cite[Proposition B.14]{bertsekas1995nonlinear}), there exists $\pi \in \mathbb{R}^{d} \setminus \{\0\}$ with $\pi \cdot c < b < 0$ for all $c \in C_{\ol\cQ}$.  Thus,  Proposition \ref{prop:cond 1b} gives
\begin{equation*}
    \rho(X_\pi) = \max_{c \in C_{\ol\cQ}} (\pi \cdot c)  < 0,
\end{equation*}
and so the market admits strong $\rho$-arbitrage by Proposition \ref{prop:reformulation of strong reg arb}.
\end{proof}

\begin{remark}
(a) By virtue of  Proposition \ref{prop:reformulation of strong reg arb}, Theorem \ref{thm: no strong reg arb} is identical to Cherny's equivalent characterisation of No Good Deals in \cite[Theorem 3.1]{cherny2008pricing}.  However, our proof is simpler since we are working with a finite number of assets. We have included it for the convenience of the reader.

(b) Example \ref{counterexample cond 1a and UI not enough} shows that when $\mathcal{Q}$ is uniformly integrable but $R\mathcal{Q}$ is not, then Theorem \ref{thm: no strong reg arb}  is false.  Example \ref{counterexample strong reg arb Q not UI} shows that Theorem \ref{thm: no strong reg arb} is also false if $R\mathcal{Q}$ is uniformly integrable but $\mathcal{Q}$ is not.  Thus, we need both parts of Condition UI simultaneously.
\end{remark}

Characterising the absence of strong $\rho$-arbitrage is important.  However, it is not enough as the risk constraint also remains void if there is no portfolio with negative risk but a portfolio $\pi \in \Pi_{1}$ with zero risk. This is illustrated by the following example.

\begin{example}
\label{exa:difference between SRA and RA}
Consider a binomial model with one riskless asset and one risky asset with returns $r=0$ and $R$, respectively, where
$\mathbb{P}[R=1] = \mathbb{P}[R=0] = 1/2$.  Set $\pi := 2$. Then for $\alpha \in (0, 1)$,  $\Pi^{\ES^\alpha}_1 = \{\pi\}$ and 
\begin{equation*}
\ES^\alpha _1  = \textnormal{ES}^\alpha(X_{\pi}) = \begin{cases}
    0,&\text{if } \alpha \leq \tfrac{1}{2} \\ \tfrac{1-2\alpha}{\alpha}<0,&\text{if } \alpha > \tfrac{1}{2}.
    \end{cases}
\end{equation*}
Hence, by either Theorem \ref{thm:strong reg arb first characterisation} or Theorem \ref{thm: no strong reg arb}, the market does not admit strong $\ES^\alpha$-arbitrage if $\alpha \leq \tfrac{1}{2}$. However, by Theorem \ref{thm:reg arb:first characterisation}(a),  
the market does admit $\ES^\alpha$-arbitrage if $\alpha \leq \tfrac{1}{2}$. 
\end{example}

Therefore, to see the whole picture, it is important to also have a dual characterisation of $\rho$-arbitrage.

\subsection{Dual characterisation of $\rho$-arbitrage} 
In this section, we provide a dual characterisation of $\rho$-arbitrage in terms of equivalent martingale measures (EMMs) for the discounted risky assets $S/S^0$. To this end, set
\begin{equation*}
\mathcal{P}  = \{Z \in \mathcal{M} :  Z > 0 \ \mathbb{P}\textnormal{-a.s.} \}.
\end{equation*}
and note that each $Z \in \cP$ is the Radon-Nikod{\'y}m derivative of an EMM for $S/S^0$.

As we did for strong $\rho$-arbitrage, we start by providing an equivalent characterisation of $\rho$-arbitrage.  However, for $\rho$-arbitrage, we need to assume that $\0$ is the unique $\rho$-optimal portfolio.

\begin{proposition}
\label{prop:reformulation of reg arb}
Assume $\Pi^{\rho}_{0} = \{\mathbf{0}\}$.  Then the market $(S^0,S)$ satisfies $\rho$-arbitrage if and only if $\rho(X_{\pi}) \leq 0$ for some portfolio $\pi \in \mathbb{R}^d \setminus \{\mathbf{0}\}$.
\end{proposition}

\begin{proof}
First assume the market satisfies $\rho$-arbitrage. As the riskless portfolio $\mathbf{0}$ has zero risk, by definition of $\rho$-arbitrage there must be another portfolio $\pi \in \mathbb{R}^d \setminus \{\mathbf{0}\}$ with $\rho(X_{\pi}) \leq 0$.

Conversely, if $\rho(X_{\pi}) \leq 0$ for some portfolio $\pi \in \mathbb{R}^d \setminus \{ \mathbf{0} \}$, then $\mathbb{E}[X_\pi]\geq 0$ by expectation boundedness of $\rho$, which in turn gives $\mathbb{E}[X_\pi] > 0$ because $\Pi^{\rho}_{0} = \{\mathbf{0}\}$.  It follows that either $\rho_1 < 0$ (in which case $\Pi^\rho_1$ may or may not be empty) or $\rho_1 =0$ (in which case $\Pi^\rho_1 \neq \emptyset$).  In either case the market admits $\rho$-arbitrage by Theorem \ref{thm:reg arb:first characterisation}.
\end{proof}

We proceed to give a preliminary dual characterisation of $\rho$-arbitrage. Note that this characterisation does not rely on the set $\tilde \cQ_{\max}$ to be nonempty.

\begin{proposition}
\label{Q intersect M empty}
Assume $\Pi^{\rho}_0 = \{\mathbf{0}\}$ and $\mathcal{Q}$ satisfies Condition I.  If $\mathcal{Q} \cap \mathcal{M} = \emptyset$, then the market $(S^0,S)$ admits $\rho$-arbitrage.
\end{proposition}

\begin{proof}
Condition I implies that the set $C_\cQ$ in \eqref{eq: C Q} is convex. If $\mathcal{Q} \cap \mathcal{M} = \emptyset$ then $\mathbf{0} \notin C_{\mathcal{Q}}$.  By the supporting hyperplane theorem (cf.~\cite[Proposition B.12]{bertsekas1995nonlinear}), there exists $\pi \in \mathbb{R}^{d} \setminus \{\0\}$ with $\pi \cdot c \leq 0$ for all $c \in C_{\mathcal{Q}}$.  By \eqref{relating Q to C},
\begin{equation*}
    \rho(X_\pi) = \sup_{c \in C_{\mathcal{Q}}}(\pi \cdot c)  \leq 0,
\end{equation*}
and the claim follows from Proposition \ref{prop:reformulation of reg arb}.
\end{proof}

We are now in a position to state and prove the main result of this paper, the dual characterisation of $\rho$-arbitrage.

\begin{theorem}
\label{thm: no reg arb equivalence}
Suppose $\mathcal{Q}$ satisfies Condition, $\Pi^{\rho}_0 = \{\mathbf{0}\}$, and $\tilde \cQ_{\max} \neq \emptyset$. Then the following are equivalent:
\begin{enumerate}
    \item The market $(S^0,S)$ does not admit $\rho$-arbitrage.
    
    \item $\tilde{\mathcal{Q}} \cap \cP \neq \emptyset$ for some $\emptyset \neq \tilde \cQ \subset \cQ$ satisfying Conditions POS, MIX and INT.
        \item $\tilde{\mathcal{Q}} \cap \cP \neq \emptyset$ for all $\emptyset \neq \tilde \cQ \subset \cQ$ satisfying Conditions POS, MIX and INT.
\end{enumerate}
\end{theorem}

\begin{proof}
$(b) \implies (a)$. Let $\emptyset \neq \tilde \cQ \subset \cQ$ satisfying Conditions POS, MIX and INT and $\pi \in \RR^d \setminus \{\0\}$.  By Proposition \ref{prop:reformulation of reg arb}, we have to show that $\rho(X_\pi) > 0$. Let $\tilde{Z} \in \tilde{\mathcal{Q}} \cap \mathcal{P}$. Then $\mathbb{E}[-\tilde{Z}X_\pi]=0$. Since $X_\pi \neq 0$ by nonredundancy of the market, this implies that $X_\pi$ is a non-constant random variable. Now the claim follows from Lemma \ref{lemma: Z tilde and Z}.

$(a) \implies (c)$. We argue by contraposition. So assume that there exists $\emptyset \neq \tilde \cQ \subset \cQ$ satisfying Conditions POS, MIX and INT such that 
$\tilde{\mathcal{Q}} \cap \mathcal{P} = \emptyset$.  This implies that $\tilde{\mathcal{Q}} \cap \mathcal{M} = \emptyset$ by Condition POS. Refining the argument of Proposition \ref{Q intersect M empty}, it suffices to show that $\0$ is not in the interior of $C_\cQ$. Seeking a contradiction, assume that $\mathbf{0} \in C_{\mathcal{Q}}^\mathrm{o}$. Then there is $\epsilon > 0$ such that $B(\mathbf{0},\epsilon) \subset \cQ$, where $B(\mathbf{0},\epsilon)$ 
denotes the open ball of of radius $\epsilon>0$ around $\0$ with respect to some norm $\Vert \cdot \Vert$. Set
\begin{equation*}
    C_{\tilde{\mathcal{Q}}} := \{ \mathbb{E}[-Z(R-r\mathbf{1})] : Z \in \tilde{\mathcal{Q}} \} \subset C_{\mathcal{Q}} \subset \mathbb{R}^d.
\end{equation*}
Then  $C_{\tilde{\mathcal{Q}}}$ is convex by Remark \ref{rem:QL}(a) and does not contain the origin because $\tilde{\mathcal{Q}} \cap \mathcal{M} = \emptyset$. Hence, $B(\mathbf{0},\epsilon) \not\subset C_{\tilde{\mathcal{Q}}}$.  As $\tilde{\mathcal{Q}} \neq \emptyset$, there is $\mathbf{x} \in C_{\tilde{\mathcal{Q}}}$. Set $\mathbf{y} := - \epsilon/(2\lVert \mathbf{x} \rVert) \mathbf{x} \in B(\mathbf{0},\epsilon) $. Then $\lambda \mathbf{x} + (1-\lambda)\mathbf{y} = \mathbf{0}$ for $\lambda:= \epsilon/(2\lVert \mathbf{x} \rVert + \epsilon)$. Letting $Z_{\mathbf{x}} \in \tilde \cQ$ and $Z_{\mathbf{y}} \in \cQ$ denote Radon-Nikod{\'y}m derivatives corresponding to $\mathbf{x}$ and $\mathbf{y}$, respectively, it follows from definition of $\cM$ in \eqref{eq:cM} and Condition MIX that $\lambda Z_{\mathbf{x}} + (1-\lambda) Z_{\mathbf{y}} \in \tilde{\mathcal{Q}} \cap \mathcal{M}$,  in contradiction to $\tilde{\mathcal{Q}} \cap \mathcal{M} = \emptyset$.

$(c) \implies (b)$. This is trivial.
\end{proof}

\begin{remark}
(a) While $\tilde{\mathcal{Q}}_{\max} \neq \emptyset$ is the minimal theoretical condition for Theorem \ref{thm: no reg arb equivalence} to hold (see Example \ref{counterexample empty interior} for a counterexample if is not satisfied), it is difficult to check in practise since we rarely can compute $\tilde{\mathcal{Q}}_{\max}$; cf.~Remark \ref{rem:Qmax}(b). Instead, it is easier (and of course sufficient) to check that $\tilde{\mathcal{Q}} \neq \emptyset$ for some $\tilde \cQ \subset \cQ$ satisfying Conditions POS, MIX and INT. In all our examples, the latter is done by showing that $1\in \tilde{\mathcal{Q}}$, which by Proposition \ref{prop:P in tilde Q rho strictly exp bdd} also implies that $\Pi^{\rho}_0 = \{\mathbf{0}\}$.

(b) If we choose for $\rho$ the worst-case risk measure, we recover a \emph{refined version} of the fundamental theorem of asset pricing in a one-period model; see Theorem \ref{thm:FTAP} below for details. In this case, the proof is particularly simple. To the best of our knowledge, the argument (for the nontrivial direction) is  new, even simpler than any of the existing proofs (cf.~e.g.~\cite[Theorem 1.7]{follmerschied:2016}) and yields a much sharper result.

\end{remark}

\section{Examples}
\label{sec:Applications}
In this section, we apply our main results to various examples of risk measures. Recall that we have already investigated the case of elliptically distributed returns in Section \ref{sec:elliptical}. Here, we do not make any assumptions on the returns, other than our standing assumptions that returns are in $L^1$ (or in some Orlicz space $L^\Phi$) and that the market $(S^0,S)$ is nonredundant and nondegenerate.

\subsection{Worst-case risk measure}
\label{subsec:worst case risk}
We start our discussion by looking at the worst-case risk measure $\mathrm{WC}: L^1 \to (-\infty, \infty]$ given by $\textnormal{WC}(X) := \esssup(-X)$.  It is a coherent risk measure and admits a dual representation with maximal dual set $\cQ_{\rho} = \cD$. However, if the returns do not lie in $L^\infty$, Condition I is not satisfied. Therefore, we look for a smaller dual set, and it turns out that a good choice is
$\mathcal{Q} := \mathcal{D} \cap L^\infty$; see Proposition \ref{prop:app:WC dual}.  Using this $\cQ$, Condition I is always satisfied. By contrast, Condition UI is never satisfied unless $\Omega$ is finite. It is not difficult to check that $\tilde \cQ = \{Z \in \mathcal{D} \cap L^\infty: Z > 0 \; \as{\P}\}$ satisfies conditions POS, MIX and INT. However, it turns out that we get a stronger dual characterisation of $\mathrm{WC}$-arbitrage if we consider the set 
\begin{equation*}
\hat\cQ := \{Z \in \mathcal{D} \cap L^\infty: Z \geq\epsilon \; \as{\P} \text{ for some } \epsilon > 0\},
\end{equation*}
which also satisfies Conditions POS, MIX and INT. Since $1 \in \hat\cQ$, it follows from Proposition \ref{prop:P in tilde Q rho strictly exp bdd} that $\Pi^{\mathrm{WC}}_0 = \{\0\}$. Theorems \ref{thm: no reg arb equivalence} and \ref{thm: no strong reg arb} now give the following result.

\begin{corollary}
\label{cor:WC1}
The market $(S^0, S)$ does not admit \textnormal{WC}-arbitrage if and only if there is $Z \in \cP \cap L^\infty$ with $Z \geq \epsilon$ $\as{\P}$ for some $\epsilon > 0$. Moreover, if $\Omega$ is finite, the market $(S^0, S)$ does not admit strong \textnormal{WC}-arbitrage if and only if $\cM \neq \emptyset$.
\end{corollary}

Combining Proposition \ref{prop:WC arbitrage equivalent to ordinary arbitrage} with Corollary \ref{cor:WC1} gives a \emph{refined} version of the one-period fundamental theorem of asset pricing for $L^1$-markets (with trivial initial information). The refinement is that we show the existence of an EMM with a \emph{positive lower bound}.
\begin{theorem}
	\label{thm:FTAP}
	Suppose that the market $(S^0,S)$ has finite first moments.\footnote{Note that if $(S^0,S)$ has finite first moments, we may assume without loss of generality that it is nonredundant, nondegenerate and satisfies $S^i_0  > 0$ for all $i \in \{1, \ldots, d\}$.}
	\begin{enumerate}
		\item The market does not admit arbitrage of the first kind if and only if there exists $Z \in \cP \cap L^\infty$ with $Z \geq \epsilon$ $\as{\P}$ for some $\epsilon > 0$. 
		\item If $\Omega$ is finite, the market does not admit arbitrage of the second kind if and only if $\cM \neq \emptyset$.
	\end{enumerate}
\end{theorem}

\begin{remark}
To the best of our knowledge, a simple proof for the existence of an EMM with positive lower bound for arbitrage-free $L^1$-markets (with trivial initial information) has not been given before. In fact, the only extant result that we are aware of that gives this lower bound for  $L^1$-markets is \cite[Corollary 2]{rokhlin:schachermayer:06}, which uses very heavy machinery from functional analysis.\footnote{Under stronger integrability conditions on the market, the result has also been established by \cite[Remark~7.5]{rasonyi:stettner:05}.} By contrast our proof is elementary and short, and might even be given in a classroom setting.
\end{remark}

\subsection{Value at Risk and Expected Shortfall}
\label{subsec:VaR and ES}
We have already introduced VaR and ES in Definition \ref{defn:VaR and ES}.  Since VaR has no dual representation, we cannot apply the results from Section \ref{section:dual characterisations}.  However, using the inequality $\textnormal{VaR}^\alpha(X) \leq \textnormal{ES}^\alpha(X)$ for $\alpha \in (0, 1)$ and $X \in L^1$, it follows that if there is (strong) $\ES^\alpha$-arbitrage, then there is (strong) $\VaR^\alpha$-arbitrage.

Unlike VaR, ES is coherent and admits for $\alpha \in (0,1)$ the following dual representation:\footnote{See e.g.~\cite[Theorem 4.47]{follmerschied:2016} (which extends to the case $X \in L^1$).}
\begin{equation}
\label{Eq:ES dual rep}
\textnormal{ES}^{\alpha}(X) = \sup_{Z \in \mathcal{Q^{\alpha}}} \mathbb{E}[-ZX]  = \max_{Z \in \mathcal{Q^{\alpha}}} \mathbb{E}[-ZX], \quad \textnormal{where} \quad
\mathcal{Q}^{\alpha} := 
\{ Z \in \mathcal{D} : \lVert Z \rVert_{\infty} \leq \tfrac{1}{\alpha} \}.
\end{equation}
This can be extended to include $\alpha \in \{0,1\}$, where $\mathcal{Q}^{0} := \mathcal{D} \cap L^\infty$ and $\mathcal{Q}^{1} :=\{1\}$ only ``contains'' the real-world measure $\P$.  Note that $\textnormal{ES}^{1}(X) = \mathbb{E}[-X]$;\footnote{For this risk measure, it is clear that the set of optimal portfolios for any $\nu \geq 0$ is given by $\Pi_\nu$, and $\textnormal{ES}_\nu = -\nu$.  Hence by Theorem \ref{thm:strong reg arb first characterisation}, the market admits strong $\rho$-arbitrage for $\textnormal{ES}^{1}$.}  $\textnormal{ES}^{0}$ corresponds to the worst-case risk measure considered in Section \ref{subsec:worst case risk}, where the supremum in \eqref{Eq:ES dual rep} is no longer attained.

For $\alpha \in (0,1)$,  Conditions I and UI are satisfied for $\ES^\alpha$ and $\cQ^\alpha$ is closed in $L^1$. Moreover, Proposition \ref{prop:ES Q tilde} shows that
\begin{equation}
\label{eq:Q tilde alpha for ES}
\tilde{\mathcal{Q}}^{\alpha}:=\{ Z \in \mathcal{D} : Z > 0 \ \mathbb{P}\textnormal{-a.s.~and } \lVert Z \rVert_{\infty} < \tfrac{1}{\alpha}  \}
\end{equation}
satisfies Conditions POS, MIX and INT. Note that $1 \in \tilde{\mathcal{Q}}^\alpha$. Using Proposition \ref{prop:P in tilde Q rho strictly exp bdd} together with Theorems \ref{thm: no strong reg arb} and \ref{thm: no reg arb equivalence},\footnote{Also note that $\ES^\alpha_1 < \infty$ because $\ES^\alpha$ is real-valued.} we arrive at the following complete description of mean-ES portfolio selection:
\begin{theorem}
	\label{thm:mean-ES portfolio selection}
	Fix $\alpha \in (0,1)$.  Then $\Pi^{\textnormal{ES}^\alpha}_\nu$ is nonempty, compact and convex for $\nu \geq 0$.  Moreover:
	\begin{enumerate}
		\item The market $(S^0,S)$ does not admit strong $\ES^\alpha$-arbitrage if and only if there exists $Z \in \mathcal{M}$ such that $\left\Vert Z \right\Vert_\infty \leq \frac{1}{\alpha}$. 
		\item The market $(S^0,S)$ does not admit $\ES^\alpha$-arbitrage if and only if there exists $Z \in \mathcal{P}$ such that $\left\Vert Z \right\Vert_\infty < \frac{1}{\alpha}$.
	\end{enumerate}
\end{theorem}

\begin{remark}
	It straightforward to check that
	\begin{equation}
	\hat{\mathcal{Q}}^{\alpha} :=\{ Z \in \mathcal{D} : \textnormal{there exists $\epsilon > 0$ such that } Z \geq \epsilon \ \mathbb{P}\textnormal{-a.s.~and } \lVert Z \rVert_{\infty} < \tfrac{1}{\alpha}  \}
	\end{equation}
is nonempty and	satisfies Conditions POS, MIX  and INT. Thus, Theorem \ref{thm:mean-ES portfolio selection}(b) can be strengthened:~The market $(S^0,S)$ does not admit $\ES^\alpha$-arbitrage if and only if there exists $Z \in \mathcal{P}$ with $Z \geq \epsilon \ \mathbb{P}$-a.s.~for some $\epsilon > 0$ and  $\left\Vert Z \right\Vert_\infty < \frac{1}{\alpha}$.
\end{remark}

\subsection{Spectral risk measures}
\label{sec:spectral}
Spectral risk measures are mixtures of Expected Shortfall risk measures that were introduced by Acerbi in \cite{acerbi2002spectral}. Here, we follow the definition of Cherny \cite{cherny2006weighted}, who has studied their finer properties in great detail. For a probability measure $\mu$ on $([0, 1], \cB_{[0, 1]})$, the \emph{spectral risk measure} $\rho^\mu : L^1 \to (-\infty, \infty]$ with respect to $\mu$ is given by
\begin{equation*}
\rho^{\mu}(X) := \int_{[0,1]}\textnormal{ES}^{\alpha}(X) \, \mu(\textnormal{d}\alpha).
\end{equation*}

\begin{remark}
(a) If $\mu$ does not have an atom at $0$, we can define the  non-increasing function $\phi^\mu : [0, 1] \to \RR_+$ by $\phi^\mu(u) := \int_{[u, 1]}\frac{1}{\alpha} \,\mu (\diff \alpha)$ and write $\rho^\mu(X) :=  \int_{0}^1 \phi^\mu(u) \VaR^u(X) \dd u $. This is the original definition of Acerbi \cite{acerbi2002spectral}. Some explicit examples for the choice of $\mu$ (or more precisely $\phi^\mu$) are given in \cite{dowd2008spectral}.
	
(b) It was shown in \cite[Theorem 7]{kusuoka2001law} for the domain $L^\infty$ that on a standard probability space where $\mathbb{P}$ is non-atomic, spectral risk measures  coincide with law-invariant, comonotone, coherent risk measures that satisfy the Fatou property. It was then shown in \cite{jouni:al:06} that the Fatou property is automatically satisfied by law-invariant coherent risk measures. The result has then been generalised to $L^1$ by \cite[Theorem 2.45]{pflug:romisch:07}.
\end{remark}
Spectral risk measures admit a dual representation. It follows from \cite[Theorem 4.4]{cherny2006weighted} that the maximal dual set $\cQ_{\rho^\mu}$ is $L^1$-closed and given by
\begin{equation*}
\mathcal{Q}_{\rho^\mu} = \bigg \{\int_{[0,1]}\zeta_{\alpha} \, \mu(\textnormal{d}\alpha) : \zeta_\alpha(\omega) \textnormal{ is jointly measurable and } \zeta_{\alpha} \in \mathcal{Q}^{\alpha} \textnormal{ for all } \alpha \in [0,1]   \bigg\},
\end{equation*}
where $\mathcal{Q}^{\alpha}$ is as in \eqref{Eq:ES dual rep}. Here, we are in a situation, were it is useful to consider a smaller dual set $\mathcal{Q}' \subset \mathcal{Q}_{\rho^\mu}$ so that we can explicitly construct a nonempty subset $\tilde \cQ' \subset \cQ'$ satisfying Conditions POS, MIX and INT. It turns out that a good choice is
\begin{align*}
\mathcal{Q}_{\mu} = \bigg \{\int_{[0,1]}\zeta_{\alpha} \, \mu(\textnormal{d}\alpha) : \ &\zeta_\alpha(\omega) \textnormal{ is jointly measurable and there is $1 >\epsilon > 0$} \\  \textnormal{such that } &\zeta_{\alpha} \in \mathcal{Q}^{\alpha} \textnormal{ for } \alpha \in [0,1-\epsilon] \textnormal{ and } \zeta_\alpha \equiv 1 \textnormal{ for } \alpha \in (1-\epsilon,1]   \bigg\},
\end{align*}
which is an $L^1$-dense subset of  $\mathcal{Q}_{\rho^\mu}$. It is shown in Proposition \ref{prop:app:spectral}(a) that $\mathcal{Q}_{\mu}$ also represents $\rho^{\mu}$. If $\mu$ does not have an atom at $0$ and $\int_{(0, 1]}\frac{1}{\alpha} \,\mu (\diff \alpha) < \infty$, it follows that  $\mathcal{Q}_{\mu}$ (and $\mathcal{Q}_{\rho^\mu}$) is bounded in $L^\infty$. Hence, $\rho^{\mu}$ is real-valued and Condition I and UI are satisfied.

If $\mu$ does not have an atom at $1$, it follows from  Proposition \ref{prop:app:spectral}(b) that the set 
\begin{align*}
\tilde \cQ_{\mu} = \bigg \{\int_{[0,1)}\zeta_{\alpha} \, \mu(\textnormal{d}\alpha) : \ &\zeta_\alpha(\omega) \textnormal{ is jointly measurable and there is $0<\epsilon<1$ and $0 < \delta < \tfrac{\epsilon}{1-\epsilon}$} \\  \textnormal{such that } &\zeta_{\alpha} \in \tilde \cQ^{\alpha(1 + \delta)} \textnormal{ for } \alpha  \in [0,1-\epsilon] \textnormal{ and } \zeta_\alpha \equiv 1 \textnormal{ for } \alpha \in (1-\epsilon,1)   \bigg\},
\end{align*}
 where $\tilde \cQ^{\alpha(1+\delta)}$ is as in \eqref{eq:Q tilde alpha for ES}, satisfies Conditions POS, MIX and INT. Note that $1 \in \tilde \cQ_{\mu}$. Using Proposition \ref{prop:P in tilde Q rho strictly exp bdd} together with Theorems \ref{thm: no strong reg arb} and \ref{thm: no reg arb equivalence} we arrive at the following result:
 
 \begin{corollary}
Let $\mu$ be a probability measure on $([0, 1], \cB_{[0, 1]})$ such that $\mu (\{0\}) = 0$ and $\int_{(0, 1]}\frac{1}{\alpha} \,\mu(\diff \alpha)< \infty$. Then 
$\Pi^{\rho^\mu}_\nu$ is nonempty, compact and convex for $\nu \geq 0$.  Moreover:
 	\begin{enumerate}
 		\item The market $(S^0,S)$ does not admit strong $\rho^\mu$-arbitrage if and only if there exists $Z \in \mathcal{M}$ such that $Z = \int_{[0,1]}\zeta_{\alpha} \,\mu(\textnormal{d}\alpha)$, where $\zeta_\alpha(\omega)$ is jointly measurable and satisfies $\zeta_{\alpha} \in \cD$ and $\Vert \zeta_{\alpha} \Vert_\infty \leq \frac{1}{\alpha}$. 
 		\item If $\mu$ does not have an atom at $1$, the market $(S^0,S)$ does not admit $\rho^\mu$-arbitrage if and only if there exists $Z \in \mathcal{P}$, $0 < \epsilon <1$ and $0 < \delta < \tfrac{\epsilon}{1-\epsilon}$ such that $Z = \mu((1-\epsilon, 1)) + \int_{[0,1-\epsilon]}\zeta_{\alpha} \,\mu(\textnormal{d}\alpha)$, where $\zeta_\alpha(\omega)$ is jointly measurable and satisfies $\zeta_{\alpha} \in \cD$ and $\Vert \zeta_{\alpha} \Vert_\infty \leq \frac{1}{\alpha(1+\delta)}$ for $\alpha \in [0, 1-\epsilon]$.
 	\end{enumerate}
 \end{corollary}
 
\subsection{Coherent risk measures on Orlicz spaces}
We proceed to discuss how our main results can be applied to the case where the returns lie in some Orlicz space $L^\Phi$ and $\rho$ is real-valued on $L^\Phi$.  Risk measures on Orlicz spaces/Orlicz hearts are well studied; see e.g.~\cite{cheridito2009risk, gao:al:19}. Not only do these spaces allow for the inclusion of \emph{unbounded} random variables, there is also an elegant duality theory.  For a brief overview of some key definition and results, see Appendix \ref{app:Orlicz}.  

We consider the following setup: Let $\Phi: [0, \infty) \to [0, \infty]$ be a Young function and $\rho: L^\Phi \to \RR$ a coherent risk measure that is expectation bounded. 

\medskip
We first consider $L^\Phi = L^\infty$, i.e., when $\Phi$ jumps to infinity, which is different from all other Orlicz spaces in that the corresponding Orlicz heart is the null space.  In this case, $\rho$ admits a dual representation if it satisfies the Fatou property (cf.~Theorem \ref{thm:dual:orlicz}(c)) and we have the following result.
\begin{corollary}
	\label{cor:L infty}
	Let $\rho:L^\infty \to \RR$ be an expectation bounded coherent risk measure on $L^\infty$ that satisfies the Fatou property.  Let $\cQ \subset \cQ_\rho$ be a convex subset with $1 \in \cQ$ and $\ol \cQ = \cQ_\rho$. Suppose that $R^i\in L^\infty$. If $\rho$ is strictly expectation bounded, then $\Pi^\rho_\nu$ is nonempty, compact and convex for all $\nu \geq 0$. Moreover:
	\begin{enumerate}
		\item If $\rho$ is continuous from below $($that is $\rho(X_n) \searrow \rho(X)$ whenever $X_n \nearrow X \ \mathbb{P}$-a.s.$)$, the market $(S^0,S)$ does not admit strong $\rho$-arbitrage if and only if $\cQ_\rho \cap \mathcal{M} \neq \emptyset$.
		\item If  there exists $\tilde{\mathcal{Q}} \subset \mathcal{Q}$ satisfying Conditions POS, MIX and INT with $1 \in \tilde \cQ$, then the market $(S^0,S)$ does not admit $\rho$-arbitrage if and only if $\tilde{\mathcal{Q}} \cap \mathcal{P} \neq \emptyset$.
	\end{enumerate}
\end{corollary}

\begin{proof}
The first assertion follows from Theorem \ref{thm:Existence of optimal portfolios} and Corollary \ref{cor:exp bounded}.  Next, since $R^i \in L^\infty$, Condition UI is satisfied if and only if the dual set $\cQ$  is uniformly integrable, which by \cite[Corollary 4.35]{follmerschied:2016} is equivalent to $\rho$ being continuous from below. Since $\bar \cQ =\cQ_\rho$, part (a) follows from Proposition \ref{prop:orlicz dual} and Theorem \ref{thm: no strong reg arb}. Finally, Condition I is trivially satisfied and so part (b) follows from Theorem \ref{thm: no reg arb equivalence}.
\end{proof}

We now consider the case of Orlicz spaces for a finite Young function.  See Theorem \ref{thm:dual:orlicz} for conditions under which $\rho$ admits a dual representation.
\begin{corollary}
	\label{cor:Orlicz}
	Let $\Phi: [0, \infty) \to [0, \infty)$ be a finite Young function with conjugate $\Psi$ and $\rho: L^\Phi \to \RR$ an expectation bounded coherent risk measure that admits a dual representation.  Let $\cQ \subset \cQ_\rho$ be a convex subset with $1 \in \cQ$ and whose closure in $L^\Psi$ is $\cQ_\rho$. Suppose that $R^i\in L^\Phi$. If $\rho$ is strictly expectation bounded, then $\Pi^\rho_\nu$ is nonempty, compact and convex for all $\nu \geq 0$. Moreover:
	\begin{enumerate}
		\item If $R^i \in H^\Phi$, the market does not admit strong $\rho$-arbitrage if and only if $\ol \cQ_\rho \cap \mathcal{M} \neq \emptyset$.\footnote{Here $\ol \cQ_\rho$ denotes the closure of $\cQ_\rho$ in $L^1$.}
		\item If  there exists $\tilde{\mathcal{Q}} \subset \mathcal{Q}$ satisfying Conditions POS, MIX and INT with $1 \in \tilde \cQ$, then the market $(S^0,S)$ does not admit $\rho$-arbitrage if and only if $\tilde{\mathcal{Q}} \cap \mathcal{P} \neq \emptyset$.
	\end{enumerate}
\end{corollary}

\begin{proof}
	The first assertion follows from Theorem \ref{thm:Existence of optimal portfolios} and Corollary \ref{cor:exp bounded}.  Next, since $R^i \in H^\Phi$, Condition UI is satisfied by Proposition \ref{prop:app:Orlicz}(b), and (a) follows from Proposition \ref{prop:orlicz dual} and Theorem \ref{thm: no strong reg arb}. Finally, Condition I follows from $R^i \in L^\Phi$ and the generalised Hölder inequality \eqref{eq:generalised holder}. Now part (b) follows from Theorem \ref{thm: no reg arb equivalence}.
\end{proof}

\begin{remark}
	\label{remark:delta 2}
If $\Phi$ does not satisfy the $\Delta_2$-condition and $R^i \in L^\Phi \setminus H^\Phi$ for some $i \in \{1, \ldots, d\}$, then it is in general not possible to provide a dual characterisation of strong $\rho$-arbitrage since condition UI is not satisfied.  The reason for this is that Proposition \ref{prop:app:Orlicz} does not extend to $L^\Phi$. However, we can often provide a dual characterisation of $\rho$-arbitrage since finding $\tilde{\mathcal{Q}} \subset \mathcal{Q}$ satisfying Conditions POS, MIX and INT with $1 \in \tilde \cQ$ is possible in many cases; cf.~Corollary \ref{cor:g ent}.
\end{remark}

\subsubsection{$g$-entropic risk measures}
We proceed to apply the above results to the class of $g$-entropic risk measures. The class of  $g$-entropic risk measures was introduced by Ahmadi-Javid \cite[Definition 5.1]{ahmadi2012entropic}. It is best understood when presented in the context of Orlicz spaces. Let $\Phi: [0, \infty) \to \RR$ be a finite superlinear Young function and $\Psi$ its conjugate. Let $g :[0, \infty) \to [0, \infty)$ be a convex function that dominates $\Psi$. For $\beta  > g(1)$, define the risk measure $\rho^{g, \beta} :L^\Phi \to \RR$ by\footnote{Note that our definition slightly differs from the definition in \cite{ahmadi2012entropic}, who considers the domain $L^\infty$ and assumes that $g$ is convex, $(-\infty, \infty]$-valued and satisfies $g(1)  =0$.}
\begin{equation*}
\rho^{g, \beta}(X) = \sup_{Z \in \mathcal{Q}^{g,\beta}} \mathbb{E}[-ZX], \quad \textnormal{where} \quad
\mathcal{Q}^{g,\beta} := 
\{ Z \in \mathcal{D} : \mathbb{E}[g(Z)] \leq \beta \}, 
\end{equation*}
and call it the \emph{$g$-entropic risk measure with divergence level $\beta$}. By convexity and nonnegativity of $g$ and the fact that $g$ dominates $\Psi$, it follows that $\mathcal{Q}^{g,\beta}$ is convex, $L^\Psi$-bounded and $L^1$-closed.\footnote{More precisely, $\Vert Z \Vert_\Psi \leq \max(1, \beta)$ for all $Z \in \cQ^{g,\beta}$ and $L^1$-closedness follow from Fatou's lemma.} By Proposition \ref{prop:orlicz dual}, we may deduce that  $\mathcal{Q}^{g,\beta} = \cQ_{\rho^{g, \beta}}$. Moreover, Proposition \ref{prop:Interior of Q g beta} shows that
\begin{align*}
\tilde{\mathcal{Q}}^{g,\beta} & := \{ Z \in \mathcal{D} : Z>0 \ \mathbb{P}\textnormal{-a.s.\ and } \mathbb{E}[g(Z)] < \beta \}
\end{align*}
satisfies Conditions POS, MIX and INT. Note that $1 \in \tilde \cQ^{g, \beta} \subset \cQ^{g, \beta}$. Applying Corollary \ref{cor:Orlicz}, we get the following result:

\begin{corollary}
	\label{cor:g ent}
Let $\Phi: [0, \infty) \to \RR$ be a superlinear finite Young function with conjugate $\Psi$, $g :[0, \infty) \to [0, \infty)$ a convex function that dominates $\Psi$ and $\beta  > g(1)$.  Suppose that $R^i \in L^\Phi$.
Then $\Pi^{\rho^{g, \beta}}_\nu$ is nonempty, compact and convex for all $\nu \geq 0$. Moreover:
\begin{enumerate}
	\item If $R^i\in H^\Phi$, the market $(S^0,S)$ does not admit strong $\rho^{g, \beta}$-arbitrage if and only if there is $Z \in \cM$ with $\mathbb{E}[g(Z)] \leq \beta$.
	\item The market $(S^0,S)$ does not admit $\rho^{g, \beta}$-arbitrage if and only if there is $Z \in \cP$ with $\mathbb{E}[g(Z)] < \beta$.
\end{enumerate}
	
\end{corollary}

We finish this section, by providing two specific examples of $g$-entropic risk measures.

\subsubsection*{Transformed norm risk measure}
Let $p \in (1, \infty)$ and  $\alpha \in (0,1)$.\footnote{The case $p=1$ corresponds to Expected Shortfall, see Section \ref{subsec:VaR and ES}.}  Define the \emph{transformed $L^p$-norm risk measure}  with sensitivity parameter $\alpha$ as 
\begin{equation*}
\rho(X) := \min_{s \in \mathbb{R}} \{ \tfrac{1}{\alpha} \lVert(s-X)^{+} \rVert_{p} - s \}, \quad X \in L^p.
\end{equation*}
It is shown in \cite[Section 5.3]{cheridito2009risk}  that this is a real-valued coherent risk measure on $L^p$ and admits the dual representation with
\begin{equation*}
\label{eq:transformed norm risk measure*}
\cQ_{\rho}= \{ Z \in \mathcal{D} : \lVert Z \rVert_{q} \leq \tfrac{1}{\alpha} \},
\end{equation*}
where $q := p/(p-1)$. Hence $\rho = \rho^{g, \beta}$, where $\Phi(x) = x^p/p$, $\Psi(y) = y^q/q$, $g = \Psi$ and $\beta := (\frac{1}{\alpha})^q/q$.

\subsubsection*{Entropic value at risk}
The entropic value at risk (EVaR) was introduced in Ahmadi-Javid \cite{ahmadi2012entropic} and further studied in~\cite{ahmadi2017analytical}. Consider the Young function $\Phi(x)=\exp(x)-1$ and fix $\alpha \in (0,1)$.
Then the \emph{entropic value at risk at level $\alpha$} is a risk measure on $L^\Phi$ given by\footnote{Note that the parametrisation in \cite{ahmadi2017analytical} is different: $\alpha$ is replaced by $1-\alpha$ and $X$ by $-X$.}  
\begin{equation*}
\textnormal{EVaR}^\alpha(X) := \inf_{z>0} \left\{ \frac{1}{z} \log \left(\mathbb{E}\left[\frac{\exp(-zX)}{\alpha}\right]\right)\right\}.
\end{equation*}
It is shown in \cite[Section 4.4]{ahmadi2017analytical} that it admits a dual representation with dual set
\begin{equation*}
\mathcal{Q} := \{ Z \in \cD: \mathbb{E}[Z\log(Z)] \leq -\log(\alpha) \}.
\end{equation*}
Hence, $\textnormal{EVaR}^\alpha = \rho^{g, \beta}$, where $\Psi(y) = (y \log(y) - y + 1)\1_{\{y \geq 1\}}$, $g(y) = y \log(y) - y + 1$ and $\beta := -\log(\alpha)$.

\section{Conclusion and outlook}
\label{section:outlook}

It has been said (cf.~\cite{cherny2008pricing}) that there have been three major revolutions in finance:~the first one was Markowitz' mean-variance analysis \cite{Markowitz1952}, which led to the CAPM of Treynor, Sharpe, Lintner and Mossin; the second revolution was the Black-Scholes-Merton formula; and the third one was the theory of coherent risk measures developed by Artzner, Delbaen, Eber and Heath \cite{artzner1999coherent}.  In this paper, we have endeavoured to link the first and third revolution by substituting the variance in classical portfolio selection with a positively homogeneous risk measure $\rho$.

We have shown that under mild assumptions, $\rho$-optimal portfolios for a fixed return exist.  However, somewhat surprisingly, $\rho$-efficient portfolios may \emph{fail} to exist.  We referred to this situation as $\rho$-arbitrage.  The first aim of this paper has been to make regulators aware of this pitfall, which is a generalisation of arbitrage of the first kind.

The second aim of this paper has been to explain why this complication arises and how to avoid it.  The fundamental theorem of asset pricing states that the market does not satisfy arbitrage of the first kind (i.e., does not admit $\rho$-arbitrage for the the worst-case risk measure) if and only if $\mathcal{P} \cap L^\infty \neq \emptyset$.  Our main result, Theorem \ref{thm: no reg arb equivalence}, extends this.  We have shown that for coherent risk measures under mild assumptions on the dual set $\cQ$, the market does not admit $\rho$-arbitrage if and only if $\mathcal{P} \cap \tilde \cQ \neq \emptyset$ for some nonempty $\tilde \cQ \subset \cQ$ satisfying properties POS, MIX and INT. We have also demonstrated that $\tilde \cQ$ can be computed explicitly for a large variety of risk measures. Furthermore, we have shown that amongst markets that do not admit arbitrage (of the first kind), $\rho$-arbitrage cannot be excluded unless $\rho$ is the worst-case risk measure. Since a worst-case approach to risk is infeasible in practise, this shows that regulators cannot avoid the existence of (strong) $\rho$-arbitrage when imposing a positively homogeneous risk measure.

\medskip
Going back to the two questions posed in the introduction, we see that it is certainly possible for ES constraints to be ineffective. The root of this issue stems specifically from positive homogeneity.  Therefore, if there is an alternative superior risk measure, it \emph{cannot} be coherent. Naturally, this leads to the following questions regarding convex, but \emph{not} positively homogeneous, measures of risk:~Do optimal portfolios exist?~Does $\rho$-arbitrage occur?~If so, can we give a dual characterisation? We intend to return to these questions in a subsequent publication.

\appendix

\section{Counterexamples}
\label{app:counterexamples}

In this appendix, we give several counterexamples to complement the results  in Sections \ref{section:portfolio optimisation} and~\ref{section:dual characterisations}.

\begin{example}
\label{example:rho one not attained but finite}
In this example we show that if all assumptions of Theorem \ref{thm:Existence of optimal portfolios} hold, but $\{\mathbf{0}\} \subsetneq \Pi^\rho_0$, the result fails. 

Let $\Omega = [-5,5] \times [1,7] \subset \mathbb{R}^2$ with the Borel $\sigma$-algebra and the uniform probability measure $\mathbb{P}$.  Let $r=0$  and assume there are two risky assets with returns $R^i(\omega) := \omega_i$ for $\omega = (\omega_1,\omega_2) \in \Omega$ and $i \in \{1, 2\}$.  Let $C$ be the closed ball of radius 2 centred at $(2,4)$, and for each $(x,y) \in C$, let $C_{(x,y)}$ be the closed ball of radius~$1$ centred at $(x,y)$, and $Z_{(x,y)}$ the Radon-Nikod\'{y}m derivative of the uniform probability measure on $C_{(x,y)}$ with respect to $\P$. Define the risk measure $\rho$ via its dual set
\begin{equation*}
   \mathcal{Q} = \{Z_{(x,y)}: (x,y) \in C\},
\end{equation*}
and note that $\mathbb{E}[Z_{(x,y)} R^1] =x$ and $\mathbb{E}[Z_{(x,y)} R^2] =y$. For this financial market (that is nonredundant and nondegenerate), $\mathbb{E}[R^1]=0$, $\mathbb{E}[R^2]=4$, and $\Pi_1 = \{ (\pi^1,\pi^2) : \pi^1 \in \mathbb{R}, \ \pi^2=1/4 \}$.  Thus, for every $\pi \in \Pi_1$ and $(x,y) \in C$,
\begin{equation*}
    \mathbb{E}[-Z_{(x,y)}X_\pi] =  -\pi \cdot (x,y) = -(\pi^1 x + \tfrac{1}{4}y).
\end{equation*}
It follows that for any $\pi \in \Pi_1$,
\begin{equation*}
    \rho(X_\pi) = \sup_{(x,y) \in C} -(\pi^1 x + \tfrac{1}{4}y) = \tfrac{1}{2}\sqrt{16(\pi^1)^2 + 1} - 2\pi^1 - 1 =: g(\pi^1).
\end{equation*}
Therefore, $\rho_1 = \inf \{ \rho(X_\pi) : \pi \in \Pi_1 \} = \inf \{ g(\pi^1) : \pi^1 \in \mathbb{R} \}= -1$ is not attained, since $g$ is strictly decreasing.  Thus, $\Pi^\rho_1$ is empty, even though $\rho$ satisfies the Fatou property on $\{X_\pi:\pi \in \RR^d\}$ and $\rho_1 \in \mathbb{R}$.  The reason Theorem \ref{thm:Existence of optimal portfolios} fails is because $\Pi^\rho_{0} = \{(\pi^1,\pi^2) : \pi^1 \geq 0, \pi^2=0 \} \supsetneq \{\mathbf{0}\}$.  
\end{example}

For the rest of the counterexamples, we take $\Omega = [0,1]$, with the Borel $\sigma$-algebra and the Lebesgue measure $\mathbb{P}$. In each example, the financial market is nonredundant and the risky returns are integrable and nondegenerate.  Moreover, we always have $1 \in \mathcal{Q}$.

\begin{example}
\label{example: weird example}
In this example we show that when Condition I  is not satisfied, the set $C_\cQ$ from \eqref{eq: C Q} may fail to be a convex subset of $\RR^d$ and \eqref{relating Q to C} may break down.

Suppose $r=0$ and there are two risky assets with returns
\begin{equation*}
R^{1}(\omega) := 
\begin{cases}
\tfrac{3}{\sqrt{\omega}}, &\text{if } \omega < \tfrac{1}{16}, \\
-\tfrac{8}{15}, &\text{if } \omega \geq \tfrac{1}{16},
\end{cases}
\quad 
\textnormal{and}
\quad 
R^{2}(\omega) := 
\begin{cases}
-\tfrac{1}{\sqrt{\omega}}, &\text{if } \omega < \tfrac{1}{16}, \\
\tfrac{24}{15}, &\text{if } \omega \geq \tfrac{1}{16}.
\end{cases}
\end{equation*}
Let $\mathcal{Q} := \{ \lambda Z +(1-\lambda) : \lambda \in [0,1] \}$, where
\begin{equation*}
    Z(\omega) := 
\begin{cases}
\tfrac{2}{\sqrt{\omega}}, &\text{if } \omega < \tfrac{1}{16}, \\
0, &\text{if } \omega\geq \tfrac{1}{16}.
\end{cases}
\end{equation*}
Note that $\mathbb{E}[-R^1] = \mathbb{E}[-R^2] = -1$, $\mathbb{E}[-Z R^1] = -\infty$ and $\mathbb{E}[-Z R^{2}] = \infty$. Thus, 
\begin{equation*}
C_{\cQ} = \{(-1, 1), (-\infty, \infty)\},
\end{equation*}
which is neither convex nor a subset of $\RR^2$.

Moreover, the portfolio $\pi = (\tfrac{1}{4}, \frac{3}{4})$ satisfies $\rho(X_\pi) = \mathbb{E}[-Z X_\pi] = 0$ but
\begin{equation*}
    \sup_{c \in C_{\mathcal{Q}}} (\pi \cdot c) = \max \{ -1, -\infty \tfrac{1}{4} + \infty \tfrac{3}{4}\} \neq 0 = \rho(X_\pi),
\end{equation*}
and so \eqref{relating Q to C} does not hold.
\end{example}

\begin{example}
	\label{example:UI}
	In this example we show that when only Condition I is satisfied but Condition UI  is not, the set $C_{\ol \cQ}$ from \eqref{eq: C bar Q} may fail to be a subset of $\RR^d$, whence \eqref{relating Q bar to C} breaks down.
	
	Suppose $r \neq 0$ (so the market is nondegenerate) and there is one risky asset with return
	\begin{equation*}
		R(\omega) := 
		\begin{cases}
			\tfrac{1}{\sqrt{\omega}}, &\text{if } \omega \in (0, \tfrac{1}{2}) ,\\
			-\tfrac{1}{\sqrt{\omega-1/2}}, &\text{if }\omega \in (\tfrac{1}{2}, 1). \\
		\end{cases}
\end{equation*}
Note that $R(\omega) = -R(\omega+\tfrac{1}{2})$ for $\omega \in (0, \tfrac{1}{2})$. Let $\cQ =  \{Z \in \cD \cap L^\infty : Z(\omega) = Z(\omega+ \tfrac{1}{2}) \text{ for all } \omega \in (0, 1/2)\}$. Then Condition I is satisfied and $\mathbb{E}[- Z (R-r)] = r$ for all $Z \in \cQ$, whence $C_\cQ = \{r\}$. Moreover, $\ol \cQ = \{Z \in \cD : Z(\omega) = Z(\omega+ \tfrac{1}{2}) \text{ for all } \omega \in (0, 1/2)\}$ and $\tfrac{|R|}{\sqrt{8}} \in \ol \cQ$. Since $\mathbb{E}[R (R)^-]=+\infty$, it follows that $C_{\ol \cQ} = \{r,\infty\}$, which is neither convex, compact, nor a subset of $\RR$. Finally, for $\pi = 1$, $\sup_{c \in C_{\mathcal{Q}}} (\pi \cdot c) = r \neq \infty = \sup_{c \in C_{\ol \cQ}} (\pi \cdot c)$.
\end{example}

\begin{example}
\label{counterexample cond 1a not enough}
In this example we show the converse of Proposition \ref{Q intersect M empty} fails.

Let $r=0$ and assume there is one risky asset whose return $R^1$ is uniformly distributed on $[0,1]$.  Let $\rho$ be the worst-case risk measure, cf.~Section \ref{subsec:worst case risk}.  Then Condition I is satisfied, $\mathcal{Q} \cap \mathcal{M} = \emptyset$ (because $\mathcal{M} = \emptyset$), but $\rho(X_\pi) \geq 0$ for any portfolio $\pi$. Therefore, by Theorem \ref{thm:strong reg arb first characterisation}, this market does not admit strong $\rho$-arbitrage, even though $\mathcal{Q} \cap \mathcal{M} = \emptyset$.
\end{example}

\begin{example}
\label{counterexample cond 1a and UI not enough}
In this example we show that when $\mathcal{Q}$ is uniformly integrable but $R\mathcal{Q}$ is not, Theorem \ref{thm: no strong reg arb} may fail.

Let the risk-free rate be given by $r = 1 + 12c$, where $c:= \int_{1/4}^{1/3} \log(1/x) \dd x$.  Suppose there is one risky asset whose return is given by
\begin{equation*}
    R(\omega) = 
\begin{cases}
\ln\left(\tfrac{1}{\omega}\right), &\text{if } \omega < \tfrac{1}{3}, \\
0, &\text{if } \omega \in [\tfrac{1}{3}, \tfrac{2}{3}], \\
-1, &\text{if } \omega > \tfrac{2}{3}.
\end{cases}
\end{equation*}
Next, for $n \geq 4$, set
\begin{equation*}
    Z_{n}(\omega) = 
    \begin{cases}
\tfrac{n}{\ln(1/\omega)}, &\text{if } \omega < \tfrac{1}{n}, \\
0, &\text{if } \omega \in [\tfrac{1}{n}, \tfrac{1}{4}], \\
k_{n}, &\text{if } \omega \in (\tfrac{1}{4}, \tfrac{1}{3}], \\
0, &\text{if } \omega > \tfrac{1}{3},
\end{cases}
\end{equation*}
where $k_{n}$ is chosen so that $\mathbb{E}[Z_{n}] = 1$.  Note that $k_{n} \uparrow 12$, and that $Z_{n}$ converges in $L^1$ to $Z = 12\mathds{1}_{(1/4, 1/3]}$.  Therefore, $(\cup_{n \geq 4} \{Z_{n}\}) \cup \{Z\}$ is uniformly integrable, and whence, if we let $\mathcal{Q}$ be the $L^1$-closed convex hull of $(Z_{n})_{n \geq 4}$, $Z$ and $1$, it will also be uniformly integrable.  Moreover,
\begin{equation*}
\label{eq:EZnR does not converge to EZR}
    \mathbb{E}[Z_{n}R] = 1 + k_{n}c \uparrow 1+12c \quad \textnormal{but} \quad \mathbb{E}[ZR] = 12c.
\end{equation*}
It follows that the set $C_{\mathcal{Q}}$ is given by 
\begin{equation*}
    C_{\mathcal{Q}} = \{\mathbb{E}[-Y(R-r)] : Y \in \mathcal{Q} \} = (0, d],
\end{equation*}
where $d:=\mathbb{E}[-(R-r)] > 1$.  Thus Condition I is satisfied, $\mathcal{Q}$ is uniformly integrable and $\ol\cQ \cap \mathcal{M}=\mathcal{Q} \cap \mathcal{M} = \emptyset$, but the market does not admit strong $\rho$-arbitrage:
\begin{equation*}
    \rho(X_{\pi}) = \sup_{c \in C_{\mathcal{Q}}} (\pi \cdot c) \geq 0, \quad \textnormal{for any portfolio } \pi \in \mathbb{R}.
\end{equation*}
\end{example}

\begin{example}
\label{counterexample strong reg arb Q not UI}
In this example we show that when $R\mathcal{Q}$ is uniformly integrable but $\mathcal{Q}$ is not, Theorem \ref{thm: no strong reg arb} may fail.

Let $r = 0$ and  suppose there is one risky asset whose return is given by
\begin{equation*}
R(\omega) = 
\begin{cases}
1, &\text{if } \omega \leq \tfrac{1}{4}, \\ 0, &\text{if } \omega \in (\tfrac{1}{4},\tfrac{3}{4}), \\
-\tfrac{1}{2}, &\text{if } \omega \geq \tfrac{3}{4}.
\end{cases}
\end{equation*}
For $n \geq 2$,  define the intervals $A_n:=(\tfrac{1}{2},\tfrac{1}{2}+\tfrac{1}{2^n})$ and set 
\begin{equation*}
    Z_{n}(\omega) = 
    \begin{cases}
2^n - \tfrac{1}{n}, &\text{if } \omega \in A_n, \\
0, &\text{if } \omega \in (\tfrac{1}{4}, \tfrac{3}{4})\setminus A_n, \\
k_{n}, &\text{if } \omega \in [0,\tfrac{1}{4}] \cup [\tfrac{3}{4},1],
\end{cases}
\end{equation*}
where $k_{n}$ is chosen so that $\mathbb{E}[Z_{n}] = 1$.  Note that $k_{n} \downarrow 0$.  Let $\mathcal{Q}$ be the closed convex hull of $(Z_{n})_{n \geq 2}$ and $1$.  Then $\mathcal{Q}$ is not uniformly integrable but $R\mathcal{Q}$ is.  Moreover,  $\mathbb{E}[R] = \tfrac{1}{8}$, $\mathbb{E}[Z_nR] =  \tfrac{1}{2} k_n \downarrow 0$ and $\tfrac{1}{2}k_2 < \tfrac{1}{8}$.  It follows that
\begin{equation*}
    C_{\mathcal{Q}} = \{\mathbb{E}[-Z(R-r)] : Z \in \mathcal{Q} \} = [-\tfrac{1}{8},0).
\end{equation*}
Thus $\ol\cQ \cap \mathcal{M} = \mathcal{Q} \cap \mathcal{M} = \emptyset$, but the market does not admit strong $\rho$-arbitrage:
\begin{equation*}
    \rho(X_{\pi}) = \sup_{c \in C_{\mathcal{Q}}} (\pi \cdot c) \geq 0, \quad \textnormal{for any portfolio } \pi \in \mathbb{R}.
\end{equation*}
\end{example}

\begin{example}
\label{counterexample empty interior}
In this example we show that when $\tilde \cQ_{\max} = \emptyset$, Theorem \ref{thm: no reg arb equivalence} fails.

Consider the financial market described in Example \ref{counterexample strong reg arb Q not UI}. Let $\mathcal{Q}$ be the convex hull of the two densities $1$ and $Y(\omega)=2\mathds{1}_{(1/2, 1]}(\omega)$,
\begin{equation*}
    \mathcal{Q} = \{ \mu Y +(1-\mu) : \mu \in [0,1] \}.
\end{equation*}
Then $\mathbb{E}[YR] = -\tfrac{1}{4}$ and $\mathbb{E}[R] = \tfrac{1}{8}$, so $C_{\mathcal{Q}} = [-\tfrac{1}{8}, \tfrac{1}{4}]$ and there is no $\rho$-arbitrage.  
\newline 
\indent However, $\tilde{\mathcal{Q}} \cap \mathcal{P} = \emptyset$ because $\tilde \cQ_{\max} = \emptyset$. 
Indeed, any $Z \in \mathcal{Q}$ is of the form 
\begin{equation*}
Z(\omega) = 
    \begin{cases}
    1-\mu, &\text{if } \omega \leq \tfrac{1}{2},\\
	1+\mu, &\text{if } \omega > \tfrac{1}{2},
    \end{cases}
\end{equation*}
for some $\mu \in [0,1]$. Therefore if $Z \in \mathcal{Q}$ and $\lambda > 0$, then $\lambda Z + (1 - \lambda) \tilde Z\in \cQ$ for some $\tilde Z \in \cD \cap L^\infty$ implies that $\tilde Z \in \cQ$ and since $\cQ$ is not dense in $\cD \cap L^\infty$, the result follows. 
\end{example}

\section{Dual representation of coherent risk measures on Orlicz spaces}
\label{app:Orlicz}

The goal of this appendix is to recall some key definitions and results on Orlicz spaces and summarise the main results on the dual representation of a real-valued coherent risk measure defined on an Orlicz space.

\subsection{Key definitions and results on Orlicz spaces}
\label{app:subsection:primer on orlicz spaces}
We begin by recalling some key definitions and results relating to Orlicz spaces and Orlicz hearts; see \cite[Chapter 10]{wilson2008weighted} and \cite[Chapter 2]{edgar1992stopping} for details.
\begin{itemize}
	\item A function $\Phi: [0,\infty) \to [0,\infty]$ is called a {\it Young function} if it is convex and satisfies $\lim_{x \to \infty} \Phi(x) = \infty$ and $\lim_{x \to 0} \Phi(x) = \Phi(0) =0$.  A Young function $\Phi$ is called {\it superlinear} if $\Phi(x)/x \to \infty$ as $x \to \infty$.\footnote{Note that a Young function is  continuous except possibly at a single point, where it jumps to $\infty$.  Thus a finite Young function is continuous.}
	
	\item Given a Young function $\Phi$, the \emph{Orlicz space} corresponding to $\Phi$ is given by  
	\begin{equation*}
	L^{\Phi} := \{ X \in L^{0} : \mathbb{E}[\Phi(c|X|)] < \infty \textnormal{ for some }c>0 \},
	\end{equation*}
	and the \emph{Orlicz heart} is the linear subspace 
	\begin{equation*}
H^{\Phi} := \{ X \in L^{\Phi} : \mathbb{E}[\Phi(c|X|)] < \infty \textnormal{ for all }c>0 \}.
	\end{equation*}
	\item $L^\Phi$ and $H^\Phi$ are Banach spaces under the \emph{Luxemburg norm} given by
	\begin{equation*}
	\lVert X\rVert_{\Phi} := \inf \left\{ \lambda >0 : \mathbb{E}\left[ \Phi\left(\left|\tfrac{X}{\lambda}\right|\right) \right] \leq 1 \right\}. 
	\end{equation*}
	\item For any Young function $\Phi$, its convex conjugate $\Psi :[0, \infty) \to [0, \infty]$ defined by
	\begin{equation*}
	\Psi(y):=\sup_{x \geq 0} \{xy-\Phi(x)\}
	\end{equation*}
	is also a Young function and its conjugate is $\Phi$.
	\item If $X \in L^\Phi$ and $Y \in L^\Psi$, we have the  generalised H{\"o}lder inequality: 
	\begin{equation}
	\label{eq:generalised holder}
	\mathbb{E}[|XY|] \leq 2\lVert X\rVert_{\Phi}\lVert Y\rVert_{\Psi}. 
	\end{equation}
	\item  Using the conjugate $\Psi$ and \eqref{eq:generalised holder}, we may define the \emph{Orlicz norm} on $L^\Phi$ by 
	\begin{equation*}
	\lVert X\rVert_{\Psi}^{*} := \sup \{\mathbb{E}[XY] : Y \in L^\Psi, \ \lVert Y\rVert_{\Psi} \leq 1 \}.
	\end{equation*}
	This norm is equivalent to the Luxemburg norm on  $L^\Phi$.
	\item When $\Phi$ jumps to infinity, then $L^\Phi = L^\infty$ (and $\Vert \cdot \Vert_\Phi$ is equivalent to $\Vert \cdot \Vert_\infty$) and $H^\Phi = \{0\}$.  
	\item When $\Phi$ is finite, the norm dual of the Orlicz heart $(H^\Phi, \lVert \cdot \rVert_{\Phi})$ (with the Luxemburg norm)  is the  Orlicz space $(L^\Psi, \lVert \cdot \rVert_{\Phi}^{*})$ (with the Orlicz norm).
	\item $\Phi$ is said to satisfy the $\Delta_2$-condition if there exists a finite constant $K > 0$ such that $\Phi(2 x) \leq K \Phi(x)$ for all $x \in [0, \infty)$. $\Phi$ satisfies the $\Delta_2$ condition if and only if $L^\Phi = H^\Phi$.
\end{itemize}

\subsection{Dual representation of coherent risk measures on Orlicz spaces}
After these preparations, we consider the following setup: Let $\Phi: [0, \infty) \to [0, \infty]$ be a Young function and $\rho: L^\Phi \to \RR$ a coherent risk measure.  To give a review of when $\rho$ admits a dual representation, we first consider two versions of the Fatou property.

\begin{definition}
	\label{def:Fatou strong}
Let $\Phi: [0, \infty) \to [0, \infty]$ be a Young function and $\rho: L^\Phi \to (-\infty,\infty]$ a map. Then $\rho$ is said to satisfy the
\begin{itemize}
\item \emph{Fatou property} on $L^\Phi$, if $X_n \to X$ $\as{\P}$ for $X_n, X \in L^\Phi$ and $|X_n| \leq Y$ $\as{\P}$ for some $Y \in L^\Phi$ implies that $\rho(X) \leq \liminf_{n \to \infty} \rho(X_n)$.
\item \emph{strong Fatou property} on $L^\Phi$, if $X_n \to X$ $\as{\P}$ for $X_n, X \in L^\Phi$ and $\sup_n \Vert X_n\Vert_\Phi < \infty$ implies that $\rho(X) \leq \liminf_{n \to \infty} \rho(X_n)$.
\end{itemize}
\end{definition}

The strong Fatou property implies the Fatou property but the converse is not true. Note, however, that the two are equivalent if $L^\Phi = L^\infty$.

\begin{remark}
\label{rmk:Fatou property}
The notion of strong Fatou property has been introduced by Gao et al.~\cite{gao:al:19} who noted in \cite{gao:xanthos:18} that for a general normed vector space $L$, the Fatou property for risk measures (which was originally only formulated on $L^\infty$) could either be understood in terms of \emph{order bounded} sequences (giving the Fatou property)  or \emph{norm bounded} sequences (giving the strong Fatou property).
\end{remark}

We proceed to summarise the existing dual representation results for  (finite) coherent risk measures on Orlicz spaces from the literature.\footnote{For Orlicz hearts, the representation result for (finite) coherent risk measures is given in \cite[Corollary 4.2]{cheridito2009risk}.}

\begin{theorem}
	\label{thm:dual:orlicz}
Let $\Phi: [0, \infty) \to [0, \infty]$ be a Young function with conjugate $\Psi$ and $\rho: L^\Phi \to \RR$ a coherent risk measure. Then $\rho$ admits a dual representation under the following conditions:
\begin{enumerate}
\item  $\Phi$ satisfies the $\Delta_2$-condition.
\item $\Psi$ satisfies the $\Delta_2$-condition and $\rho$ satisfies the Fatou property.
\item $\Phi$ is a superlinear Young function and $\rho$ satisfies the strong Fatou property.
\end{enumerate}
\end{theorem}

\begin{proof}
(a) In this case $L^\Phi = H^\Phi$ and the result follows from \cite[Corollary 4.2]{cheridito2009risk}.

(b) This follows from \cite[Theorem 2.5]{gao:al:19} or \cite[Proposition 2.5]{delbaen:owari:19} and Fenchel-Moreau duality.\footnote{Note that for ``(4) $\Rightarrow$ (1)'' in \cite[Theorem 3.7]{gao:al:19}, the assumption of an atomless probability space is not needed.}

(c) This follows from \cite[Theorem 3.2]{delbaen2002coherent} in the case that $L^\Phi = L^\infty$ and from \cite[Theorem~2.4]{gao:xanthos:18} in the general case.
\end{proof}

\begin{remark}
(a) If a coherent risk measure $\rho: L^\Phi \to \RR$ admits a dual representation, it is straightforward to check that it satisfies the Fatou property. The converse is false if both $\Phi$ and $\Psi$ fail to satisfy the $\Delta_2$-condition; see \cite[Theorem 4.2]{gao:al:19} for a generic counterexample.\footnote{Note, however, that $\rho$ in \cite[Theorem 4.2]{gao:al:19} is  $(-\infty, \infty]$-valued.}

(b) A coherent risk measure that admits a dual characterisation does not need to satisfy the strong Fatou property; in fact if $\Phi$ is a superlinear Young function and $\rho: L^\Phi \to \RR$  admits a dual characterisation such that $\cQ_\rho \not\subset H^\Psi$, then $\rho$ fails to  satisfy the strong Fatou property by \cite[Theorem 2.4]{gao:xanthos:18}.
\end{remark}

Finally, we show that all coherent risk measures on Orlicz spaces that satisfy a dual representation (independent of whether one of the conditions of Theorem \ref{thm:dual:orlicz} is satisfied) have a nice maximal dual set.

\begin{proposition}
	\label{prop:orlicz dual}
	Let $\Phi: [0, \infty) \to [0, \infty]$ be a Young function with conjugate $\Psi$ and $\rho: L^\Phi \to \RR$ a coherent risk measure. If $\rho$ admits a dual representation, then the maximal dual set $\cQ_\rho$ is $L^\Psi$-closed and $L^\Psi$-bounded if $\Phi$ is finite. If $\Phi$ satisfies the $\Delta_2$-condition, $\cQ_\rho$ is also $L^1$-closed.
	Moreover, if $\cQ \subset \cQ_\rho$ has $L^\Psi$-closure $\cQ_\rho$, then $\cQ$ represents $\rho$, and if $\cQ \subset \cQ_\rho$ represents $\rho$, then $\cQ_\rho \subset \ol \cQ$, where $\ol \cQ$ denotes the $L^1$-closure.
\end{proposition}

\begin{proof}
If $\Phi$ jumps to $\infty$, i.e., $L^\Phi = L^\infty$, then the result follows from \cite[Theorem 3.2]{delbaen2002coherent}. So assume for the rest of the proof that $\Phi$ is finite. Denote by $\rho_H$ the restriction of $\rho$ to $H^\Phi$. Then $\cA_{\rho_H} \subset \cA_{\rho}$ and hence $\cQ_{\rho_H} \supset \cQ_\rho$. It follows from \cite[Corollary 4.2]{cheridito2009risk} and Proposition \ref{prop:app:Orlicz}(a) that  $\cQ_{\rho_H}$ is $L^\Psi$-bounded and $L^1$-closed. Hence, $\cQ_\rho$ is $L^\Psi$-bounded. This together with the definition of $\cQ_\rho$ and the generalised Hölder inequality \eqref{eq:generalised holder} implies that $\cQ_\rho$ is $L^\Psi$-closed. If $\Phi$ satisfies the $\Delta_2$-condition then $\cA_{\rho_H} = \cA_{\rho}$ and so $\cQ_\rho = \cQ_{\rho_H}$ is $L^1$-closed. Moreover, if $\cQ \subset \cQ_\rho$ has $L^\Psi$-closure $\cQ_\rho$, then $\cQ$ represents $\rho$ by the generalised Hölder inequality \eqref{eq:generalised holder} and if  $\cQ \subset \cQ_\rho$ represents $\rho$, then  $\ol \cQ_\rho = \ol \cQ$ because otherwise by the Hahn-Banach separation theorem (for the pairing $(L^\infty, L^1)$), there exists $X \in L^\infty$ such that $\sup_{Z \in \ol \cQ_\rho } \mathbb{E}[- Z X] \neq \sup_{Z \in \ol \cQ} \mathbb{E}[- Z X]$.
\end{proof}

\section{Additional results}
\label{app:additional proofs}

\begin{proposition}
	\label{prop:cond 1b equiv to cherny}
For $\cQ \subset \cD$, set $L^1(\mathcal{Q}):= \{X \in L^0 : \lim_{a \to \infty} \sup_{Z \in \mathcal{Q}} \mathbb{E}[Z|X|\mathds{1}_{\{{|X|>a}\}}] = 0 \}$.
If $\cQ$ is UI and $X \in L^1$, the following are equivalent:
	\begin{enumerate}
		\item $X \in L^1(\mathcal{Q})$
		\item $X\mathcal{Q}$ is UI.
	\end{enumerate}
\end{proposition}

\begin{proof}
	First assume that $X\mathcal{Q}$ is uniformly integrable. Then, for any $a>0$ and $Z \in \mathcal{Q}$,
	\begin{align*}
	\mathbb{E}[Z|X|\mathds{1}_{\{|X|>a\}}] & =   \mathbb{E}[Z|X|\mathds{1}_{\{|X| > a\}}\mathds{1}_{\{Z \leq 1\}}] + \mathbb{E}[Z|X|\mathds{1}_{\{|X|>a\}}\mathds{1}_{\{Z > 1\}}] \\ & \leq \mathbb{E}[|X|\mathds{1}_{\{|X| > a\}}] + \mathbb{E}[Z|X|\mathds{1}_{\{Z|X|>a\}}].
	\end{align*}
	Taking the supremum over $\mathcal{Q}$ on both sides, letting $a \to \infty$ and using that $X \in L^1$ and $X\cQ$ is UI yields
	\begin{equation*}
	\lim_{a \to \infty} \sup_{Z \in \mathcal{Q}} \mathbb{E}[Z|X|\mathds{1}_{\{|X|>a\}}] \leq  \lim_{a \to \infty} \mathbb{E}[|X|\mathds{1}_{\{|X| > a\}}] + \lim_{a \to \infty} \sup_{Z \in \mathcal{Q}} \mathbb{E}[Z|X|\mathds{1}_{\{Z|X|>a\}}] = 0.
	\end{equation*}
	
Conversely, assume that $X \in L^1(\mathcal{Q})$.  For any  $a, b> 0$ and $Z \in \cQ$,
	\begin{align*}
	\mathbb{E}[Z|X|\mathds{1}_{\{Z|X|>a^2\}}] &\leq  \mathbb{E}[Z|X|\mathds{1}_{\{Z > a\}}]  + \mathbb{E}[Z|X|\mathds{1}_{\{|X|>a\}}] \\
	&\leq \mathbb{E}[Z|X|\mathds{1}_{\{Z > a\}}\mathds{1}_{\{|X| \leq b\}}] +\mathbb{E}[Z|X|\mathds{1}_{\{Z > a\}}\mathds{1}_{\{|X|> b\}}]  + \mathbb{E}[Z|X|\mathds{1}_{\{|X|>a\}}] \\
	&\leq b \mathbb{E}[Z\mathds{1}_{\{|Z| > a\}}] +\mathbb{E}[Z|X|\mathds{1}_{\{|X|> b\}}]  + \mathbb{E}[Z|X|\mathds{1}_{\{|X|>a\}}].
	\end{align*}
Taking the supremum over $\mathcal{Q}$ on both sides, letting $a \to \infty$ and using that $\cQ$ is UI and $X \in L^1(\cQ)$ yields
	\begin{equation*}
\lim_{a \to \infty}	\sup_{Z \in \mathcal{Q}} \mathbb{E}[Z|X|\mathds{1}_{\{Z|X|>a\}}] \leq  \sup_{Z \in \mathcal{Q}} \mathbb{E}[Z|X|\mathds{1}_{\{|X|>b\}}].
	\end{equation*}
	Now, the result follows when letting $b \to \infty$ and using again that $X \in L^1(\cQ)$.
\end{proof}

\begin{proposition}
	\label{prop: cond 1 implies weak cts map}
	Suppose Condition UI is satisfied. 
	\begin{enumerate}
		\item The set $\ol \cQ$ and $R^i \ol \cQ$ for $i \in \{1, \ldots, d\}$ are uniformly integrable.
		\item The $\mathbb{R}^d$-valued map $F: \ol \cQ \to \mathbb{R}^d$ given by $F(Z) = \mathbb{E}[-Z(R-r\mathbf{1})]$ is weakly continuous.
	\end{enumerate}
\end{proposition}

\begin{proof}
	(a) Fix $i \in \{1, \ldots, d\}$. The Dunford-Pettis theorem implies that $\ol \cQ$ and $\overline {R^i \cQ}$ are UI.  It suffices to show that $R^i \ol \cQ \subset \overline {R^i \cQ} $. So let $Z \in \ol \cQ$. Then there exists a sequence $(Z_n)_{n \in \NN} \subset \cQ$ such that $Z_n$ converges to $Z$ in $L^1$ and hence in probability. It follows that $R^i Z_n$ converges to $R^i Z$ in probability and hence also in $L^1$ as $(R^i Z_n)_{n \in \NN} \subset R^i \cQ$ is UI. It follows that $R^i Z \in \overline{R^i \cQ}$.
	
	(b) Since $F(\lambda Z^1 + (1 -\lambda) Z^2) = \lambda F( Z^1) + (1 -\lambda) F(Z^2)$ for $Z^1, Z^2 \in \ol \cQ$ and $\lambda \in [0, 1]$, preimages under $F$ of convex sets are convex. Since $\ol \cQ$ is convex as $\cQ$ is convex, it therefore suffices to show that $F$ is strongly continuous. So let $(Z_n)_{n \in \NN} \subset \ol \cQ$ be a sequence that converges to $Z$ in $L^1$ and hence in probability. Then $-Z_n(R^i-r)$ converges to $-Z(R^i-r)$ in probability and hence also in $L^1$ by part (a) for each $i \in \{1, \ldots, d\}$.
\end{proof}

\begin{lemma}
	\label{lemma: Z tilde and Z}
	Assume $\tilde{\mathcal{Q}} \subset\cD$ satisfies Conditions POS and INT.  Let $\tilde{Z} \in \tilde{\mathcal{Q}}$ and $X \in L^1$ be a non-constant random variable.  If $\mathbb{E}[-\tilde{Z}X] = 0$, then there exists $Z \in \mathcal{Q}$ such that $\mathbb{E}[-ZX] > 0$.
\end{lemma}

\begin{proof}
Note that $\tilde{Z} > 0$ $\as{\P}$ by Condition POS.	Define $\tilde{\mathbb{Q}} \approx \P$ by $\frac{\diff \tilde{\mathbb{Q}}}{\diff \P} := \tilde{Z}$ and $A:=\{X<0\}$.  Since $X$ is non-constant and $\mathbb{E}^{\tilde{\mathbb{Q}}}[X] = 0$, it follows that $\tilde{\mathbb{Q}}[A] \in (0,1)$.  Seeking a contradiction, suppose that $\mathbb{E}[-ZX] \leq 0$ for all $Z \in \mathcal{Q}$. Let $\cE$ be an $L^\infty$-dense subset of $\cD \cap L^\infty$ corresponding to $\tilde Z$ in Condition INT. Let $Z^\prime \in \cE$. Then there exists $\lambda>0$ such that $\lambda Z^\prime + (1-\lambda) \tilde{Z} \in  \mathcal{Q}$. Thus,
Since $Z^\prime$ was chosen arbitrarily, we may deduce that
	\begin{equation*}
	\sup_{Z \in \cE}(\mathbb{E}[-ZX]) \leq 0,
	\end{equation*}
	which together with Proposition \ref{prop:app:WC dual} below implies that $X \geq 0 \ \mathbb{P}$-a.s. Since $\tilde{\mathbb{Q}} \approx \P$, it follows that $\tilde{\mathbb{Q}}[A] = 0$ and we arrive at a contradiction.
\end{proof}

\begin{proposition}
	\label{prop:app:Fatou}
	Let $\Phi: [0, \infty) \to [0, \infty]$ be a Young function. Let $(Y_n)_{n \in \NN}$ be a sequence in $L^\Phi$ that converges in probability to some random variable $Y$. Then
	\begin{equation*}
	\Vert Y \Vert_\Phi \leq \liminf_{n \to \infty} \Vert Y_n \Vert_\Phi.
	\end{equation*}
\end{proposition}

\begin{proof}
Set $K := \liminf_{n \to \infty} \Vert Y_n \Vert_\Phi$. We may assume without loss of generality that $K < \infty$.  After passing to a subsequence, we may assume without loss of generality that $(Y_n)_{n \in \NN}$ converges to $Y$ $\as{\P}$ If $\Phi$ jumps to infinity, then $\Vert \cdot \Vert_\Phi$ is equivalent to $\Vert \cdot \Vert_\infty$ and the result follows. So assume that $\Phi$ is finite and hence continuous.  For any $\epsilon > 0$, we can pass to a further subsequence and assume without loss of generality that $\Vert Y_n \Vert_\Phi \leq K+\epsilon$ for all $n$.  Then by the definition of the Luxemburg norm, $\mathbb{E}[\Phi(|Y_n/(K+\epsilon)|)] \leq 1$ for all $n$.  Fatou's lemma gives 
\begin{equation*}
\mathbb{E}\left[\Phi\left(\left|\tfrac{Y}{K+\epsilon}\right|\right)\right] \leq \liminf_{n \to \infty} \mathbb{E}\left[\Phi\left(\left|\tfrac{Y_{n}}{K+\epsilon}\right|\right)\right] \leq 1.
\end{equation*}
This implies that $\Vert Y \Vert_\Phi \leq K+\epsilon$.  By letting $\epsilon \to 0$, we conclude that $\Vert Y \Vert_\Phi \leq K$.
\end{proof}

\begin{proposition}\label{prop:app:Orlicz}
Let $\Phi: [0, \infty) \to [0, \infty)$ be a finite Young function with conjugate $\Psi$. Let $\rho: H^\Phi \to \RR$ be a coherent risk measure. Denote by $\cQ_{\rho}$ the maximal dual set. Then
\begin{enumerate}
\item $\cQ_\rho$ is $L^1$-closed and $L^\Psi$-bounded. 
\item  If $R \in H^\Phi$, then $R\cQ_\rho$ is uniformly integrable.
\end{enumerate}
\end{proposition}

\begin{proof}
(a) It follows from \cite[Corollary 4.2]{cheridito2009risk} that  $\cQ_\rho \cap L^\Psi$ is $L^\Psi$ bounded and represents $\rho$. It suffices to show that  $\cQ_\rho \cap L^\Psi$ is $L^1$-closed. Indeed, this implies that  $\cQ_\rho \subset L^\Psi$ because otherwise,  by the Hahn-Banach separation theorem (for the pairing $(L^1, L^\infty)$), there exists $X \in L^\infty$ such that $\sup_{Z \in \cQ_\rho \cap L^\Psi} \mathbb{E}[- Z X] < \sup_{Z \in \cQ_\rho} \mathbb{E}[- Z X]$, in contradiction to the fact that both $\cQ_\rho \cap L^\Psi$ and $\cQ_\rho $ represent $\rho$ on $L^\infty$. 

Set $K := \sup_{Z \in \mathcal{Q_\rho}} \lVert Z \rVert_{\Psi} < \infty$. Let $(Z_{n})_{n \geq 1}$ be a sequence in $\cQ_\rho \cap L^\Psi$ that converges to $Z \in L^1$. Then $Z \in \cD$ and $\Vert Z \Vert_\Psi \leq K$ by Proposition \ref{prop:app:Fatou}. Let $X \in \cA_{\rho} \subset H^\Phi$. We have to show that $\mathbb{E}[Z  X] \geq 0$. Since $\mathbb{E}[Z_n X] \geq 0$ by the fact that $Z_n \in \cQ_\rho \cap L^\Psi$, it suffices to show that $\mathbb{E}[ZX] = \lim_{n \to \infty} \mathbb{E}[Z_n X]$. For any $n \in \NN$ and $a_n>0$, the generalised H{\"o}lder inequality (\ref{eq:generalised holder}) yields
	\begin{align}
	|\mathbb{E}[Z_{n}X] - \mathbb{E}[ZX]|  &\leq \mathbb{E}[|X||Z_{n}-Z|] = \mathbb{E}[|X||Z_{n}-Z|\mathds{1}_{\{|X|>a_n\}}] +  \mathbb{E}[|X||Z_{n}-Z|\mathds{1}_{\{|X| \leq a_n\}}] \notag \\ 
	&\leq  \mathbb{E}[|X|Z_{n} \mathds{1}_{\{|X|>a_n\}}] +  \mathbb{E}[|X|Z\mathds{1}_{\{|X| > a_n\}}] + a_n\lVert Z_{n}-Z\rVert_{1} \notag \\
	 &\leq (2K+2K)\lVert X \mathds{1}_{\{|X|>a_n\}} \rVert_{\Phi} + a_n\lVert Z_{n}-Z\rVert_{1}. \label{eq:pf:prop:app:Orlicz}
	\end{align} 
	Now if we choose $a_n := \min(n, \tfrac{1}{\sqrt{\lVert Z_{n}-Z\rVert_{1}}})$ and let $n \to \infty$, the right hand side of \eqref{eq:pf:prop:app:Orlicz} converges to $0$ by
	order continuity of of $H^\Phi$ (see e.g.~\cite[Theorem 2.1.14]{edgar1992stopping}).
	
(b) First, consider the case that $R = 1$. If $\Phi$ is not superlinear, then $\Psi$ jumps to infinity, and hence $\cQ_{\rho}$ is $L^\infty$-bounded by part (a) and therefore UI. If $\Phi$ is superlinear (and finite), then $\Psi$ is superlinear and finite. Set $K:= \sup_{Y \in \mathcal{Q}} \lVert Y\rVert_{\Psi} < \infty$ and define the superlinear function $\tilde \Psi$ by $\tilde{\Psi}(y) := \Psi(y/K)$.  By the definition of the Luxemburg norm,
	\begin{equation*}
	\mathbb{E}[\tilde{\Psi}(Y)] = \mathbb{E}[\Psi(Y/K)] \leq 1, \quad \textnormal{for all} \ Y \in \cQ_\rho.
	\end{equation*}
	This implies $\sup_{Y \in \cQ_\rho} \mathbb{E}[\tilde{\Psi}(Y)] \leq 1 < \infty$.  Since $\tilde{\Psi}$ is superlinear, the de la Vall\'ee-Poussin theorem implies that $\cQ_{\rho}$ is UI.
	
	Next, assume that $R \in H^\Phi$. By Proposition \ref{prop:cond 1b equiv to cherny}, it is enough to show that  $R \in L^1(\cQ_\rho)$ where
	\begin{equation*}
	L^1(\cQ_\rho) := \{X \in L^0 : \lim_{a \to \infty} \sup_{Z \in \cQ_\rho} \mathbb{E}[Z|X|\mathds{1}_{\{{|X|>a}\}}] = 0 \}.
	\end{equation*}
	Since $R \in H^\Phi$, the generalised  H{\"o}lder inequality and order continuity of  $H^\Phi$ give
	\begin{equation*}
	\lim_{a \to \infty} \sup_{Z \in \mathcal{Q}} \mathbb{E}[Z|X|\mathds{1}_{\{{|X|>a}\}}]  \leq   \lim_{a \to \infty} 2\sup_{Z \in \mathcal{Q}} \lVert Z \rVert_{\Psi} \lVert X\mathds{1}_{\{{|X|>a}\}} \rVert_{\Phi} = 0. \qedhere
	\end{equation*}
\end{proof}

\begin{proposition}
	\label{prop:app:WC dual}
Let $\cE$ be an $\sigma(L^\infty, L^1)$-dense subset of $\cD \cap L^\infty$.\footnote{Note that $\cD \cap L^\infty$ is $\sigma(L^\infty, L^1)$-closed.} Then for all $X \in L^1$.
\begin{equation*}
\sup_{Z \in \cE} \E[-Z X] = \textnormal{WC}(X).
\end{equation*}
\end{proposition}

\begin{proof}
Define the coherent risk measure $\rho: L^1 \to (-\infty, \infty]$ by $\rho(X) := \sup_{Z \in \cE} (\mathbb{E}[-ZX])$. To show that $\rho =  \textnormal{WC}$, let $X \in L^1$ and set $c := \esssup(- X) = \textnormal{WC}(X)$.

First, assume that $c < \infty$. Then monotonicity of the expectation gives $\rho(X) \leq \textnormal{WC}(X)$. For the reverse inequality, let $\epsilon > 0$ and set  $Z := \mathds{1}_{\{-X \geq c -\epsilon\}} / \P[-X \geq c -\epsilon] \in \cD \cap L^\infty$. Then $\E[-Z X] \geq c -\epsilon$. Since $\cE$ is $\sigma(L^\infty, L^1)$-dense in $\cD \cap L^\infty$, there exists a net $(Z_i)_{i \in I}$ in  $\cE$ which converges to $Z$ in  $\sigma(L^\infty, L^1)$. Thus,
\begin{align*}
\rho(X) &\geq \lim_{i \in I}\mathbb{E}[-Z_i X] = \mathbb{E}[-Z X] \geq c - \epsilon = \textnormal{WC}(X) - \epsilon.
\end{align*}
Letting $\epsilon \to 0$ yields $\rho(X) \geq \textnormal{WC}(X)$.

Finally, assume that $c = \infty$. Let $N > 0$ be given and set $X_N := \max (X, -N)$. Then $X_N \geq X$ and $\textnormal{WC}(X_N) = N$. By monotonicity of $\rho$ and the first part,
\begin{equation*}
\rho(X) \geq \rho (X_N) = \textnormal{WC}(X_N) = N.
\end{equation*}
Letting $N \to \infty$ yields $\rho(X) = \infty = \textnormal{WC}(X)$. 
\end{proof}

\begin{proposition}
	\label{prop:ES Q tilde}
	Fix $\alpha \in (0,1)$.  Then $ \tilde{\mathcal{Q}}^\alpha:=\{ Z \in \mathcal{D} : Z > 0 \ \mathbb{P}\textnormal{-a.s.~and } \lVert Z \rVert_{\infty} < \tfrac{1}{\alpha}  \}$ is a nonempty subset of $\mathcal{Q}^\alpha$ satisfying Conditions POS, MIX and INT.
\end{proposition}

\begin{proof}
	It is clear that $1 \in \tilde{\mathcal{Q}}^\alpha \subset \mathcal{Q}^\alpha$, and by definition $\tilde{\mathcal{Q}}^\alpha$ satisfies POS.  If $Z \in \mathcal{Q}^\alpha$, $\tilde{Z} \in \tilde{\mathcal{Q}}^\alpha$ and $\lambda \in (0,1)$, then $\lambda Z + (1-\lambda) \tilde{Z} > 0 \ \mathbb{P}$-a.s., and by the triangle inequality
	\begin{equation*}
	\lVert \lambda Z + (1-\lambda) \tilde{Z} \rVert_\infty \leq \lambda \lVert Z \rVert_\infty + (1-\lambda) \lVert \tilde{Z} \rVert_\infty < \tfrac{1}{\alpha},
	\end{equation*}
 so $\tilde{\mathcal{Q}}^\alpha$ satisfies Condition MIX.  To show Condition INT, let $\tilde{Z} \in \tilde{\mathcal{Q}}^{\alpha}$. Set $\cE := \mathcal{D} \cap L^\infty$ and let $Z \in \cE$. Since $\lVert Z \rVert_\infty < \infty$ and $\lVert \tilde{Z} \rVert_\infty < \tfrac{1}{\alpha}$  there is $\lambda \in (0, 1)$ such that $\lambda \lVert \tilde{Z} \rVert_\infty + (1-\lambda) \lVert Z \rVert_\infty \leq \tfrac{1}{\alpha}$.  By the triangle inequality it follows that $\lambda \tilde{Z} + (1-\lambda)Z \in \cQ^{\alpha}$.
\end{proof}

\begin{proposition}
	\label{prop:app:spectral}
	Assume $\mu$ is a probability measure on $([0,1], \cB_{[0, 1]})$ and $\rho^\mu$ the corresponding spectral risk measure. 
	\begin{enumerate}
		\item $\rho^\mu$ is represented by 
		\begin{align*}
		\mathcal{Q}_{\mu} = \bigg \{\int_{[0,1]}\zeta_{\alpha} \, \mu(\textnormal{d}\alpha) : \ &\zeta_\alpha(\omega) \textnormal{ is jointly measurable and there is $1 >\epsilon > 0$} \\ \textnormal{such that } &\zeta_{\alpha} \in \mathcal{Q}^{\alpha} \textnormal{ for } \alpha \in [0,1-\epsilon] \textnormal{ and } \zeta_\alpha \equiv 1 \textnormal{ for } \alpha \in (1-\epsilon,1]   \bigg\}.
		\end{align*}
		\item If $\mu$ does not have an atom at $1$, the set 
			\begin{align*}
		\tilde \cQ_{\mu} = \bigg \{\int_{[0,1)} \tilde \zeta_{\alpha} \, \mu(\textnormal{d}\alpha)  : \ &\tilde \zeta_\alpha(\omega) \textnormal{ is jointly measurable and there is $\epsilon \in (0,1)$, $\delta \in (0, \tfrac{\epsilon}{1-\epsilon})$} \\  \textnormal{such that } &\tilde \zeta_{\alpha} \in \tilde \cQ^{\alpha(1 + \delta)} \textnormal{ for } \alpha \in [0,1-\epsilon] \textnormal{ and } \tilde{\zeta}_\alpha \equiv 1 \textnormal{ for } \alpha \in (1- \epsilon,1)  \bigg\},
		\end{align*}
		 is nonempty and satisfies Conditions POS, MIX and INT.
		 \end{enumerate}
	\end{proposition}

\begin{proof}
(a) It follows from \cite{cherny2006weighted} that $\rho^\mu$ is represented by 
\begin{equation*}
\mathcal{Q}_{\rho^\mu} = \bigg \{\int_{[0,1]}\zeta_{\alpha} \, \mu(\textnormal{d}\alpha) : \zeta_\alpha(\omega) \textnormal{ is jointly measurable and } \zeta_{\alpha} \in \mathcal{Q}^{\alpha} \textnormal{ for all } \alpha \in [0,1]   \bigg\}.
\end{equation*}
Let $Z =  \int_{[0,1]}\zeta_{\alpha} \, \mu(\textnormal{d}\alpha)\in \mathcal{Q}_{\rho^\mu}$. Set $Z_n := \int_{[0,1-1/n]}\zeta_{\alpha} \, \mu(\textnormal{d}\alpha) + \mu((1-1/n, 1]) \in \cQ_\mu$. Then $$\lim_{n \to \infty}\Vert Z_n - Z \Vert_\infty \leq \lim_{n \to \infty}\frac{n}{n -1} \mu((1-1/n, 1))= 0.$$ This implies that $\mathbb{E}[- Z X] = \lim_{n \to \infty} \mathbb{E}[-Z_n X]$ for all $X \in L^1$.

(b) Since $1 \in \tilde{\cQ}^\beta$ for all $\beta \in [0, 1)$ and $\tilde{\cQ}^\beta$ only contains positive random variables, it follows that $1 \in \tilde \cQ_{\mu}$ and Condition POS is satisfied.
	
To show Condition MIX, let $Z \in \mathcal{Q}_\mu$, $\tilde Z \in \tilde \cQ_{\mu}$ and $\lambda \in (0, 1)$.  Then there is $\epsilon \in (0,1)$ and $\delta \in (0,\tfrac{\epsilon}{1-\epsilon})$ such that $Z = \int_{[0, 1- \epsilon] }\zeta_{\alpha}\,\mu(\textnormal{d}\alpha) + \mu((1-\epsilon, 1))$ and $\tilde Z = \int_{[0, 1-\epsilon] }\tilde \zeta_{\alpha}\,\mu(\textnormal{d}\alpha) + \mu((1-\epsilon, 1))$, where $\zeta_{\alpha} \in \cQ^\alpha$ and $\tilde \zeta_{\alpha} \in \tilde \cQ^{\alpha(1 + \delta)}$ for $\alpha \in [0,1-\epsilon]$.  Set $\delta' := \tfrac{ \delta(1 -\lambda)}{1 + \delta \lambda} \in (0, \delta)$. A simple calculation shows that $\lambda \zeta_{\alpha} + (1 - \lambda)\tilde \zeta_\alpha \in \tilde \cQ^{\alpha(1 + \delta')}$ for all $\alpha \in [0,1- \epsilon]$. Thus,
\begin{equation*}
\lambda Z + (1 -\lambda) \tilde Z = \int_{[0, 1- \epsilon] } \lambda \zeta_{\alpha} + (1 -\lambda) \tilde \zeta_{\alpha}\,\mu(\textnormal{d}\alpha) + \mu((1-\epsilon, 1)) \in \tilde \cQ_\mu.
\end{equation*}
	
	Finally, to show Condition INT, let $\tilde Z \in \tilde{\cQ}_\mu$ and set
	\begin{align*}
	\cE := \Big\{\int_{[0, 1)} \zeta_\alpha \, \mu(\textnormal{d}\alpha) : \ &\zeta_\alpha(\omega) \textnormal{ is jointly measurable and there is } 1 > \gamma, \epsilon > 0 \text{ such that} \\
& \zeta_\alpha \in \cQ^\gamma \textnormal{ for } \alpha \in [0,1-\epsilon] \textnormal{ and } \zeta_\alpha \equiv 1 \textnormal{ for } \alpha \in (1-\epsilon,1) \Big\}.
	\end{align*}
It is straightforward to check that $\cE$ is a dense subset of $\cD \cap L^\infty$. Let $Z \in \cE$.  Then there exists $\epsilon, \gamma \in (0,1)$ and $\delta \in (0,\tfrac{\epsilon}{1-\epsilon})$ such that $\tilde Z = \int_{[0,1- \epsilon] }\tilde \zeta_{\alpha}\,\mu(\textnormal{d}\alpha) + \mu((1-\epsilon, 1))$ and $Z = \int_{[0,1- \epsilon] } \zeta_{\alpha}\,\mu(\textnormal{d}\alpha) + \mu((1-\epsilon, 1))$, where  $\tilde \zeta_{\alpha} \in \tilde \cQ^{\alpha(1 + \delta)}$ and $\zeta_{\alpha} \in \cQ^\gamma$ for $\alpha \in [0,1-\epsilon]$. Set $\lambda' := \tfrac{\delta \gamma}{(2 + \delta)(1 + \delta - \gamma)} \in (0, 1)$.
A simple calculation shows that $\lambda' \zeta_{\alpha} + (1 - \lambda')\tilde \zeta_\alpha \in \tilde \cQ^{\alpha(1+\delta/2)}$ for all $\alpha \in [0,1- \epsilon]$. Thus,
\begin{equation*}
\lambda' Z + (1 -\lambda') \tilde Z = \int_{[0,1- \epsilon] } \lambda' \zeta_{\alpha} + (1 -\lambda') \tilde \zeta_{\alpha}\,\mu(\textnormal{d}\alpha) + \mu((1-\epsilon, 1)) \in   \tilde \cQ_\mu \subset \cQ_{\mu}. \qedhere
\end{equation*}
\end{proof}

\begin{proposition}
	\label{prop:Interior of Q g beta}
	Let $g:[0,\infty) \to \mathbb{R}$ be a convex function and $\beta > g(1)$.  Let $\mathcal{Q}^{g,\beta} := \{ Z \in \mathcal{D} : \mathbb{E}[g(Z)] \leq \beta \}$. Then
	$\tilde{\mathcal{Q}}^{g,\beta}  := \{ Z \in \mathcal{D} : Z>0 \ \mathbb{P}\textnormal{-a.s.\ and } \mathbb{E}[g(Z)] < \beta \} $ is nonempty and satisfies Conditions POS, MIX and INT.
\end{proposition}

\begin{proof}
It is clear that $\tilde{\mathcal{Q}}^{g,\beta}$ satisfies Condition POS and $1 \in \tilde{\mathcal{Q}}^{g,\beta}$ since $\beta > g(1)$. To show condition MIX,
	let $Z \in \mathcal{Q}^{g,\beta}$, $\tilde{Z} \in \tilde{\mathcal{Q}}^{g,\beta}$ and $\lambda \in (0,1)$.  By the convexity of $g$,
	\begin{equation*}
	\mathbb{E}[g(\lambda \tilde{Z} + (1-\lambda)Z)] \leq \mathbb{E}[\lambda g(\tilde{Z}) + (1-\lambda)g(Z)] = \lambda \mathbb{E}[g(\tilde{Z})] + (1-\lambda)\mathbb{E}[g(Z)] < \beta.
	\end{equation*}
To show Condition INT, let $\tilde{Z} \in \tilde{\mathcal{Q}}^{g,\beta}$. Set $\cE := \mathcal{D} \cap L^\infty$ and let $Z \in \cE$. Since $\E[g(Z)] < \infty$ and $\E[g(\tilde Z)] < \beta$  there is $\lambda \in (0, 1)$ such that $\lambda \mathbb{E}[g(\tilde{Z})] + (1-\lambda)\mathbb{E}[g(Z)] \leq \beta$. Now convexity of $g$ implies that $\lambda \tilde{Z} + (1-\lambda)Z \in {\mathcal{Q}}^{g,\beta}$.
\end{proof}

\small
\bibliography{rhoarb} 
\bibliographystyle{amsplain}

\end{document}